\documentclass[usenames, dvipsnames]{lmcs}
\pdfoutput=1

\usepackage[utf8]{inputenc}

\usepackage{lastpage}
\lmcsdoi{19}{1}{3}
\lmcsheading{}{\pageref{LastPage}}{}{}%
{Nov.~24,~2021}{Jan.~13,~2023}{}

\usepackage{etex}
\usepackage{graphicx}
\usepackage{lineno}
\usepackage [utf8] {inputenc}
\usepackage[T1]{fontenc}  
\usepackage{xcolor}
\usepackage{prooftree}
\usepackage{mathpartir}
\usepackage{amsmath}
\usepackage{bbm}
\usepackage{stmaryrd}
\usepackage{xspace}
\usepackage{graphicx}
\usepackage{mathtools,amssymb}
\usepackage{url}
\usepackage{relsize} 
\usepackage[curve,matrix,arrow,cmtip]{xy}
\usepackage{tikz}
\usepackage{mathcomp}

\newcommand {\Infer} [5] [] {
  \inferrule*[%
    left={\textsc{#2}},%
    right={$\begin{array}{l} {#5} \end{array}$}, 
    vcenter,%
    #1
  ]%
  {#3}{#4}}
\newcommand {\IInfer} [5] [] {
  \inferrule*[%
    fraction={===}, 
    left={\textsc{#2}},%
    right={$\begin{array}{l} #5 \end{array}$}, 
    vcenter,%
    #1
  ]%
  {#3}{#4}}
\newcommand{\NamedRule}[5][]{ \Infer[#1]{#2}{ #3 }{#4}{#5} }

\newcommand{\NamedCoRule}[5][]{\IInfer[#1]{#2}{ #3 }{#4}{#5}} 

\newcommand{\Nset}{\mathcal{N}}

\newcommand{\gtNet}[2]{{\sf gt}(#1,#2)}

\newcommand{\refToDef}[1]{Definition~\ref{#1}}
\newcommand{\refToFigure}[1]{Figure~\ref{#1}}
\newcommand{\refToSection}[1]{Section~\ref{#1}}
\newcommand{\refToExample}[1]{Example~\ref{#1}}
\newcommand{\refToLemma}[1]{Lemma~\ref{#1}}
\newcommand{\refToTheorem}[1]{Theorem~\ref{#1}}

\newcommand{\refToCor}[1]{Corollary~\ref{#1}}

\newcommand{\refToProp}[1]{Proposition~\ref{#1}}
\newcommand{\refToPropItem}[2]{\refToProp{prop:#1}(\ref{prop:#1:#2})} 

\newcommand{\rulename}[1]{{[\textsc{#1}]}}
\newcommand{\coDefGr}{::=_\rho}

\newcommand{\set}[1]{\{#1\}}
\newcommand{\kf}[1]{\ensuremath{\mathsf{#1}\xspace}}

\newcommand{\pp}{{\sf p}}
\newcommand{\q}{{\sf q}}
\newcommand{\pr}{{\sf r}}
\newcommand{\ps}{{\sf s}}

\newcommand{\PP}{\ensuremath{P}}
 
\newcommand{\la}{\M}
\newcommand{\M}{\lambda}
\newcommand{\sendL}[2]{#1!#2}
\newcommand{\rcvL}[2]{#1?#2}
\newcommand{\Participants}{\ensuremath{{\sf Part}}}
\newcommand{\Labels}{\ensuremath{{\sf Lab}}}

\newcommand{\lpart}[1]{\ensuremath{{\sf part}(#1)}}

\newcommand{\inpg}[4]{#1?\{\Seq{#3}{#4}\}_{#2}}
\newcommand{\inp}[5]{\inpg{#1}{#2\in#3}{#4_{#2}}{#5_{#2}}}

\newcommand{\outpone}[3]{#1!\Seq{#2}{#3}}

\newcommand{\outg}[4]{#1!\{\Seq{#3}{#4}\}_{#2}}
\newcommand{\oup}[5]{\outg{#1}{#2\in#3}{#4_{#2}}{#5_{#2}}}
\newcommand{\btg}[4]{#1\dagger\{\Seq{#3}{#4}\}_{#2}}
\newcommand{\bp}[5]{\btg{#1}{#2\in#3}{#4_{#2}}{#5_{#2}}}

\newcommand{\inact}{\ensuremath{\mathbf{0}}}
\newcommand{\val}{v}
\newcommand{\Q}{\ensuremath{Q}}
\newcommand{\R}{\ensuremath{R}}
\newcommand{\Nt}{\mathbbmss{N}}
\newcommand{\parN}{\mathrel{\|}}
\newcommand{\confAs}[2]{#1\parN#2}
\newcommand{\Msg}{\mathcal{M}} 
\newcommand{\addMsg}[2]{#1\cdot #2}

\newcommand{\commO}[3]{#1\,#3!#2}
\newcommand{\commI}[3]{#1\,#3?#2}
\newcommand{\commIO}[3]{#1\,#3\dagger#2}
\newcommand{\CommAs}[3]{\commO{#1}{#2}{#3}}
\newcommand{\CommAsI}[3]{\commI{#1}{#2}{#3}}

\newcommand{\mq}[3]{\langle#1,#2,#3\rangle}

\newcommand{\Seq}[2]{#1;#2}
\newcommand{\pP}[2] {#1[\![\,#2\,]\!]}
\newcommand{\stackred}[1]{\xrightarrow{#1}}
\newcommand{\lockred}[1]{\xRightarrow{#1}}
\newcommand{\lockredsub}[2]{\xRightarrow{#1}_{\!\!\!\!#2}}
\newcommand{\asCom}{\beta}
\newcommand{\play}[1]{\ensuremath{{\sf play}(#1)}}
\newcommand{\plays}[1]{\ensuremath{{\sf players}(#1)}}
\newcommand{\ms}{{\sf m}}
\newcommand{\lockSC}{\ensuremath{\gamma}}
\newcommand{\co}{\beta}

\newcommand{\wgs}[2]{\ensuremath{{\sf wg}(#2,#1)}}
\newcommand{\inc}[1]{\ensuremath{{\sf inp}(#1)}}

\newcommand{\G}{\ensuremath{{\sf G}}}

\newcommand{\End}{\kf{End}}

\newcommand{\agtbg}[5]{#1\,#2\dagger\{\Seq{#4}{#5}\}_{#3}}
\newcommand{\agtb}[6]{\agtbg{#1}{#2}{#3\in #4}{#5_#3}{#6_#3}}
\newcommand{\agtOg}[5]{#1\,#2!\{\Seq{#4}{#5}\}_{#3}}
\newcommand{\agtoneO}[4]{#1\,#2!\Seq{#3}{#4}}
\newcommand{\agtO}[6]{\agtOg{#1}{#2}{#3\in #4}{#5_#3}{#6_#3}}

\newcommand{\agtSOS}[3]{#1\,#2!#3}

\newcommand{\agtIg}[5]{#1\,#2?\{\Seq{#4}{#5}\}_{#3}} 
\newcommand{\agtII}[6]{\agtIg{#1}{#2}{#3\in #4}{#5_{#3}}{#6_{#3}}}  
\newcommand{\agtI}[4]{#1\,#2?\{\Seq{#3_i}{#4_i}\}_{i\in I}}
\newcommand{\agtJ}[4]{#1\,#2?\{\Seq{#3_j}{#4_j}\}_{j\in J}}
\newcommand{\agtIS}[3]{#1\,#2?#3}

\newcommand{\parG}{\mathrel{\|}}

\newcommand{\DD}{\mathcal D}

\newcommand{\tupleOK}[2]{\vdash_{\sf b}#1\parG#2}
\newcommand{\tupleOKW}[2]{\vdash_{\sf wb}#1\parG#2}

\newcommand{\OKA}[2]{\vdash_{\sf b}^{\mathcal I}#1\parG#2}

\newcommand{\concat}[2]{\ensuremath{#1\,{\cdot}\,#2}}

\newcommand{\weight}{\ensuremath{{\sf depth}}}

\newcommand{\IPaths}[1]{{\sf Paths}(#1)}
\newcommand{\ipth}{\xi}

\newcommand{\ple}[1]{\ensuremath{\langle #1 \rangle }} 
\newcommand{\rn}[1]{\rulename{#1}}

\newcommand{\algwf}[3]{#1\OKA{#2}{#3}}
\newcommand{\algcomp}[3]{ #1 \vdash_{\mathsf{agr}} (#2,#3) }

\newcommand{\algread}[2]{ \vdash_{\mathsf{read}} (#1,#2) } 
 
\newcommand{\algreadinf}[3]{ #1 \vdash_{\mathsf{dread}} (#2,#3) }
\newcommand{\algreadinfNot}[3]{ #1 \not\vdash_{\mathsf{dread}} (#2,#3) }
\newcommand{\Gset}{\mathcal{G}} 
\newcommand{\Bset}{\mathcal{H}}
\newcommand{\Hset}{\mathcal{H}}

\newcommand{\refToLem}[1]{Lemma~\ref{lem:#1}} 
\newcommand{\refToThm}[1]{Theorem~\ref{thm:#1}} 
\newcommand{\refToLemItem}[2]{\refToLem{#1}(\ref{lem:#1:#2})} 
\newcommand{\refToFig}[1]{Figure~\ref{fig:#1}}

\newcommand{\OK}[3]{\vdash_{\sf ok}(#1,#2,#3)}

\newcommand{\AS}{\mathcal{A}}
\newcommand{\GN}{\mathcal{GN}}

\newcommand{\pairA}[2]{(#1,#2)}

\newcommand{\cl}{{\sf p}}
\newcommand{\s}{{\sf s}}
\newcommand{\hq}{{\it nd}}
\newcommand{\lql}{{\it pr}}
\newcommand{\ok}{{\it ok}}
\newcommand{\ko}{{\it ko}}

\newcommand{\nD}{{\it nd}}

\newcommand{\pR}{{\it pr}}
\newcommand{\Srv}{\ensuremath{S}}

\newcommand{\X}{X}

\newcommand{\tyn}[2]{#1\vdash #2} 
\newcommand{\tynI}[2]{#1\vdash_{\sf i} #2} 
\newcommand{\tynIP}[2]{#1\vdash_{\sf i} #2}

\newcommand{\Y}{Y}

\newcommand{\agteq}[2]{#1 \eqcirc #2} 

\newcommand{\fun}[3]{\ensuremath{#1 \colon #2\rightarrow #3}\xspace} 
\newcommand{\N}{\mathbb{N}}

\newcommand{\itr}[1]{I$'$-#1} 
\newcommand{\itrp}[1]{I$'$-#1}
\newcommand{\ib}[1]{IB-#1}
\newcommand{\agr}[1]{A-#1}

\newcommand{\ub}{\eta} 
\newcommand{\ug}{\chi} 
\newcommand{\uG}{\Gamma}

\newcommand{\vars}[1]{\mathsf{vars}(#1)} 
\newcommand{\dom}[1]{\mathsf{dom}(#1)} 
\newcommand{\eqsys}{E} 
\newcommand{\gsol}[2][]{\mathsf{sol}_{#1}(#2)} 
\newcommand{\Goals}{\mathcal{S}} 
\newcommand{\tyalg}[3]{#1 \vdash #2 \Rightarrow #3} 
\newcommand{\infn}[1]{I-#1} 
\newcommand{\Gpat}{\mathbb{G}} 
\newcommand{\pst}[2]{\mathsf{pos}(#1,#2)}
\newcommand{\three}[3]{#1;#2;#3}
\newcommand{\pair}[2]{(#1,#2)}
\newcommand{\ipair}[2]{(#2,#1)}
\newcommand{\sbtplus}{+} 
\newcommand{\sbtord}{\preceq} 
\newcommand{\II}{\mathcal I}

\begin{document}

\title{Deconfined Global Types for Asynchronous Sessions}
\thanks{This work was partially funded by the MUR project ``T-LADIES'' (PRIN 2020TL3X8X)}

\author[F.~Dagnino]{Francesco Dagnino\lmcsorcid{0000-0003-3599-3535}}[a]	
\author[P.~Giannini]{Paola Giannini\lmcsorcid{0000-0003-2239-9529}}[b]
\author[M.~Dezani-Ciancaglini]{Mariangiola Dezani-Ciancaglini\lmcsorcid{0000-0002-3341-0941}}[c]
 
\address{DIBRIS, Universit\`{a} di Genova, Italy}	
\email{francesco.dagnino@dibris.unige.it}  

\address{DiSSTE, Universit\`{a} del Piemonte Orientale, Alessandria, Italy} 
\email{paola.giannini@uniupo.it}\thanks{This original research has the financial support of the Universit\`{a}  del Piemonte Orientale.}

\address{Dipartimento di Informatica, Universit\`{a} di Torino, Italy} 
\email{dezani@di.unito.it}

\begin{abstract}
Multiparty sessions with asynchronous communications and global types play an important role for the modelling of interaction protocols in distributed systems.  In designing such calculi the aim is to enforce, by typing,  good properties for all participants, maximising, at the same time, the accepted behaviours.   Our type system improves the state-of-the-art by typing all asynchronous sessions and preserving the key properties of Subject Reduction, Session Fidelity and Progress when some well-formedness conditions are satisfied. 
The type system comes together with a sound and complete type inference algorithm. 
The well-formedness conditions are undecidable, but an algorithm checking an expressive restriction of them recovers the effectiveness of typing. 

\end{abstract}

\maketitle


\section{Introduction}\label{intro} 

{\em Multiparty sessions}~\cite{CHY08,CHY16} are at the core of  communication-based programming,   since they formalise 
 message exchange protocols.  
A key choice in the modelling is synchronous versus asynchronous communications, giving rise to synchronous and asynchronous multiparty sessions. 
In the  multiparty session approach {\em global types} play the fundamental role of  describing the whole scenario, while the behaviour of participants is implemented by processes. 
A natural question  is   when a set of processes agrees with a global type,  meaning that participants behave  according  to the protocol described by the global type.  
The straightforward answer is the design of type assignment systems relating processes and global types.    
 Typically, global types are {\em projected} onto participants to get the local behaviours prescribed by the  protocol and then the  processes implementing the participants are checked against such local behaviours.   
In conceiving such systems one  wants to permit  all possible typings which guarantee desirable properties: the mandatory Subject Reduction, but also Session Fidelity and Progress. {\em Session Fidelity}~\cite{CHY08,CHY16} means that the content and the order of exchanged messages respect the prescriptions of  the global type. {\em Progress}~\cite{DY11,CDPY16} requires that all participants willing to communicate  will be able to  do it and, in case of asynchronous communication, also that all sent messages (which usually are in a queue)  will  be received.

A standard way of getting more permissive typings is through {\em subtyping}~\cite{GH05},  which is used to compare local behaviours obtained by projection to the actual behaviours of participants.  
Following the {\em substitution principle}~\cite{LW94},
  we can safely put a process of some type where a process of a bigger type is expected.   
 In the natural subtyping for synchronous multiparty sessions, bigger types have  less inputs and more outputs~\cite{DH12}.    
This subtyping is not only correct, but also {\em complete}, that is, 
any extension of this subtyping would be unsound~\cite{GJPSY19}. 
A powerful subtyping for asynchronous sessions was proposed in~\cite{HMY09} and recently proved to be complete~\cite{GPPSY21}. The key idea of this subtyping is the possibility of  anticipating outputs  before inputs   
to improve efficiency. 
 This additional flexibility is justified by the fact that in an asynchronous setting outputs are \emph{non-blocking} operations, hence they never prevent subsequent actions to be performed.  
 An important issue  
of this subtyping is its undecidability~\cite{BCZ17,LY17},  
 that makes the whole type assignment system undecidable.  
To overcome this problem, some decidable restrictions of this subtyping were proposed~\cite{BCZ17,LY17,BCZ18} 
and a sound, but  necessarily  not complete, decision algorithm  is  presented in~\cite{BCLYZ21}.

Asynchronous communications better represent the exchange of messages between participants in different localities, and  they  are more suitable for implementations. So it is interesting to find   alternatives to subtyping which increase       
typability of asynchronous multiparty sessions,  and still  ensure  all desired properties.  
Recently a more permissive design of global types has been proposed~\cite{CDG20}. 
 The key idea is to make the syntax of global types much closer to   processes.  
This allows us to simplify the type assignment and to recover its decidability.  
Then, we  study well-formedness conditions on global types ensuring the desired properties  on processes.  
In other words, instead of directly deriving such properties from the syntax of global types through the type assignment, 
we split the problem in two steps:  we  first assign global types to networks and then transfer properties from types to processes through  the  typing relation. 
In this way, potential undecidability issues are confined to the second step, that is, they only depend on the complexity of the property one tries to enforce. 
 In this way,  the type assignment remains decidable and provides an abstraction from a local to a global perspective, which  simplifies reasoning on global properties; 
however, by itself it does not ensure any of such properties.  

More in detail, the formalism  proposed in~\cite{CDG20} 
is based on the simple idea of  splitting   outputs and inputs in the syntax of global types,  rather than modelling each communication action as a whole.  
In this way outputs can anticipate inputs,  thus capturing their non-blocking nature directly at the level of global types.  
 We dub ``deconfined'' such global types. 
 The freedom gained by this definition is rather limited  
in~\cite{CDG20}, whose main focus was to define an event structure semantics for asynchronous  multiparty sessions.  
In particular,  the well-formedness conditions that global types had to satisfy  in~\cite{CDG20}
still strongly confined their use. 

In the present paper (which is the journal version of~\cite{DGD21}) we extend the syntax of global types in~\cite{CDG20,DGD21},  allowing input  choices  
for global types, and  significantly  enlarging  the set of well-formed global types. 
In this way we are able to type also an example requiring a subtyping 
which fails for the algorithm in~\cite{BCLYZ21}.  
The idea is that the well-formedness of global types must guarantee that all participants waiting for a message are not stuck  and that all sent messages find the corresponding readers. 
This last condition is particularly 
delicate for  non-terminating  computations in which the number of unread messages may  grow indefinitely.  
 Under this condition the type system  enjoys Subject Reduction, Session Fidelity and Progress.
 Not surprisingly, the well-formedness of global types turns out to be undecidable, hence 
we design a decidable restriction to keep our system effective. 
The proposed algorithm extends similar ones  presented in~\cite{CDG20,DGD21}, adapting them to the more expressive syntax of our global types. 
In particular, we gain expressivity by:
\begin{itemize}
\item 
requiring that at  least  one input in a choice of inputs matches an output or a message in the queue;
\item 
allowing an unbound number  of unread messages when all  of them  will be eventually   read. 
\end{itemize}

We illustrate the proposed calculus with a {\em running example} in which the number of unread messages is unbounded.    
We choose this example since typing this session in standard type systems for multiparty sessions requires subtyping.
Indeed, this is the running example of~\cite{BCLYZ21}, where a decidable restriction of 
asynchronous subtyping  is presented.  In addition this example is typable  neither  
 in~\cite{CDG20} nor  in~\cite{DGD21}. 

A hospital server $\ps$ waits   to receive  from a  patient  $\pp$ 
either some data $\nD$ or a request  to send her  a  report $\pR$. 
In our calculus such a process is represented by $\Srv=\pp?\{\nD;\Srv_1,\pR;\Srv_1\}$, where  $\pp?$ means an input from participant $\pp$  and the labels that can be read, in this case $\nD$  and $\pR$, are between curly brackets and followed by the process representing the remaining behaviour.
After the reception of one of the two labels 
the server answers  by sending  either $\ok$ or $\ko$ and then it waits for another request. 
 This is expressed by   $\Srv_1=\pp!\{\ok;\Srv,\ko;\Srv\}$, where $\pp!$ means an output to participant $\pp$ and again the two labels $\ok$ and $\ko$ are put between curly brackets. 
The patient,  to save time, starts by sending some data and then waits
for the response from the server, i.e., $\PP=\ps!\nD;\PP_1$ and $\PP_1=\ps?\{\ok;\PP,\ko;\ps!\pR;\PP\}$. If the patient receives $\ok$ she continues sending next data. 
In case she receives $\ko$ she sends the request for her report and then starts sending next data.
So the multiparty session $\confAs{\pP\pp\PP\parN\pP\ps\Srv}\emptyset$, where $\emptyset$ is the empty queue, can  execute  as follows:
\[
\confAs{\pP{\pp}{\PP} \parN  \pP{\ps}{\Srv}}\emptyset \stackred{\CommAs\pp{\nD}\ps}\confAs{\pP{\pp}{\PP_1}\parN\pP{\ps}{\Srv}}{\mq\pp{\nD}\ps}\stackred{\CommAsI\pp{\nD}\ps}
  \confAs{\pP{\pp}{\PP_1}\parN\pP{\ps}{\Srv_1}}{\emptyset}
\] 
decorating  transition  
arrows with communications and denoting by $\mq\pp{\nD}\ps$ the message   sent  from $\pp$ to $\ps$ with  label $\nD$. 
The interaction may continue  as shown in \refToFig{ths}.  If the server  repeatedly responds  $\ko$ the queue can grow unboundedly. However, each message will be
eventually read by the server. 

\begin{figure}
\begin{math}\begin{array}{lll}
  \confAs{\pP{\pp}{\PP_1}\parN\pP{\ps}{\Srv_1}}{\emptyset}
  &\stackred{\CommAs\ps{\ko}\pp}&
  \confAs{\pP{\pp}{\PP_1}\parN\pP{\ps}{\Srv}}{\mq\ps{\ko}\pp}\\
  &\stackred{\CommAsI\ps{\ko}\pp}&
  \confAs{\pP{\pp}{\ps!\pR;\PP} \parN  \pP{\ps}{\Srv}}{\emptyset}\\
  &\stackred{\CommAs\pp{\pR}\ps}&
  \confAs{\pP{\pp}{\PP} \parN  \pP{\ps}{\Srv}}{\mq\pp{\pR}\ps}\\
   &\stackred{\CommAs\pp{\nD}\ps}&
  \confAs{\pP{\pp}{\PP_1} \parN  \pP{\ps}{\Srv}}{\mq\pp{\pR}\ps\cdot\mq\pp{\nD}\ps}\\
  &\stackred{\CommAsI\pp{\pR}\ps}&
  \confAs{\pP{\pp}{\PP_1} \parN  \pP{\ps}{\Srv_1}}{\mq\pp{\nD}\ps}\\
   &\stackred{\CommAs\ps{\ko}\pp}&
  \confAs{\pP{\pp}{\PP_1} \parN  \pP{\ps}{\Srv}}{\mq\pp{\nD}\ps\cdot\mq\ps{\ko}\pp}\\
  &\stackred{\CommAsI\ps{\ko}\pp}&
  \confAs{\pP{\pp}{\ps!\pR;\PP} \parN  \pP{\ps}{\Srv}}{\mq\pp{\nD}\ps}\\
  &\stackred{\CommAs\pp{\pR}\ps}&
  \confAs{\pP{\pp}{\PP} \parN  \pP{\ps}{\Srv}}{\mq\pp{\nD}\ps\cdot\mq\pp{\pR}\ps}\\
  &\stackred{\CommAs\pp{\nD}\ps}&
  \confAs{\pP{\pp}{\PP_1} \parN  \pP{\ps}{\Srv}}{\mq\pp{\nD}\ps\cdot\mq\pp{\pR}\ps\cdot\mq\pp{\nD}\ps}\\
\end{array}
\end{math}
\caption{A transition of the hospital session.}\label{fig:ths}
\end{figure}

The network  $\pP{\pp}{\PP} \parN  \pP{\ps}{\Srv}$ can be typed by the global type $\G=\Seq{\CommAs{\pp}{\nD}{\ps}}{\G_1}$, where
$\G_1=\agtIS{\pp}{\ps}{\set{\Seq{\nD}{\G_2}, \Seq{\pR}{\G_2}}}$ and 
$\G_2=\agtSOS{\ps}{\pp}{\set{\Seq{\ok}{\Seq{\CommAsI{\ps}{\ok}{\pp}}{\G}}, \Seq{\ko}{\Seq{\CommAsI{\ps}{\ko}{\pp}}{\Seq{\CommAs{\pp}{\pR}{\ps}}{\G}}}}}$. 
The type $\G$ prescribes that $\pp$  put in the queue the label $\nD$ and then $\pp$ and $\ps$ follow the protocol described by $\G_1$. 
The type $\G_1$ asks the server to read either the label $\nD$ or the label $\pR$ and in both cases the interaction follows the communications in $\G_2$.

\subsubsection*{Outline}  Our calculus of multiparty sessions is presented in \refToSection{sec:calculus}, where the Progress property is defined.  \refToSection{sect:type-system} introduces our  type system: we call it ``wild'' since each network can be typed in it.  We define an LTS for global types with queues and we show Session Fidelity. 
 Together with the type system,  we give a sound  and complete  type inference algorithm 
proving that type inference is decidable. 
In \refToSection{sec:tt} we tame global types  to    guarantee  Subject Reduction and Progress. 
Unfortunately the balancing predicate, which 
ensures  these properties, is undecidable, as shown in \refToSection{sect:dec-algo}. 
The effectiveness of our type system is recovered in \refToSection{sect:algo} by an algorithm  that checks  an inductive restriction of this predicate. 
Related and future works are discussed in \refToSection{sec:rfw}.



\section{A Core Calculus for Multiparty Sessions}\label{sec:calculus}
 Since our focus is on typing by means of global types, we only consider one multiparty session instead of many interleaved multiparty sessions. This allows us to depart from the standard syntax of processes with channels~\cite{CHY08,BCDDDY08} in favour of simpler processes with output and input operators and explicit participants as in~\cite{DGJPY15,DS19,GJPSY19}.

We assume the following base sets:   \emph{
 participants}  $\pp,\q,\pr\in\Participants$, and
 \emph{labels} $\la\in\Labels$.

 \begin{defi}[Processes]\label{p} 
  {\em Processes} $\PP$ are defined by:
\[\begin{array}{rcl}
\PP & \coDefGr  & 
\inact
\mid
\oup\pp{i}{I}{\la}{\PP}%
\mid
\inp\pp{i}{I}{\la}{\PP}%
\end{array}
\]
where $I\neq\emptyset$ and $\la_j\not=\la_h\,$ for $j\neq h$.
\end{defi}
The symbol $ \coDefGr$, in the definition above and in  other definitions, 
indicates that the productions should be interpreted \emph{coinductively}.
 That is, they define possibly infinite processes.  
However, we assume such processes to be \emph{regular},  i.e.,  with 
finitely many distinct subprocesses. In this way, we only obtain processes which are solutions of 
 finite sets  of equations, see \cite{Cour83}.  
We choose this formulation as we will use coinduction in some definitions and proofs and, moreover, it allows us to avoid explicitly handling variables, thus simplifying a lot the technical development. 

 A process of shape $\oup\pp{i}{I}{\la}{\PP}$ (\emph{internal choice}) chooses a  label  in the set  $\{\la_i\mid i\in I\}$ to be sent to $\pp$, and then behaves differently depending on the sent  label.  
 A process of shape $\inp\pp{i}{I}{\la}{\PP}$ (\emph{external choice}) waits for receiving one of the  labels  $\{\la_i\mid i\in I\}$ from $\pp$, and then behaves differently depending on the received  label. Note that  the set of indexes in choices is assumed to be non-empty, and the corresponding  labels  to be all different.  
An internal choice which is a singleton is simply written $\outpone{\pp}{\la}{\PP}$, 
 analogously for an external choice.  We omit traling $\inact$ and we use $\bp\pp{i}{I}\la\PP$ to denote either $\oup\pp{i}{I}{\la}{\PP}$  or  $\inp\pp{i}{I}{\la}{\PP}$. 

In a full-fledged calculus,  labels would carry values, 
namely they would be of shape $\la(\val)$.  For simplicity, here we
consider pure labels.

 \emph{Messages} are triples $\mq\pp{\la}\q$ denoting that participant $\pp$ has sent label $\la$ to participant $\q$. 
Sent messages are stored in a queue,  from which  they
are subsequently fetched by the receiver.

Message queues $\Msg$ are defined by:
\[\Msg::=\emptyset \mid
  \addMsg{\mq\pp{\la}\q}{\Msg}\]
The order of messages in the queue is the order in which they will be
read. Since 
order matters only 
between messages with
the same sender and receiver, we  always  consider message queues modulo the  following  structural equivalence:\[\addMsg{\addMsg{\Msg}{\mq\pp{\la}\q}}{\addMsg{\mq\pr{\la'}\ps}{\Msg'}}\equiv
  \addMsg{\addMsg{\Msg}{\mq\pr{\la'}\ps}}{\addMsg{\mq\pp{\la}\q}{\Msg'}}
  ~~\text{if}~~\pp\not=\pr~~\text{or}~~\q\not=\ps 
\]
Note, in particular, that
$\addMsg{\mq\pp{\la}\q}{\mq\q{\la'}\pp} \equiv
\addMsg{\mq\q{\la'}\pp}{\mq\pp{\la}\q}$. These two
equivalent queues represent a situation in which both participants
$\pp$ and $\q$ have sent a label 
to the other one, and neither of
them has read the message. This situation may
happen in a multiparty session with asynchronous communication.

Multiparty sessions are comprised of  networks, i.e. pairs participant/process  
of shape $\pP{\pp}{\PP}$ composed  in parallel,  each with a different participant $\pp$, and a message queue.

\begin{defi}[Networks and multiparty sessions]\label{nd} 
\begin{enumerate}
\item\label{nd1} {\em Networks} are defined by\\ $\Nt  ::= \pP{\pp_1}{\PP_1} \parN \cdots \parN \pP{\pp_n}{\PP_n}$, where $n>0$ and  $\pp_i \neq \pp_j$  for $i\neq j$.
\item\label{nd2} {\em Multiparty sessions} are defined by $\Nt \parallel \Msg$, where $\Nt$ is a network and $\Msg$ is a message queue. 
\end{enumerate}
\end{defi}
In the following we use session as short for multiparty session.

 We assume the standard structural congruence on sessions (denoted $\equiv$), that is, we consider sessions modulo permutation of components and adding/removing components of  the  shape $\pP\pp\inact$.

 If $\PP\neq\inact$ we write $\pP{\pp}{\PP}\in\Nt$ as short for $\Nt\equiv  \pP{\pp}{\PP} \parN\Nt'$ for some $\Nt'$.
This abbreviation is justified by the associativity and commutativity of $\parN$.

We define $\plays\Nt=\set{\pp\mid\pP{\pp}{\PP}\in\Nt}$.

To define the {\em asynchronous operational semantics} of
sessions, we use an LTS whose labels record  the outputs and the inputs. 
To this end,  \emph{communications}  (ranged over by $\beta$) are either  the asynchronous emission of a
label  $\M$ from participant $\pp$ to participant $\q$ (notation $\CommAs{\pp}{\M}{\q}$) or 
the actual reading  by participant $\q$
of the  label  $\M$ sent by participant $\pp$ (notation $\CommAsI{\pp}{\M}{\q}$).  

\begin{figure}[t]
\begin{math}
\begin{array}{c}  
\\[5pt]
\rulename{Send}\quad\confAs{\pP{\pp}{\oup\q{i}{I}{\la}{\PP}}\parN\Nt}{\Msg} \stackred{\CommAs\pp{\la_h}\q}
  \confAs{\pP{\pp}{\PP_h}\parN\Nt}{\addMsg{\Msg}{\mq\pp{\la_h}\q}}\quad  \text{where }
   h \in I
   \\[3pt]
 \rulename{Rcv}\quad\confAs{\pP{\q}{\inp\pp{i}{I}{\la}{\Q}}\parN\Nt}{\addMsg{\mq\pp{\la_h}\q}{\Msg}}\stackred{\CommAsI\pp{\la_h}\q}
 \confAs{\pP{\q}{\Q_h}\parN\Nt}{\Msg}\quad   \text{where }  h \in I 
\end{array}
\end{math}
\caption{LTS for asynchronous sessions.}\label{fig:asynprocLTS}
\end{figure}
The LTS semantics of sessions 
 is  specified by
the two Rules \rulename{Send} and \rulename{Rcv} given in
\refToFigure{fig:asynprocLTS}.  Rule \rulename{Send} allows a
participant $\pp$ with an internal choice (a sender) to send one of
its possible  labels  $\la_h$, by adding  the corresponding message  to the
queue. Symmetrically, Rule \rulename{Rcv} allows a participant $\q$
with an external choice (a receiver) to read the first message 
in the queue sent  to her by a  given  participant $\pp$, if  its label $\la_h$  is one of those she is waiting for.

The semantic property we aim to ensure,  usually called 
\emph{progress}~\cite{DY11,CDPY16,H2016},  
is the conjunction of a safety property, {\em deadlock-freedom}, and two  liveness  
properties: \emph{input lock-freedom} and \emph{orphan message-freedom}. 
Intuitively, a session is deadlock-free
if, in every reachable state of computation, it is either terminated  (i.e. of the shape $\pP\pp\inact\parN\emptyset$)  or
it can move. It is input lock-free if every component
wishing to do an input can eventually do so. 
Finally, it is orphan-message-free if every message stored in
the queue is eventually read.

The following terminology and notational conventions are
standard. 

If $\Nt \parallel \Msg \stackred{\beta_1}\cdots \stackred{\beta_n}
\Nt' \parallel \Msg'$ for some $n\geq 0$ (where by convention
$\Nt' \parallel \Msg' = \Nt \parallel \Msg $ if $n = 0$), then we say that
$\Nt' \parallel \Msg'$ is a \emph{derivative} of $\Nt \parallel \Msg$.
We write $\Nt \parallel \Msg \stackred{\beta}$ if  $\Nt \parallel \Msg \stackred{\beta}\Nt' \parallel \Msg'$ for some $\Nt', \Msg'$.

\begin{defi}[Live, terminated, deadlocked sessions]
\label{def:session-status}
A session $\Nt \parallel \Msg$ is said to be\\
- \emph{live} if  $\Nt \parallel
\Msg \stackred{\beta}$ for some $\beta$;\\
- \emph{terminated}
if  $\Nt\equiv\pP\pp\inact$ 
and $\Msg =\emptyset$;\\
- \emph{deadlocked} if it is neither live 
nor terminated.
\end{defi}
 To formalise progress (Definition~\ref{def:lock-freedom}) we introduce  
another transition relation on sessions, which
describes their lockstep execution: at each step, all components that
are able to move execute exactly one asynchronous output or input. 

 We define the {\em player of a communication} as the sender in case of output and as the  receiver in case of input: \[
 \play{\CommAs{\pp}{\M}{\q}}=\pp\quad\quad
\play{\CommAsI{\pp}{\M}{\q}}=\q 
\]

Let $\Delta$ denote a non empty
set of communications.  
We say that $\Delta$ is \emph{coherent} for a session $\Nt\parN\Msg$ if 
\begin{enumerate}
\item for all $\co_1,\co_2\in\Delta$, $\play{\co_1} = \play{\co_2}$ implies $\co_1 = \co_2$, and 
\item for all $\co \in \Delta$, $\Nt\parN\Msg \stackred{\co}$. 
\end{enumerate}
The \emph{lockstep transition relation} 
$\Nt\parN\Msg \lockred{\Delta} \Nt' \parN\Msg'$ is defined by:
\[\Nt \parallel \Msg \lockred{\Delta} \Nt' \parallel \Msg' \text{ if } 
\begin{array}[t]{l} 
\Delta = \set{\beta_1,\ldots, \beta_n} \text{ is a maximal coherent set for $\Nt\parN\Msg$ and } \\[3pt] 
\Nt \parallel \Msg \stackred{\beta_1}\cdots \stackred{\beta_n} \Nt' \parallel \Msg'
\end{array}
\]
The notion of derivative can be reformulated for lockstep 
computations
as follows.

If $\Nt \parallel \Msg \lockred{\Delta_1}\cdots \lockred{\Delta_n}
\Nt' \parallel \Msg'$ for some $n\geq 0$ (where by convention
$\Nt' \parallel \Msg' = \Nt \parallel \Msg $ if $n = 0$), then we say that
$\Nt' \parallel \Msg'$ is a \emph{lockstep derivative} of
$\Nt \parallel \Msg$. Clearly each lockstep derivative is a
derivative, but   not vice versa. 

 A lockstep computation is an either finite or infinite sequence of lockstep transitions, and it is \emph{complete} if either it is
finite and cannot be extended (because the last session is not live), or it is infinite. 
Let $\lockSC$ range over lockstep
computations.  

Formally, a lockstep computation $\lockSC$ can be denoted
as follows, where $x \in \mathbf{N} \cup \set{\omega}$ is the length of $\lockSC$:
\[\lockSC = \set { \Nt_k \parallel \Msg_k
  \lockredsub{\Delta_k}{k}\Nt_{k+1} \parallel \Msg_{k+1}}_{k< x}
\]
That is, $\lockSC$ is represented as the set of its successive
lockstep transitions, where the  arrow  subscript $k$ is used to indicate that
the transition occurs in the \mbox{$k$-th} step of the computation. This is
needed in order to distinguish equal transitions occurring in
different steps. For instance, in the session $\Nt \parallel  \mq\pp{\la}\q 
$, where
$\Nt =\pP{\pp}{\PP} \parN \pP{\q}{\Q}$ with $\PP = \sendL{\q}\la; \PP$
and $\Q = \rcvL{\pp}{\la}; Q$, all lockstep
transitions with $k \geq 1$ are of the form
\[\Nt \parallel \mq\pp{\la}\q
\lockredsub{\set{\CommAs\pp{\la}\q,
    \CommAsI\pp{\la}\q}}{k}\Nt \parallel \mq\pp{\la}\q
\]  
We  can  now formalise the progress property: 
\begin{defi}[Input-enabling session] 
  A session $\Nt \parallel \Msg$ is \emph{input-enabling} if
  $\,\pP{\pp}{\inp\q{i}{I}{\la}{\PP}} \in \Nt$ implies that, 
  for all complete\[\lockSC = \set {
    {\Nt_k \parallel \Msg_k} \lockredsub{\Delta_k}{k}\Nt_{k+1} \parallel
    \Msg_{k+1}}_{k <  x}\]  with  
    $\Nt_0 \parallel \Msg_0$  
    =
  \mbox{$\Nt \parallel \Msg$}, there exists $h < x$ such that  $\CommAsI{\pp}{\M_j}{\q} \in\Delta_h$ for some $j\in I$. 
\end{defi}

\begin{defi}[Queue-consuming session] A session $\Nt \parallel
  \Msg\,$ is \emph{queue-consuming} if $\Msg\equiv\addMsg{\mq{\pp}{\la}{\q}}{\Msg'}$
  implies
  that,
  for all  complete 
  \[\lockSC = \set {
    \Nt_k \parallel \Msg_k \lockredsub{\Delta_k}{k}\Nt_{k+1} \parallel
    \Msg_{k+1}}_{k < x}\]  with  
    $  \Nt_0 \parallel \Msg_0   
 = \Nt \parallel \Msg$, there exists $h < x$ such that
  $\CommAsI{\pp}{\M}{\q} \in\Delta_h$.
\end{defi}

\begin{defi}[Progress] 
\label{def:lock-freedom} 
A session 
has the {\em progress} property if:
\begin{enumerate}
\item (Deadlock-freedom) None of its lockstep derivatives is deadlocked;
\item (No locked inputs) All its lockstep derivatives are
  input-enabling;
\item (No orphan messages)
All its lockstep derivatives are queue-consuming.
\end{enumerate}
\end{defi}
It is easy to see that deadlock-freedom implies no locked inputs and no orphan messages for finite computations. 
\begin{exa}
\label{ex:growing-queue}
Let 
$\Nt=\pP{\pp}{\PP} \parN \pP{\q}{\Q} \parN \pP{\pr}{\R}$, where $\PP =
\sendL{\q}\la; \PP$,  $\Q = \rcvL{\pp}{\la}; \rcvL{\pr}{\la'}; Q$
 and $\R = \sendL{\q}\la'; \R$.\\
The unique complete lockstep computation of $\confAs\Nt\emptyset$ is the following: 
\[\begin{array}{lll}
\confAs\Nt\emptyset &\lockred{\set{\CommAs\pp{\la}\q, \CommAs\pr{\la'}\q}}&
  \confAs{\Nt}{\addMsg{\mq{\pp}{\la}{\q}}{\mq{\pr}{\la'}{\q}}}\\
&\lockred{\set{\CommAs\pp{\la}\q, \CommAsI\pp{\la}\q, \CommAs\pr{\la'}\q}}&
  \confAs{\pP{\pp}{\PP} \parN \pP{\q}{\rcvL{\pr}{\la'}; Q} \parN
    \pP{\pr}{\R}}{\addMsg{\mq{\pr}{\la'}{\q}}{\addMsg{\mq{\pp}{\la}{\q}}{\mq{\pr}{\la'}{\q}}}}\\
&\lockred{\set{\CommAs\pp{\la}\q, \CommAsI\pr{\la'}\q, \CommAs\pr{\la'}\q}}&
  \confAs{\Nt}{\addMsg{\mq{\pp}{\la}{\q}}{\addMsg{\mq{\pr}{\la'}{\q}}{\addMsg{\mq{\pp}{\la}{\q}}{\mq{\pr}{\la'}{\q}}}}}\\
&\cdots &\cdots 
\end{array}
\]
 It is easy to check that $\confAs\Nt\emptyset$ has the progress property. Indeed, every input communication in $\Q$ is eventually enabled, and, even though the queue grows at each step of the lockstep computation, every message in the queue is eventually read. 
\end{exa}



\section{A Wild Type System}
\label{sect:type-system} 

In this section we first present global types  and the type system  for networks, then we give  an  LTS for global types  in parallel with queues  which allows us to get Session Fidelity. 
Lastly, we  present a sound and complete type inference algorithm. 
 We call \emph{wild} this type system since the freedom in the syntax of global types allows us to type all networks, see \refToTheorem{th:gtNet}. 

\subsection{Global types}
\label{sect:types} 

Our global types can be obtained  from  the standard ones~\cite{CHY08,CHY16} by splitting output and input communications. For this reason we  call  our global types deconfined.

\begin{defi}[Global  types] 
\label{def:GlobalTypesAs}
\emph{Global types} 
$\G$ are defined by the following grammar: 
\[\begin{array}{rcl}
\G & \coDefGr &   \End  \mid \agtO{\pp}{\q}i I{\la}{\G}
              \mid \agtI \pp\q \la \G                  
\end{array}\]
where $I\neq\emptyset$, $\pp\neq \q$ and $\la_j\not=\la_h\,$ for $j\neq h$. 
\end{defi}
\noindent
As for processes, $ \coDefGr$ indicates that global types are  \emph{regular}.   

The global type $\agtO{\pp}{\q}i I{\la}{\G}$  specifies that player $\pp$ sends a label $\la_k$ with $k\in I$ to participant $\q$ and then the interaction described by the global type $\G_k$ takes place. 
The global type $\agtI \pp\q \la \G   $  specifies that player $\q$ receives a label $\la_k$ with $k\in I$ from participant $\pp$  and then the interaction described by the global type $\G_k$ takes place.  
An output choice which is a singleton is simply written $\Seq{\commO{\pp}{\la}{\q}}{\G}$ and similarly for an input choice. 
We omit trailing $\End$ and we use $ \agtb\pp\q{i}{I}\la\G$ to denote either $\agtO{\pp}{\q}i I{\la}{\G}$ or $\agtI \pp\q \la \G   $. 

We define  $\play{ \agtO{\pp}{\q}i I{\la}{\G}}=\pp$ and $\play{ \agtI \pp\q \la \G }=\q$. 
The set of players of a global type, notation 
$\plays{\G}$, is  
the smallest set of participants satisfying the following equations: 
\[\begin{array}{lll}
\plays{\End} &= \emptyset \\
\plays{ \agtO{\pp}{\q}i I{\la}{\G}} &= \set{\pp}\cup\bigcup_{i\in I}\plays{\G_i}\\
\plays{\agtI \pp\q \la \G} &= \set{\q}\cup\bigcup_{i\in I}\plays{\G_i}
\end{array}\]
Notice that the sets of players are always finite thanks to the regularity of global types.

\begin{figure}
\begin{math}
\begin{array}{c}
\NamedCoRule{\rn{End}}{}{ \tyn\End{\pP\pp\inact} }{} \\[3ex] 
\NamedCoRule{\rn{Out}}{
  \tyn{\G_i}{\pP\pp{\PP_i}\parN\Nt}\ \ \plays{\G_i}\setminus\set\pp=\plays\Nt\ \ \forall i \in I
}{ \tyn{\agtO{\pp}{\q}{i}{I}{\la}{\G}}{\pP\pp{\oup{\q}{i}{I}{\la}{\PP}}\parN\Nt}}{} 
\\[3ex] 
\NamedCoRule{\rn{In}}{
  \tyn{\G_i}{\pP\pp{\PP_i}\parN\Nt} \ \ \plays{\G_i}\setminus\set\pp=\plays\Nt\ \ \forall i \in I
}{ \tyn{\agtI{\q}{\pp}{\la}{\G}}{\pP\pp{\inp{\q}{j}{J}{\la}{\PP}}\parN\Nt} }{I\subseteq J} 
\end{array} 
\end{math}
\caption{Typing rules for networks.}\label{fig:cntr} 
\end{figure}

Global types are an abstraction of networks,  
usually described by projecting global types to local types which are assigned to processes.   
The simplicity of our calculus and the flexibility of our global types  allow us to formulate a type system deriving directly global types for networks,  judgments $\tyn\G\Nt$, see \refToFig{cntr}. 
The double line here and in the following indicates that the rules are coinductive. 
Rules \rn{Out} and \rn{In} just add simultaneously outputs and inputs to global types and to corresponding processes inside networks.  
Rule \rn{Out} requires the same outputs in the process and in the global type, while the subtyping for session types~\cite{DH12} allows less outputs in the process than in the global  type. The only consequence of our choice is that  less global types are derivable for a  given network. The gain is a formulation of Session Fidelity ensuring that networks implement corresponding global types, see \refToTheorem{thm:sfA}. Instead
Rule \rn{In} allows more inputs in the process than in the global type, just mimicking the subtyping for session types~\cite{DH12}.\\ 
 The condition $ \plays{\G_i}\setminus\set\pp=\plays\Nt$  for all $i \in I$ ensures that  the players of $\Nt$ and of  $\G$ coincide whenever $\tyn\G\Nt$, as stated in the following lemma, whose proof is straightforward.  
\begin{lem}\label{ep}
If $\tyn\G\Nt$, then $ \plays{\G}=\plays\Nt$.
\end{lem} 
\noindent
This lemma 
forbids for example to derive $\tyn\G{\pP\pp\PP\parN\pP\q\Q}$ with $\G=\Seq{\CommAs\pp{\la}\q}\G$ and $\PP=\Seq{\q!\la}\PP$ and $\Nt=\pP\q\Q$ and $\Q$ arbitrary, 
 since $\q$ is not a player of $\G$. 
Anyway the given type system is wild since it allows us to type  networks with non-matching inputs and outputs, as shown in \refToExample{mp}.  This motivates the well-formedness conditions on global types we will discuss in \refToSection{sec:tt}. 
\begin{exa}\label{mp} 
We can derive $\tyn{\pp\q!\la;\pp\q?\la'}{\pP\pp{\q!\la}\parN\pP\q{\pp?\la'}}$ and this network in parallel with the empty queue evolves to $\pP\q{\pp?\la'}\parN\mq\pp{\la}\q$ which is deadlocked. 
\end{exa}
\begin{exa}\label{ex:vs}
The typing of the running example discussed in the Introduction is the content of \refToFig{exvd}.  
Notice that the premises of the  
application of Rule \rn{In}  with conclusion $\tyn{\G_1}{\pP\cl{\PP_1}\parN\pP\s\Srv}$  are two occurrences of the derivation $\DD$, since the process  
$\Srv$ is 
$\pp?\{\nD;\Srv_1,\pR;\Srv_1\}$. 
\end{exa}

\begin{figure} 
\prooftree
\hspace{-5em}\DD\qquad\qquad\qquad\DD
\Justifies
\prooftree
 \tyn{\G_1}{\pP\cl{\PP_1}\parN\pP\s\Srv}
\Justifies
 \tyn{\G}{\pP\cl\PP\parN\pP\s\Srv}
 \using \rn{Out}
\endprooftree
 \using \rn{In}
\endprooftree

\bigskip

where $\DD$ is the following derivation

\bigskip
\prooftree
\prooftree
 \tyn{\G}{\pP\cl\PP\parN\pP\s\Srv}
\Justifies
 \tyn{\s\cl?\ok;\G}{\pP\cl{\PP_1}\parN\pP\s\Srv}
  \using \rn{In}
\endprooftree\quad
\prooftree
 \tyn{\G}{\pP\cl\PP\parN\pP\s\Srv}
\Justifies
\prooftree
 \tyn{\cl\s!\lql;\G}{\pP\cl{\s!\lql;\PP}\parN\pP\s\Srv}
\Justifies
 \tyn{\s\cl?\ko;\cl\s!\lql;\G}{\pP\cl{\PP_1}\parN\pP\s\Srv}
  \using \rn{In}
\endprooftree
 \using \rn{Out}
\endprooftree
\Justifies
 \tyn{\G_2}{\pP\cl{\PP_1}\parN\pP\s{\Srv_1}}
  \using \rn{Out}
\endprooftree 
\caption{A typing of the hospital network $\pP{\pp}{\PP} \parN  \pP{\ps}{\Srv}$, where $\PP=\ps!\nD;\PP_1$, $\PP_1=\ps?\{\ok;\PP,\ko;\ps!\pR;\PP\}$, $\Srv=\pp?\{\nD;\Srv_1,\pR;\Srv_1\}$, $\Srv_1=\pp!\{\ok;\Srv,\ko;\Srv\}$ with the global type 
$\G=\Seq{\CommAs{\pp}{\nD}{\ps}}{\G_1}$, where
$\G_1=\agtIS{\pp}{\ps}{\set{\Seq{\nD}{\G_2}\, ,\ \Seq{\pR}{\G_2}}}$ and 
$\G_2=\agtSOS{\ps}{\pp}{\set{\Seq{\ok}{\Seq{\CommAsI{\ps}{\ok}{\pp}}{\G}}\, ,\ \Seq{\ko}{\Seq{\CommAsI{\ps}{\ko}{\pp}}{\Seq{\CommAs{\pp}{\pR}{\ps}}{\G}}}}}$. }\label{fig:exvd}
\end{figure} 

 The regularity of processes and global types ensures the decidability of type checking. 

\medskip

We now show  that our type system allows us to type \emph{any} network, thus justifying the name ``wild''. 
We start by defining  the map $\sf gt$ which associates a pair, made of a network and a list of participants, with a global type. 
We denote by $\epsilon$ the empty list, by $\cdot$ the concatenation  of lists and by $\preceq$ the lexicographic order on labels.
 We use $\pi$ to range over lists of participants and write $\lpart\pi$ for the set of participants occurring in $\pi$. 
 Essentially $\gtNet\Nt\pi$ is the global type whose first communication is done by the first participant $\pp$ of $\pi$ having a process different from $\inact$. If $\pp$ does an output, then the global type is an output with the same set of labels.  If $\pp$ does an input, then the global type is an input with only one label. For inputs 
the global type could have as set of labels 
any subset of the set of labels in the process; our choice produces simpler global types, which make examples more readable.

\begin{defi}\label{gtN}
Let $\Nt$ be a network and $\pi$ be a list of participants.   
The global type {\em $\gtNet\Nt\pi$ associated with $\Nt$ and $\pi$} is corecursively defined by the following equations: 
\begin{align*} 
\gtNet{\Nt}{\epsilon} &= \End \\ 
\gtNet{\Nt}{\pp\cdot\pi} &= \begin{cases} 
\gtNet{\Nt'}{\pi} & \text{if } \Nt\equiv\pP\pp\inact\parN\Nt'\\ 
\agtOg\pp\q{i\in I}{\la_i}{\gtNet{\pP\pp{\PP_i}\parN\Nt'}{\pi\cdot\pp}} & \text{if } \Nt\equiv\pP\pp{\oup\q{i}{I}\la\PP}\parN\Nt'\\ 
 \Seq{\CommAsI{\q}{\la_k}{\pp}}{\gtNet{\pP\pp{\PP_k}\parN\Nt'}{\pi\cdot\pp}}  & \text{if } \Nt\equiv\pP\pp{\inp\q{i}{I}\la\PP}\parN\Nt' \\
&\text{and }k\in I\text{ and } \la_k\preceq\la_i ~\forall i\in I 
\end{cases} 
\end{align*} 
\end{defi}

Notice that $\gtNet\Nt\pi$ is well defined as equations are \emph{productive}, i.e., they always eventually unfold a 
constructor.\footnote{This can be easily checked as in the only non-productive clause   (the first alternative in the second equation)   the list of participants decreases.}

\begin{exa}\label{ex:gyN}
Let $\Nt$ be the hospital network, i.e. $\Nt\equiv\pP{\pp}{\PP} \parN  \pP{\ps}{\Srv}$, where $\PP=\ps!\nD;\PP_1$, $\PP_1=\ps?\{\ok;\PP,\ko;\ps!\pR;\PP\}$, $\Srv=\pp?\{\nD;\Srv_1,\pR;\Srv_1\}$, $\Srv_1=\pp!\{\ok;\Srv,\ko;\Srv\}$. Starting from the participant list $\pp\cdot\ps$ and letting $\Nt'=\pP{\pp}{\PP_1} \parN  \pP{\ps}{\Srv}$ we get  
\[
\begin{array}{lcl}
\gtNet\Nt{\pp\cdot\ps}&= & \Seq{\CommAs{\pp}{\nD}{\ps}}{\agtIS{\pp}{\ps}\Seq{\nD}{\Seq{\CommAsI{\ps}{\ko}{\pp}}{\agtSOS{\ps}{\pp}{\set{\Seq{\ok}{\Seq{\CommAs{\pp}{\pR}{\ps}}{\gtNet\Nt{\ps\cdot\pp}}}\, ,\ \Seq{\ko}{\Seq{\CommAs{\pp}{\pR}{\ps}}{\gtNet\Nt{\ps\cdot\pp}}}}}}}}\\
\gtNet\Nt{\ps\cdot\pp}&=& \Seq{\CommAsI{\pp}{\nD}{\ps}}{\Seq{\CommAs{\pp}{\nD}{\ps}}{\agtSOS{\ps}{\pp}{\set{\Seq{\ok}{\gtNet{\Nt'}{\pp\cdot\ps}}\, ,\ \Seq{\ko}{\gtNet{\Nt'}{\pp\cdot\ps}}}}}}\\
\gtNet{\Nt'}{\pp\cdot\ps} &=& \Seq{\CommAsI{\ps}{\ko}{\pp}}
{\Seq{\CommAsI{\pp}{\nD}{\ps}}
  {\Seq{\CommAs{\pp}{\pR}{\ps}}
    {\agtSOS{\ps}{\pp}{\set{\Seq{\ok}{\gtNet\Nt{\pp\cdot\ps}}\, ,\ \Seq{\ko}{\gtNet\Nt{\pp\cdot\ps}}}}}
}
}
\end{array} 
\] 
Starting from the participant list $\pi=\pp\cdot\ps\cdot\pp\cdot\pp\cdot\ps$ we get 
\[\gtNet\Nt{\pi}=\Seq{\CommAs{\pp}{\nD}{\ps}}{\agtIS{\pp}{\ps}\Seq{\nD}{\Seq{\CommAsI{\ps}{\ko}{\pp}}{\Seq{\CommAs{\pp}{\pR}{\ps}}{\agtSOS{\ps}{\pp}{\set{\Seq{\ok}{\gtNet\Nt{\pi}}\, ,\ \Seq{\ko}{\gtNet\Nt{\pi}}}}}}}}\] 
These examples show that we can obtain different global types for the same network starting from different lists of participants,  i.e. we can get $\gtNet\Nt\pi\neq\gtNet\Nt{\pi'}$ when $\pi\neq\pi'$.  
Notice that there is no $\pi$ such that  $\Nt$ is the network and $\gtNet\Nt\pi$ is the global type of \refToFigure{fig:exvd}, since in the two branches of the derivation $\DD$ the participants $\pp$ and $\ps$ alternate in different ways.
\end{exa}

Since processes are regular terms,  they have finite sets of subterms. Therefore  given a network $\Nt$ 
the set of networks $\Nt'$ such that, if $\pP{\pp}{\PP'}\in\Nt'$, then $\pP{\pp}{\PP}\in\Nt$ and
$\PP'$ is a subterm of $\PP$ is finite. Moreover the number of rotations of a list is the length of the list. 
Therefore $\gtNet\Nt\pi$ is a regular term, since it is generated by a finite number of equations. 

The following lemma shows why we use the list $\pi$ to build the global type $\gtNet\Nt\pi$. 
Namely, this lemma ensures that, when $\plays\Nt\subseteq\lpart\pi$, each player of the network $\Nt$ is also a player  of $\gtNet\Nt\pi$.

\begin{lem}\label{lem:gtNet}
If $\plays{\Nt}\subseteq\lpart{\pi}$, then $\plays{\Nt}=\plays{\gtNet{\Nt}{\pi}}$.
\end{lem}
\begin{proof}
The inclusion $\plays{\gtNet{\Nt}{\pi}}\subseteq\plays{\Nt}$ easily follows from  \refToDef{gtN}.  
In fact, only a player of $\Nt$ which occurs as first element of $\pi$ becomes a player of $\gtNet{\Nt}{\pi}$ and then ${\sf gt}$ is applied to networks $\Nt'$ such that $\plays{\Nt'}\subseteq\plays{\Nt}$.

To show the reverse implication, i.e. $\plays{\Nt}\subseteq\plays{\gtNet{\Nt}{\pi}}$,
we define the position of participant $\pp$ in the list $\pi$, notation $\pst\pi\pp$, by:
\[\pst{\q\cdot\pi'}\pp=\begin{cases}
 1     & \text{if }\pp=\q \\
1+ \pst{\pi'}\pp     & \text{otherwise}
\end{cases}\]
Let $\pp$ be a player of $\Nt$.  
The condition $\plays{\Nt}\subseteq\lpart{\pi}$ implies $n=\pst\pi\pp\geq 1$. 
We prove by induction on $n$ that $\pp \in \plays{\gtNet\Nt\pi}$. 
If $n=1$, then $\pi=\pp\cdot\pi'$ and, since $\pp$ is a player of $\Nt$, we have
$\Nt\equiv\pP\pp{\bp\q{i}{I}\la\PP}\parN\Nt'$. Therefore
$\gtNet{\Nt}{\pi} $ is either $\agtOg\pp\q{i\in I}{\la_i}{\G_i}$ or 
$\Seq{\CommAsI{\q}{\la_k}{\pp}}{\G_k}$, where $k\in I$ and $\la_k\preceq\la_i$ for all $i\in I$. 
Hence $\pp \in \plays{\gtNet\Nt\pi}$. 
If $n>1$, then $\pi=\q\cdot\pi'$ with $\pp\neq\q$ and $\pst\pi\pp = 1 + \pst{\pi'}{\pp}$. 
We distinguish two cases. 
\begin{itemize} 
\item If $\q \notin\plays\Nt$, that is, $\Nt \equiv \pP\q\inact \parN\Nt'$, then 
$\gtNet{\Nt}{\pi} = \gtNet{\Nt'}{\pi'} $. 
Since $\pst{\pi'}{\pp} = n-1$, by induction hypothesis, we get $\pp \in \plays{\gtNet{\Nt'}{\pi'}} = \plays{\gtNet\Nt\pi}$, as needed. 
\item Otherwise, $\Nt \equiv \pP\q{\bp\pr{i}{I}\la\Q} \parN\Nt'$ and so 
$\gtNet{\Nt}{\pi} $ is either $\agtOg\q\pr{i\in I}{\la_i}{\G_i}$ or 
$\Seq{\CommAsI{\pr}{\la_k}{\q}}{\G_k}$,  with $k\in I$ and $\la_k\preceq\la_i$ for all $i\in I$, 
 where $ \G_i=\gtNet{\pP\q{\Q_i}\parN\Nt'}{\pi'\cdot\q}$ for all $i\in I$. 
Note that, by definition, we have $\pst{\pi'\cdot\q}\pp = \pst{\pi'}\pp = n-1$, then by induction  hypothesis, we get $\pp\in\plays{\G_i}$ for all $i\in I$, hence $\pp \in \plays{\gtNet\Nt\pp}$. \qedhere
\end{itemize}
\end{proof} 

\begin{thm}\label{th:gtNet}
Let  $\Nt=\pP{\pp_1}{\PP_1}\parN\cdots\parN\pP{\pp_n}{\PP_n}$, then $ \tyn{\gtNet{\Nt}{\pp_1\cdots\pp_n}}{\Nt}$.
\end{thm}
\begin{proof}
The proof is by coinduction, hence  
we show that the set 
\[
\GN = \set{(\G,\Nt) \mid \G=\gtNet{\Nt}{\pi}\text{ for some $\pi$ such that }\plays{\Nt}\subseteq\lpart{\pi} }
\]
is consistent with respect to the rules defining the typing judgement $\tyn{\G}{\Nt}$, i.e.,
any $(\G,\Nt)\in\GN$ is the consequence of one of the typing rules whose premises 
are again in $\GN$.  
Let $(\G,\Nt) \in \GN$ and split cases on the shape of $\G$. 
\begin{itemize}
\item If $\G = \End$, then 
$\End=\gtNet{\Nt}{\pi}$ for some $\pi$ such that $\plays{\Nt}\subseteq\lpart{\pi}$.  By \refToLemma{lem:gtNet} $\plays\Nt=\plays\End$, which implies
$\Nt\equiv\pP{\pp}{\inact}$. 
Rule \rn{End} derives $\tyn\End{\pP\pp\inact}$ and has no premises. 
\item If $\G=\agtOg\pp\q{i\in I}{\la_i}{\G_i}$, then, by definition of $\GN$, we have 
$\agtOg\pp\q{i\in I}{\la_i}{\G_i} = \gtNet{\Nt}{\pi}$  for some $\pi$ such that $\plays{\Nt}\subseteq\lpart{\pi}$. 
By \refToDef{gtN} $\gtNet{\Nt}{\pi}=\agtOg\pp\q{i\in I}{\la_i}{\G_i}$ implies $\pi=\pi'\cdot\pp\cdot\pi''$ with $\lpart{\pi'}\cap\plays{\Nt}=\emptyset$
and 
$\Nt\equiv\pP\pp{\oup\q{i}{I}\la\PP}\parN\Nt'$ and
$\G_i=\gtNet{\pP\pp{\PP_i}\parN\Nt'}{\pi''\cdot\pp}$ for all $i\in I$. 
Therefore  $(\G_i,\pP\pp{\PP_i}\parN\Nt')\in\GN$ for all $i\in I$, since $\lpart{\pp\cdot\pi''} = \lpart{\pi''\cdot\pp}$. 
Moreover, from 
\refToLemma{lem:gtNet}, $\plays{\G_i}=\plays{\pP\pp{\PP_i}\parN\Nt'}$ for all $i\in I$. 
So $\plays{\G_i}\setminus\set{\pp} = \plays{\Nt'}$ and 
Rule \rn{Out} can be used to derive
$\tyn{\agtOg\pp\q{i\in I}{\la_i}{\G_i}}{\Nt}$ from premises in $\GN$.
\item 
If $\G$  is an input type,   then, by definition of $\GN$, we get $\G=\Seq{\CommAsI{\q}{\la}{\pp}}{\G'}$
for some $\la$ and $\G'$ and 
$\Seq{\CommAsI{\q}{\la}{\pp}}{\G'}= \gtNet{\Nt}{\pi}$  for some $\pi$ such that $\plays{\Nt}\subseteq\lpart{\pi}$. 
By \refToDef{gtN} $\gtNet{\Nt}{\pi}=\Seq{\CommAsI{\q}{\la}{\pp}}{\G'}$ implies $\pi=\pi'\cdot\pp\cdot\pi''$ with $\lpart{\pi'}\cap\plays{\Nt}=\emptyset$
and 
$\Nt\equiv\pP\pp{\inp\q{i}{I}\la\PP}\parN\Nt'$ and $\la=\la_k$ and $\G'=\G_k=\gtNet{\pP\pp{\PP_k}\parN\Nt'}{\pi''\cdot\pp}$
for some $k\in I$  such that $\la_k\preceq\la_i$ for all $i\in I$.  Therefore  $(\G',\pP\pp{\PP_k}\parN\Nt')\in\GN$, since $\lpart{\pp\cdot\pi''} = \lpart{\pi''\cdot\pp}$. 
Moreover, from 
\refToLemma{lem:gtNet}, $\plays{\G'}=\plays{\pP\pp{\PP_k}\parN\Nt'}$. 
So $\plays{\G'}\setminus\set{\pp} = \plays{\Nt'}$ and 
Rule \rn{In} can be used to derive
$\tyn{\Seq{\CommAsI{\q}{\la}{\pp}}{\G'}}{\Nt}$ from premises in $\GN$. 
\end{itemize} 
Then, the thesis follows, since 
$(\gtNet\Nt{\pp_1\cdots\pp_n},\Nt)\in\GN$ for $\Nt \equiv \pP{\pp_1}{\PP_1}\parN\cdots\parN\pP{\pp_n}{\PP_n}$. 
\end{proof}

\begin{figure}[h] 
\begin{math} 
\begin{array}{c}
\NamedRule{\rulename{Top-Out}}
 { } 
 {\agtO{\pp}{\q}i I{\la}{\G}\parG\Msg \stackred{\CommAs\pp{\la_h}\q}\G_h\parG\addMsg\Msg{\mq\pp{\la_h}\q}}
 {h\in I}
\\[4ex]
\NamedRule{\rulename{Top-In}}
 { } 
 {\agtI \pp\q \la \G\parG\addMsg{\mq\pp{\la_h}\q}\Msg \stackred{\CommAsI\pp{\la_h}\q}\G_h\parG\Msg}
 {h\in I}
\\[4ex]
\NamedRule{\rulename{Inside-Out}}
 {\G_i\parG \Msg\cdot\mq\pp{\la_i}\q \stackred\asCom\G'_i \parG \Msg'\cdot\mq\pp{\la_i}\q \quad \forall i \in I} 
 {\agtO{\pp}{\q}i I{\la}{\G}\parG\Msg \stackred \asCom\agtO{\pp}{\q}i I{\la}{\G'}\parG\Msg' }  
 {\pp\ne\play{\asCom} }
\\[4ex]
\NamedRule{\rulename{Inside-In}}
 {\G_i\parG\Msg\stackred\asCom\G_i'\parG\Msg'\quad \forall i \in I}
 {\agtI \pp\q \la \G\parG \mq\pp{\la_h}\q\cdot\Msg \stackred\asCom\agtI \pp\q \la {\G'}\parG  \mq\pp{\la_h}\q\cdot\Msg' }  
 {\q\ne\play{\asCom} ~~ h\in I}
\end{array}
\end{math} 
\caption{LTS for type configurations. }\label{fig:ltsgtAs}
\end{figure}

Global types provide an overall description of asynchronous multiparty protocols and their semantics is specified by an LTS on 
 types  in parallel with queues, dubbed \emph{type configurations}. 
Reduction rules are reported in 
\refToFig{ltsgtAs}. 
The first two rules 
 reduce outputs and inputs at top level in the standard way. 
The remaining two rules 
reduce  inside output and input choices.  These rules  are needed to enable interleaving between independent communications despite the sequential structure of global types. 
For example, we want to allow \linebreak $\Seq{\CommAs\pp{\la}\q}{\CommAs\pr{\la'}\ps}\parG\emptyset\stackred{\CommAs\pr{\la'}\ps}{\CommAs\pp{\la}\q}\parG\mq\pr{\la'}\ps$ when $\pp\neq\pr$, because, intuitively, actions performed by different players should be independent. 
This justifies the conditions $\pp\ne\play{\asCom}$ and $\q\ne\play{\asCom}$
 in Rules \rulename{Inside-Out}  and \rulename{Inside-In}, respectively. 
The requirements on the shapes of queues in these rules are more interesting. 
They enforce that inputs and outputs on the same channel, namely, the same ordered pair of participants, happen in the order prescribed by the global type, that is, an input cannot consume a message produced by a subsequent output. 
Indeed,  if we would take the following version of Rule \rulename{Inside-In} 
\[\NamedRule{}
 {\G_i\parG\Msg\stackred\asCom\G_i'\parG\Msg'\quad \forall i \in I}
 {\agtI \pp\q \la \G\parG \Msg \stackred\asCom\agtI \pp\q \la {\G'}\parG  \Msg' }  
 {\q\ne\play{\asCom} }\]
we would get  $\Seq{\CommAsI\pp{\la}\q}{\CommAs\pp{\la}\q}\parG\emptyset\stackred{\CommAs\pp{\la}\q}{\CommAsI\pp{\la}\q}\parN\mq\pp{\la}\q\stackred{\CommAsI\pp{\la}\q}\End\parG\emptyset$. 
This reduction breaks the order on the channel $\pp\q$ prescribed by  the global type 
$\Seq{\CommAsI\pp{\la}\q}{\CommAs\pp{\la}\q}$, 
which is to do first an input and then the corresponding output.
In our LTS $\Seq{\CommAsI\pp{\la}\q}{\CommAs\pp{\la}\q}\parG\emptyset$ is instead  stuck, as it should.  
In fact,  
the shape of the queue in Rule \rulename{Inside-In} ensures that $\co$ is not the matching output for any  input in the choice.  Similarly,  
the shapes of queues in Rule \rulename{Inside-Out}  ensures that $\co$ is not the matching input for any output in the choice.

A last remark is that the LTS rules are inductive.  Therefore, the premises of Rules \rulename{Inside-Out}  and \rulename{Inside-In} oblige $\play{\asCom}$ to occur in all $\G_i$ for $i\in I$.  
This condition is in line  with the boundedness of global types, see \refToDef{def:bound}.

 We are now ready to state and prove Session Fidelity, showing that a   session can correctly perform the protocol described by the global type derivable for its network.  
The  reverse property,  i.e. that the sessions can only do inputs and outputs allowed by the global types of their networks, will be proved for ``tamed'' global types in the Subject Reduction Theorem, \refToTheorem{thm:srA}.

\begin{thm}[Session Fidelity]\label{thm:sfA}
If $\tyn{\G}\Nt$ and 
$\G\parG\Msg\stackred\co\G'\parN\Msg'$, then $\Nt\parN\Msg\stackred{\co}\Nt'\parN\Msg'$ and  $\tyn{\G'}{\Nt'}$. 
\end{thm}
\begin{proof}
The proof is by induction on the reduction rules.
\begin{description}
\item [\rn{Top-In}] Then $\G=\agtI{\pp}{\q}{\la}{\G}$ and $\Msg\equiv\addMsg{\mq\pp{\la_h}\q}{\Msg'}$ and $\co=\CommAsI\pp{\la_h}\q$ and $\G'=\G_h$ with $h\in I$. Since $\tyn{\G}\Nt$ must be derived using rule \rn{In} we get $\Nt\equiv\pP\pp{\inp{\q}{j}{J}{\la}{\PP}}\parN\Nt_0$ with $I\subseteq J$ and $ \tyn{\G_i}{\pP\pp{\PP_i}\parN\Nt_0}$ for all $i\in I$. We conclude $\Nt\parN\Msg\stackred{\co}\pP\pp{\PP_h}\parN\Nt_0\parN\Msg'$  by Rule \rn{Rcv}  and $ \tyn{\G_h}{\pP\pp{\PP_h}\parN\Nt_0}$.
\item[\rulename{Inside-Out}] Then $\G=\agtO{\pp}{\q}i I{\la}{\G}$ and $\G'=\agtO{\pp}{\q}i I{\la}{\G'}$ and  $\G_i\parG \Msg\cdot\mq\pp{\la_i}\q \stackred\asCom\G'_i \parG \Msg' \cdot\mq\pp{\la_i}\q$ for all $i \in I$ and $\pp\ne\play{\asCom}$. Since $\tyn{\G}\Nt$ must be derived using rule \rn{Out} we get $\Nt\equiv\pP\pp{\oup{\q}{i}{I}{\la}{\PP}}\parN\Nt_0$ and $ \tyn{\G_i}{\pP\pp{\PP_i}\parN\Nt_0}$ for all $i\in I$. By induction $\pP\pp{\PP_i}\parN\Nt_0\parN\Msg\cdot\mq\pp{\la_i}\q\stackred{\co}\Nt_i\parN\Msg'\cdot\mq\pp{\la_i}\q$ and $\tyn{\G'_i}{\Nt_i}$ for all $i\in I$. The condition $\pp\ne\play{\asCom}$ ensures that these reductions do not modify the processes of participant $\pp$, i.e. $\Nt_i\equiv\pP\pp{\PP_i}\parN\Nt_0'$ for all $i\in I$ and some $\Nt'_0$. Moreover these reductions do not depend on the messages $\mq\pp{\la_i}\q$, which are in the queue before and after these reductions. We   get $\Nt\parN\Msg\stackred{\co}\pP\pp{\oup{\q}{i}{I}{\la}{\PP}}\parN\Nt_0'\parN\Msg'$. 
By \refToLemma{ep} $\tyn{\G'_i}{\pP\pp{\PP_i}\parN\Nt_0'}$ implies $\plays{\G'_i}\setminus\set\pp=\plays{\Nt'_0}$ for all $i\in I$. Then we can derive  
 $ \tyn{\G'}{\pP\pp{\oup{\q}{i}{I}{\la}{\PP}}\parN\Nt_0'}$ using Rule \rulename{Out}. 
\end{description}
The proofs for the remaining rules are similar.
\end{proof}

 We extend the   properties of multiparty sessions to type configurations. 
Let us define the \emph{lockstep transition relation}   
$\G\parN\Msg \lockred{\Delta} \G' \parN\Msg'$ by:
\[\G \parallel \Msg \lockred{\Delta} \G' \parallel \Msg' \text{ if } 
\begin{array}[t]{l} 
\Delta = \set{\beta_1,\ldots, \beta_n} \text{ is a maximal coherent set for $\G\parN\Msg$ and } \\[3pt] 
\G \parallel \Msg \stackred{\beta_1}\cdots \stackred{\beta_n} \G' \parallel \Msg'
\end{array}
\]
where $\Delta$ is coherent for $\G\parN\Msg$ if 
\begin{enumerate}
\item for all $\co_1,\co_2\in\Delta$, $\play{\co_1} = \play{\co_2}$ implies $\co_1 = \co_2$, and 
\item for all $\co \in \Delta$, $\G\parN\Msg \stackred{\co}$. 
\end{enumerate}

 A type configuration  $\G\parN\Msg$  is  {\em deadlocked} if $\G\neq\End$ and there is no $\beta$ such that $\G\parN\Msg \stackred{\co}$. 

\begin{defi}[Input-enabling type configuration] 
  A type configuration $\G \parallel \Msg$ is \emph{input-enabling} if
  $\pp \in \plays\G$ implies that, 
  for all complete 
  \[\lockSC = \set {
    {\G_k \parallel \Msg_k} \lockredsub{\Delta_k}{k}\G_{k+1} \parallel
    \Msg_{k+1}}_{k <  x}\]  with  
    $\G_0 \parallel \Msg_0$  
    =
  \mbox{$\G \parallel \Msg$}, there exists $h < x$ such that  $\beta \in\Delta_h$ with $\play\beta=\pp$. 
\end{defi}
\noindent
This definition ensures that no player is stuck, and this coincides with input enabling, since outputs can always 
be done. 

\begin{defi}[Queue-consuming type configuration] 
  A type configuration $\G \parallel
  \Msg\,$ is \emph{queue-consuming} if $\Msg\equiv\addMsg{\mq{\pp}{\la}{\q}}{\Msg'}$
  implies
  that,
  for all  complete 
  \[\lockSC = \set {
    \G_k \parallel \Msg_k \lockredsub{\Delta_k}{k}\G_{k+1} \parallel
    \Msg_{k+1}}_{k < x}\]  with  
    $  \G_0 \parallel \Msg_0   
 = \G \parallel \Msg$, there exists $h < x$ such that
  $\CommAsI{\pp}{\M}{\q} \in\Delta_h$.
\end{defi}


\subsection{Type Inference}
\label{sect:inference}

In this  subsection, we will describe an algorithm to infer global types from networks, proving its soundness and completeness  with respect to the typing system. 
That is, the algorithm   applied to a network $\Nt$  enumerats  all and only those global types which can be derived for $\Nt$.  
Note that, since a network may have more than one global type, to be complete, the algorithm needs to be non-deterministic. 

The first step towards defining an algorithm is the introduction of a finite representation for global types.\footnote{In a similar way we can and should give a finite representation for processes,  but we avoid it to reduce noise.} 
Since global types are regular terms, from results in \cite{Cour83,AdamekMV06}, they can be represented as finite systems of regular syntactic equations 
formally defined below. 
First, a global type pattern is a finite term generated by the following grammar: 
\[
\Gpat ::= \End \mid \agtO\pp\q{i}{I}\la\Gpat \mid \agtII\pp\q{i}{I}\la\Gpat \mid \X 
\]
where $\X$ is a variable taken from a countably infinite set. 
  We   denote by $\vars\Gpat$ the set of variables occurring in $\Gpat$.

A \emph{substitution} $\theta$ is a finite partial map from variables to global types. 
  We   denote 
by $\theta \sbtplus \sigma$ the union of two substitutions such that $\theta(\X) = \sigma(\X)$, for all $\X \in \dom\theta\cap\dom\sigma$, 
and by $\Gpat\theta$ the application of $\theta$ to $\Gpat$.
   We define   $\theta \sbtord \sigma$ if $\dom\theta\subseteq\dom\sigma$ and 
$\theta(\X) = \sigma(\X)$, for all $\X \in \dom\theta$.  
Note that, if $\vars\Gpat\subseteq\dom\theta$, then $\Gpat\theta$ is a global type.

An \emph{equation} has shape $\agteq\X\Gpat$ and a \emph{(regular) system of equations} $\eqsys$ is a finite set of equations such that 
$\agteq\X{\Gpat_1}$ and $\agteq\X{\Gpat_2}\in\eqsys$ imply $\Gpat_1 = \Gpat_2$. 
  We   denote by $\vars\eqsys$ the set $\bigcup \{ \vars\Gpat\cup\{\X\} \mid \agteq\X\Gpat\in\eqsys\}$ and by $\dom\eqsys$ the set $\{\X\mid \agteq\X\Gpat \in \eqsys \}$. 
A \emph{solution} of a system $\eqsys$ is a substitution $\theta$ such that $\vars\eqsys \subseteq \dom\theta$ and, for all $\agteq\X\Gpat\in\eqsys$, $\theta(\X) = \Gpat\theta$ holds. 
  We   denote by $\gsol\eqsys$ the set of all solutions of $\eqsys$ and note that 
$\eqsys_1\subseteq\eqsys_2$ implies $\gsol{\eqsys_2}\subseteq\gsol{\eqsys_1}$. 

The algorithm follows essentially the structure of coSLD resolution of coinductive logic programming \cite{Simon06,SimonBMG07,SimonMBG06,AnconaD15}, namely, the extension of standard SLD resolution capable to deal with regular terms and coinductive predicates. 
A \emph{goal} is a pair   $\pair\Nt\X$   of a network $\Nt$ and a variable $\X$. 
The algorithm takes  as  input a goal $\pair\Nt\X$ and returns a set of equations $\eqsys$ such that the solution for the variable $\X$ in $\eqsys$ is a global type for the network $\Nt$. 
The key idea borrowed from coinductive logic programming is to keep track of already encountered goals to detect cycles, avoiding non-termination. 

The inference judgement has the following shape
$\tyalg{\Goals}{\pair\Nt\X}{\eqsys}$, where 
$\Goals$ is a set of goals, all with different variables  which are all  different from $\X$. 
Rules defining the inference algorithm are reported in \refToFig{tyalg}. 
\begin{figure}
\begin{math}
\begin{array}{c}
\NamedRule{\rn{\infn{End}}}{}{ \tyalg{\Goals}{\pair{\pP\pp\inact}{\X}}{\{\agteq\X\End\}}}{} 
\qquad
\NamedRule{\rn{\infn{Cycle}}}{ }{ \tyalg{\Goals, \pair\Nt\Y}{\pair\Nt\X}{\{\agteq\X\Y\}}}{} 
\\[4.5ex]
\NamedRule{{\rn{\infn{Out}}}}{
  \tyalg{ \Goals' }{\pair{\pP\pp{\PP_i}\parN\Nt}{\Y_i}}{\eqsys_i} \ \ \forall i \in I 
}{ \tyalg{\Goals}{\pair{\pP{\pp}{\PP} \parN \Nt}{\X}}{\eqsys} }
{\begin{array}{l}
 \Goals'=\Goals, \pair{\pP{\pp}{\PP} \parN \Nt}{\X} \\
\PP = \oup\q{i}{I}\la\PP \\ 
\Y_i\   fresh\ \forall i \in I \\
\eqsys = \{\agteq\X{\agtO\pp\q{i}{I}\la\Y}\}\cup\bigcup_{i \in I} \eqsys_i  \\ 
 \plays{\three{\Goals'}{\eqsys_i}{\Y_i}}\setminus\set\pp = \plays{\Nt}\ \forall i \in I    
\end{array}}
\\[6.5ex] 
\NamedRule{\rn{\infn{In}}}{
  \tyalg{\Goals'}{\pair{\pP\pp{\PP_i}\parN\Nt}{\Y_i}}{\eqsys_i} \ \ \forall i \in I 
}{ \tyalg{\Goals}{\pair{\pP{\pp}{\PP} \parN \Nt}{\X}}{\eqsys} }
{\begin{array}{l}
\Goals'=\Goals, \pair{\pP{\pp}{\PP} \parN \Nt}{\X}\\
\PP = \inp\q{j}{J}\la\PP \quad 
\emptyset \ne I\subseteq J \\ 
\Y_i \ fresh\ \forall i\in I  \\
\eqsys = \{\agteq\X{ \agtII\q\pp{i}{I}\la\Y}\} \cup\bigcup_{i\in I} \eqsys_i  \\ 
\plays{\three{\Goals'}{\eqsys_i}{\Y_i}}\setminus\set\pp = \plays{\Nt}\ \forall i \in I 
\end{array}}  
\end{array}
\end{math}
\caption{Rules of the inference algorithm.}
\label{fig:tyalg}
\end{figure}
For a terminated  network  the algorithm returns just one equation $\agteq\X\End$ (Rule \rn{\infn{End}}). 
For other networks, the algorithm selects a participant and analyses its process. 
If the process is an output choice (Rule \rn{\infn{Out}}), the algorithm continues analysing all branches of the output choice. 
If the process is an input choice (Rule \rn{\infn{In}}), the algorithm selects a subset of its branches and  continues analysing them. 
In both cases, subnetworks are analysed adding to the set $\Goals$ the goal in the conclusion of the rule, to be able to subsequently recognise when  a cycle is encountered. 
After having evaluated subnetwork, the algorithm collects all the resulting equations plus another one for the current variable. 
Note that variables for goals in the premises are fresh. 
This is important to ensure that the set of equations $\eqsys$ in the conclusion is indeed a regular system of equations (there is at most one equation for each variable). 
Finally, Rule \rn{\infn{Cycle}} detects cycles: if the network in the current goal appears also in the set $\Goals$ the algorithm can stop, 
returning just one equation unifying the two variables associated with the network. 

In Rules \rn{\infn{Out}} and \rn{\infn{In}}, 
the side condition   $\plays{\three{\Goals'}{\eqsys_i}{\Y_i}}\setminus\set\pp = \plays{\Nt}$ for all $i \in I $ is needed to ensure that the resulting global type associated with the variable $\X$ satisfies the conditions on players required by Rules \rn{Out} and \rn{In} in \refToFigure{fig:cntr}.  
The set $\plays{\three\Goals \eqsys\Gpat}$ is defined as the set of players of a global type, but with the following additional clause to handle variables: 
\[
\plays{\three\Goals\eqsys\X} = \begin{cases}
\plays{\three\Goals\eqsys\Gpat} & \text{if } \agteq\X\Gpat \in \eqsys \\ 
\plays{\Nt} & \text{if } \X\notin\dom\eqsys\text{ and } \pair\Nt\X \in \Goals \\ 
\emptyset  & \text{otherwise} 
\end{cases} 
\]

Let $\eqsys$ be a system of equations and $\Goals$ a set of goals. 
A solution $\theta \in \gsol\eqsys$   {\em agrees} with $\Goals$ if 
$\pair\Nt\X\in\Goals$ implies $\plays{\theta(\X)} = \plays\Nt$ for all $\X\in\vars\eqsys$. 
We   denote by $\gsol[\Goals]{\eqsys}$ the set of all solutions of $\eqsys$ agreeing with $\Goals$. 
We say that a system of equations $\eqsys$ is {\em guarded} if 
  $\agteq\X\Y$ and $ \agteq\Y\Gpat$ in $\eqsys$ imply   that $\Gpat$ is not a variable. 
Finally, $\eqsys$ is \emph{$\Goals$-closed}  if it is guarded and 
$\dom\eqsys\cap\vars\Goals = \emptyset$ and $\vars\eqsys\setminus\dom\eqsys \subseteq\vars\Goals$. 

  \begin{exa}\refToFigure{fig:exia} shows an application of the inference algorithm to the hospital network for getting a set of equations whose solution is the global type derived for this network in \refToFigure{fig:exvd}. We notice that $\II$ and $\II'$ in \refToFigure{fig:exia} only differ for the names of variables, so they become the same derivation $\DD$ in \refToFigure{fig:exvd}. The set of obtained equation is
\[\begin{array}{lll}
\agteq\X {\Seq{\CommAs{\pp}{\nD}{\ps}}{\X_1}&\agteq{\X_1} {\agtIS{\pp}{\ps}{\set{\Seq{\nD}{\X_2}\, ,\, \Seq{\pR}{\X'_2}}}}}\\
\agteq{\X_2} {\agtSOS{\ps}{\pp}{\set{\Seq{\ok}{\X_3}\, ,\,  \Seq{\ko}{\X_5}}}}&\agteq{\X_3}{\Seq{\CommAsI{\ps}{\ok}{\pp}}{\X_4}}&\agteq{\X_4}\X\\\agteq{\X_5}{\Seq{\CommAsI{\ps}{\ko}{\pp}}{\X_6}}&
\agteq{\X_6}{\Seq{\CommAs{\pp}{\pR}{\ps}}{\X_7}}&\agteq{\X_7}\X\\ \agteq{\X'_2} {\agtSOS{\ps}{\pp}{\set{\Seq{\ok}{\X'_3}\, ,\, \Seq{\ko}{\X'_5}}}}&\agteq{\X'_3}{\Seq{\CommAsI{\ps}{\ok}{\pp}}{\X'_4}}&\agteq{\X'_4}\X\\\agteq{\X'_5}{\Seq{\CommAsI{\ps}{\ko}{\pp}}{\X'_6}}&\agteq{\X'_6}{\Seq{\CommAs{\pp}{\pR}{\ps}}{\X'_7}}&\agteq{\X'_7}\X\end{array}\]

\end{exa}

\begin{figure} 
 
\prooftree
\hspace{-10em}
\II\qquad\qquad\qquad\II'
\justifies
\prooftree
 \tyalg{\pair{\pP{\pp}{\PP} \parN \pP{\ps}{\Srv}}{\X}}{\pair{\pP{\pp}{\PP_1} \parN \pP{\ps}{\Srv}}{\X_1}}{\eqsys_1}
\justifies
 \tyalg{\emptyset}{\pair{\pP{\pp}{\PP} \parN \pP{\ps}{\Srv}}{\X}}{\eqsys}
 \using \rn{\infn{Out}}
\endprooftree
 \using \rn{\infn{In}}
\endprooftree

\bigskip

where $\II$ is 

\bigskip

\prooftree
\prooftree
\prooftree
\justifies
 \tyalg{\Goals_1}{\pair{\pP{\pp}{\PP} \parN \pP{\ps}{\Srv}}{\X_4}}{\eqsys_4}
 \using \rn{\infn{Cycle}}
 \endprooftree
\justifies
 \tyalg{\Goals}{\pair{\pP{\pp}{\PP_1} \parN \pP{\ps}{\Srv}}{\X_3}}{\eqsys_3}
  \using \rn{\infn{In}}
\endprooftree
\prooftree
\prooftree
\justifies
\tyalg{\Goals_3}{\pair{\pP{\pp}{\PP} \parN \pP{\ps}{\Srv}}{\X_7}}{\eqsys_7}
\using \rn{\infn{Cycle}}
 \endprooftree
\justifies
\prooftree
\tyalg{\Goals_2}{\pair{\pP{\pp}{\s!\lql;\PP} \parN \pP{\ps}{\Srv}}{\X_6}}{\eqsys_6}       
 \justifies
 \tyalg{\Goals}{\pair{\pP{\pp}{\PP_1} \parN \pP{\ps}{\Srv}}{\X_5}}{\eqsys_5}
  \using \rn{\infn{In}}
\endprooftree
 \using \rn{\infn{Out}}
\endprooftree
\justifies
 \tyalg{\pair{\pP{\pp}{\PP} \parN \pP{\ps}{\Srv}}{\X},\pair{\pP{\pp}{\PP_1} \parN \pP{\ps}{\Srv}}{\X_1}}{\pair{\pP{\pp}{\PP_1} \parN \pP{\ps}{\Srv_1}}{\X_2}}{\eqsys_2}
  \using \rn{\infn{Out}}
\endprooftree 

\bigskip

and $\II'$ is

\prooftree
\prooftree
\prooftree
\justifies
 \tyalg{\Goals'_1}{\pair{\pP{\pp}{\PP} \parN \pP{\ps}{\Srv}}{\X'_4}}{\eqsys'_4}
 \using \rn{\infn{Cycle}}
 \endprooftree
\justifies
 \tyalg{\Goals'}{\pair{\pP{\pp}{\PP_1} \parN \pP{\ps}{\Srv}}{\X'_3}}{\eqsys'_3}
  \using \rn{\infn{In}}
\endprooftree
\prooftree
\prooftree
\justifies
\tyalg{\Goals'_3}{\pair{\pP{\pp}{\PP} \parN \pP{\ps}{\Srv}}{\X'_7}}{\eqsys'_7}
\using \rn{\infn{Cycle}}
 \endprooftree
\justifies
\prooftree
\tyalg{\Goals'_2}{\pair{\pP{\pp}{\s!\lql;\PP} \parN \pP{\ps}{\Srv}}{\X'_6}}{\eqsys'_6}       
 \justifies
 \tyalg{\Goals'}{\pair{\pP{\pp}{\PP_1} \parN \pP{\ps}{\Srv}}{\X'_5}}{\eqsys'_5}
  \using \rn{\infn{In}}
\endprooftree
 \using \rn{\infn{Out}}
\endprooftree
\justifies
 \tyalg{\pair{\pP{\pp}{\PP} \parN \pP{\ps}{\Srv}}{\X},\pair{\pP{\pp}{\PP_1} \parN \pP{\ps}{\Srv}}{\X_1}}{\pair{\pP{\pp}{\PP_1} \parN \pP{\ps}{\Srv_1}}{\X'_2}}{\eqsys'_2}
  \using \rn{\infn{Out}}
\endprooftree

\bigskip

and 

\bigskip

$\eqsys=\set{\agteq\X {\Seq{\CommAs{\pp}{\nD}{\ps}}{\X_1}}}\cup\eqsys_1\quad\eqsys_1=\set{\agteq{\X_1} {\agtIS{\pp}{\ps}{\set{\Seq{\nD}{\X_2}\, ,\ \Seq{\pR}{\X'_2}}}}}\cup\eqsys_2\cup\eqsys_2'$ 

$\eqsys_2=\set{\agteq{\X_2} {\agtSOS{\ps}{\pp}{\set{\Seq{\ok}{\X_3}\, ,\ \Seq{\ko}{\X_5}}}}}\cup\eqsys_3\cup\eqsys_5\quad\eqsys_3=\set{\agteq{\X_3}{\Seq{\CommAsI{\ps}{\ok}{\pp}}{\X_4}}}\cup\eqsys_4$

$\eqsys_4=\set{\agteq{\X_4}\X}\quad\eqsys_5=\set{\agteq{\X_5}{\Seq{\CommAsI{\ps}{\ko}{\pp}}{\X_6}}}\cup\eqsys_6$

$\eqsys_6=\set{\agteq{\X_6}{\Seq{\CommAs{\pp}{\pR}{\ps}}{\X_7}}}\cup\eqsys_7\quad\eqsys_7=\set{\agteq{\X_7}\X}$

$\eqsys'_2=\set{\agteq{\X'_2} {\agtSOS{\ps}{\pp}{\set{\Seq{\ok}{\X'_3}\, ,\ \Seq{\ko}{\X'_5}}}}}\cup\eqsys'_3\cup\eqsys'_5\quad\eqsys'_3=\set{\agteq{\X'_3}{\Seq{\CommAsI{\ps}{\ok}{\pp}}{\X'_4}}}\cup\eqsys'_4$

$\eqsys'_4=\set{\agteq{\X'_4}\X}\quad\eqsys'_5=\set{\agteq{\X'_5}{\Seq{\CommAsI{\ps}{\ko}{\pp}}{\X'_6}}}\cup\eqsys'_6$

$\eqsys'_6=\set{\agteq{\X'_6}{\Seq{\CommAs{\pp}{\pR}{\ps}}{\X'_7}}}\cup\eqsys'_7\quad\eqsys'_7=\set{\agteq{\X'_7}\X}$

$\Goals=\set{\pair{\pP{\pp}{\PP} \parN \pP{\ps}{\Srv}}{\X},\pair{\pP{\pp}{\PP_1} \parN \pP{\ps}{\Srv}}{\X_1},\pair{\pP{\pp}{\PP_1} \parN \pP{\ps}{\Srv_1}}{\X_2}}$

$\Goals'=\set{\pair{\pP{\pp}{\PP} \parN \pP{\ps}{\Srv}}{\X},\pair{\pP{\pp}{\PP_1} \parN \pP{\ps}{\Srv}}{\X_1},\pair{\pP{\pp}{\PP_1} \parN \pP{\ps}{\Srv_1}}{\X_2'}}$

$\Goals_1=\Goals\cup\set{\pair{\pP{\pp}{\PP_1} \parN \pP{\ps}{\Srv}}{\X_3}}\quad\Goals_1'=\Goals'\cup\set{\pair{\pP{\pp}{\PP_1} \parN \pP{\ps}{\Srv}}{\X_3'}}$

$\Goals_2=\Goals\cup\set{\pair{\pP{\pp}{\PP_1} \parN \pP{\ps}{\Srv}}{\X_5}}\quad\Goals_2'=\Goals'\cup\set{\pair{\pP{\pp}{\PP_1} \parN \pP{\ps}{\Srv}}{\X_5'}}$

$\Goals_3=\Goals_2\cup\set{\pair{\pP{\pp}{\s!\lql;\PP} \parN \pP{\ps}{\Srv}}{\X_6}}\quad\Goals_3'=\Goals_2'\cup\set{\pair{\pP{\pp}{\s!\lql;\PP} \parN \pP{\ps}{\Srv}}{\X_6'}}$

\caption{An application of the inference algorithm to the hospital network, where 
  $\PP$, $\PP_1$, $\Srv$, $\Srv_1$, $\G$, $\G_1$ and $\G_2$ are defined in the caption of \refToFigure{fig:exvd}. }\label{fig:exia}
\end{figure} 

 Toward proving properties of the inference algorithm,   we check a couple of auxiliary lemmas. 

\begin{lem}\label{lem:vars-eq}
If $\tyalg{\Goals}{\pair\Nt\X}{\eqsys}$, then 
$\eqsys$ is $\Goals$-closed. 
\end{lem}
\begin{proof}
By a straightforward induction on the derivation of $\tyalg{\Goals}{\pair\Nt\X}{\eqsys}$. 
\end{proof} 

\begin{lem}\label{lem:eq-play}
If $\eqsys$ is an $\Goals$-closed system of equations and $\vars{\Gpat}\subseteq\vars\eqsys$, then 
$\plays{\three\Goals\eqsys\Gpat} = \plays{\Gpat\theta}$ for all $\theta \in \gsol[\Goals]{\eqsys}$. 
\end{lem}
\begin{proof}
To prove the inclusion $\plays{\Gpat\theta}\subseteq \plays{\three\Goals\eqsys\Gpat}$, let $\pp\in\plays{\Gpat\theta}$.  We show  $\pp\in\plays{\three\Goals\eqsys\Gpat}$ by induction on the least distance $d$ of a communication with player  $\pp$ from the root of $\Gpat\theta$. 
First of all, 
it is easy to see that  there is $\Gpat'$ such that $\plays{\three\Goals\eqsys\Gpat} = \plays{\three\Goals\eqsys{\Gpat'}}$ and $\Gpat\theta = \Gpat'\theta$  
and 
either $\Gpat' = \agtb\pr\ps{i}{I}\la\Gpat$ or $\Gpat' = \X$ and $\X \notin\dom\eqsys$. 
Indeed, we have $\Gpat\ne\End$ since $\plays{\End\theta} = \plays\End = \emptyset$. 
First we show that  
$\Gpat = \X \in\dom\eqsys$ is impossible. In this case 
$\agteq\X{\Gpat_1} \in \eqsys$ and  
we have $\Gpat\theta = \Gpat_1\theta$ and $\plays{\three\Goals\eqsys\Gpat} = \plays{\three\Goals\eqsys{\Gpat_1}}$, since $\theta$ is a solution of $\eqsys$. 
Hence, again $\Gpat_1\ne\End$ and if $\Gpat_1 = \Y \in \dom\eqsys$, namely, $\agteq\Y{\Gpat_2}\in\eqsys$, we have 
$\Gpat_1\theta = \Gpat_2\theta$ and $\plays{\three\Goals\eqsys{\Gpat_1}} = \plays{\three\Goals\eqsys{\Gpat_2}}$ and, 
since $\eqsys$ is $\Goals$-closed and so guarded, we have that $\Gpat_2$ is not a variable. 
\begin{description}
\item [Case $d = 0$] 
If $\Gpat' = \X \notin\dom\eqsys$, then $\pair\Nt\X \in \Goals$ and $\plays{\three\Goals\eqsys{\Gpat'}} = \plays{\Nt}$. 
Since $\theta$ agrees with $\Goals$, we have $\plays{\Gpat'\theta} = \plays{\theta(\X)} = \plays{\Nt}$, hence $\pp \in \plays{\three\Goals\eqsys{\Gpat'}}$. 
If $\Gpat' = \agtb\pr\ps{i}{I}\la\Gpat$, then $\Gpat'\theta = \agtbg\pr\ps{i\in I}{\la_i}{\Gpat_i\theta}$ and $\play{\Gpat'} = \play{\Gpat'\theta} = \pp$. 
By definition we have $\plays{\three\Goals\eqsys{\Gpat'}} = \set{\play{\Gpat'}} \cup \bigcup_{i \in I} \plays{\Goals,\eqsys,\Gpat_i}$, hence 
$\pp \in \plays{\three\Goals\eqsys{\Gpat'}}$.  

\item [Case $d > 0$]
If $\Gpat' = \X \notin\dom\eqsys$, the proof is as above. 
If $\Gpat' = \agtb\pr\ps{i}{I}\la\Gpat$, then $\pp \ne \play{\Gpat'\theta}$, hence  $\pp\ne\pr$ if $\dagger=!$ and
$\pp\ne\ps$ if $\dagger=?$.  
We have $\Gpat'\theta = \agtbg\pr\ps{i\in I}{\la_i}{\Gpat_i\theta}$ and there is $k \in I$ such that $\pp\in\plays{\Gpat_k\theta}$ and the distance decreases. 
Then, by induction hypothesis, we get $\pp \in \plays{\three\Goals\eqsys{\Gpat_k}} \subseteq \plays{\three\Goals\eqsys{\Gpat'}}$, as needed. 
\end{description}

To prove the other inclusion, $\plays{\three\Goals\eqsys{\Gpat}}\subseteq \plays{\Gpat\theta}$, we just have to check that the sets $\plays{\Gpat\theta}$ respect the equations defining $\plays{\three\Goals\eqsys{\Gpat}}$. 
All cases are trivial except for $\Gpat = \X$. 
If $\X \in \dom \eqsys$, that is, $\agteq\X{\Gpat'} \in \eqsys$, then 
$\Gpat\theta = \theta(\X) = \Gpat'\theta$, hence $\plays{\Gpat\theta} = \plays{\Gpat'\theta}$, as needed. 
Otherwise, $\X\in\vars\Goals$, that is, $\pair\Nt\X \in \Goals$, hence 
$\plays{\three\Goals\eqsys{\Gpat}} = \plays\Nt$. 
Since $\theta$ agrees with $\Goals$, we have $\plays{\Gpat\theta} = \plays{\theta(\X)} = \plays\Nt$, as needed. 
\end{proof}

 \begin{figure}
 \begin{math}
 \begin{array}{c}
 \NamedRule{\rn{\itrp{End}}}{}{ \tynIP\Nset{\pP\pp\inact:\End} }{} \qquad 
 \NamedRule{\rn{\itrp{Cycle}}}{}{ \tynIP{\Nset,\ipair\G\Nt}{\Nt:\G} }{} \\[3ex] 
 \NamedRule{\rn{\itrp{Out}}}{\mbox{$\begin{array}{c}
   \tynIP{\Nset,\ipair\G\Nt}{\pP\pp{\PP_i}\parN\Nt':\G_i} \\ \plays{\G_i}\setminus\set\pp=\plays{\Nt'}
   \ \ \forall i \in I\end{array}$}
 }{ \tynIP{\Nset}{\Nt:\G} }{\begin{array}{c}\G=\agtO{\q}{\pp}i I{\la}{\G}\\ \Nt=\pP\pp{\oup{\q}{i}{I}{\la}{\PP}}\parN\Nt' \end{array}} \\[3ex] 
 \NamedRule{\rn{\itrp{In}}}{\mbox{$\begin{array}{c}
   \tynIP{\Nset,\ipair\G\Nt}{\pP\pp{\PP_i}\parN\Nt':\G_i} \\ \plays{\G_i}\setminus\set\pp=\plays{\Nt'}
 \ \ \forall i \in I\end{array}$}
 }{ \tynIP{\Nset}{\Nt:\G} }{\begin{array}{c}\G=\agtI{\q}{\pp}{\la}{\G}\\ \Nt=\pP\pp{\inp{\q}{j}{J}{\la}{\PP}}\parN\Nt'\\ I\subseteq J\end{array}} 
 \end{array} 
 \end{math}
 \caption{Inductive typing rules for networks.}\label{fig:itr}
 \end{figure}

To show soundness   and completeness   of our inference algorithm, it is handy to formulate an inductive version of our typing rules, see \refToFig{itr}, 
where $\Nset$ ranges over sets of pairs $\ipair\G\Nt$. 
We can give an inductive formulation since all infinite derivations using the typing rules of  \refToFig{cntr}  are regular, i.e. the number of different subtrees of a derivation for a judgement $\tyn\G\Nt$ is finite.  In fact,   it is bounded by the product of the number of different subterms of $\G$ and the number of different subnetworks of $\Nt$, which are both finite as $\G$ and (processes in) $\Nt$ are regular. 
 Applying  the standard transformation according to \cite[Theorem 5.2]{Dagnino21} (see Definition 5.1 for the notational convention) from a coinductive to an inductive formulation we get the typing rules shown in \refToFig{itr}. 

  In the following two lemmas we relate inference and inductive  derivability.  

\begin{lem}\label{lem:inf-sound}
If $\tyalg{\Goals}{\pair\Nt\X}{\eqsys}$,  then    $\tynI{\Goals\theta}{\Nt : \theta(\X)}$  
for all $\theta \in \gsol[\Goals]{\eqsys}$ such that $\vars\Goals\subseteq\dom\theta$.  \end{lem}
\begin{proof}
By induction on the derivation of $\tyalg{\Goals}{\pair\Nt\X}{\eqsys}$. 
\begin{description}
\item [\rn{\infn{End}}]
We have $\eqsys = \set{\agteq\X\End}$, hence $\theta(\X) = \End$ and the thesis follows by rule $\rn{\itr{End}}$. 
\item [\rn{\infn{Cycle}}] 
We have $\eqsys = \set{\agteq\X\Y}$ and $\Goals = \Goals',\pair\Nt\Y$. 
Then, $\theta(\X) = \theta(\Y)$ and   the thesis follows by Rule \rn{\itr{Cycle}}.  
\item [\rn{\infn{Out}}] 
We have $\Nt \equiv \pP\pp{\oup\q{i}{I}\la\PP} \parN \Nt'$   and   $\tyalg{\Goals,\pair\Nt\X}{\pair{\Nt_i}{\Y_i}}{\eqsys_i}$ with $\Y_i$ fresh and $\Nt_i \equiv \pP\pp{\PP_i} \parN\Nt'$   and $\plays{\three{\Goals,\pair\Nt\X}{\eqsys_i}{\Y_i}}\setminus\set\pp = \plays{\Nt'}$    for all $i \in I$ and $\eqsys = \set{\agteq\X{\agtO\pp\q{i}{I}\la\Y}} \cup\bigcup_{i \in I} \eqsys_i$.
Since $\eqsys_i\subseteq\eqsys$, we have $\theta \in \gsol{\eqsys_i}$.   Being $\theta \in \gsol[\Goals]{\eqsys}$, \refToLem{eq-play} implies $\plays{\three\Goals\eqsys\X} = \plays{\Nt}$. So   we get that $\theta$ agrees with $\Goals,\pair\Nt\X$. 
Then, by induction hypothesis, we   have 
$\tynI{\Goals\theta, \ipair {\theta(\X)}\Nt}{\Nt_i : \theta(\Y_i)}$ for all $i \in I$. The thesis follows by Rule \rn{\itr{Out}}, since 
$\theta(\X) = \agtOg\pp\q{i\in I}{\la_i}{\theta(\Y_i)}$ and $\plays{\three{\Goals,\pair\Nt\X}{\eqsys_i}{\Y_i}}\setminus\set\pp = \plays{\Nt'}$ implies 
$\plays{\theta(\Y_i)}\setminus\set\pp= \plays{\Nt'}$ for all $i \in I$ by \refToLem{eq-play}.  
\item [\rn{\infn{In}}]
 We have $\Nt \equiv \pP\pp{\inp\q{j}{J}\la\PP} \parN \Nt'$ and $\tyalg{\Goals,\pair\Nt\X}{\pair{\Nt_i}{\Y_i}}{\eqsys_i}$ with $\Y_i$ fresh and $\Nt_i \equiv \pP\pp{\PP_i} \parN\Nt'$   and $\plays{\three{\Goals,\pair\Nt\X}{\eqsys_i}{\Y_i}}\setminus\set\pp = \plays{\Nt'}$    for all $i \in I\subseteq J$, and $\eqsys = \set{\agteq\X{\agtII\q\pp{i}{I}\la\Y}} \cup\bigcup_{i \in I} \eqsys_i$. 
Since $\eqsys_i\subseteq\eqsys$, we have $\theta \in \gsol{\eqsys_i}$.    Being $\theta \in \gsol[\Goals]{\eqsys}$, \refToLem{eq-play} implies $\plays{\three\Goals\eqsys\X} = \plays{\Nt}$. So   we get that $\theta$ agrees with $\Goals,\pair\Nt\X$. 
Then, by induction hypothesis, we   have 
$\tynI{\Goals\theta,  \ipair{\theta(\X)}\Nt}{\Nt_i : \theta(\Y_i)}$, for all $i \in I$. The thesis follows by Rule \rn{\itr{In}}, since 
$\theta(\X) = \agtIg\q\pp{i\in I}{\la_i}{\theta(\Y_i)}$ and $\plays{\three{\Goals,\pair\Nt\X}{\eqsys_i}{\Y_i}}\setminus\set\pp = \plays{\Nt'}$ implies 
$\plays{\theta(\Y_i)}\setminus\set\pp= \plays{\Nt'}$ for all $i \in I$ by \refToLem{eq-play}.    \qedhere
\end{description}
\end{proof}

\begin{lem}\label{lem:inf-complete}
If $\tynI{\Nset}{\Nt:\G}$ and   $\plays{\G'}=\plays{\Nt'}$ for all $\pair{\Nt'}{\G'}\in\Nset$,   then, for all $\Goals$, $\X$ and $\sigma$  such that $\X\notin\vars\Goals$, $\dom\sigma = \vars\Goals$ and $\Goals\sigma = \Nset$,  there are $\eqsys$ and $\theta$ such that
$\tyalg{\Goals}{\pair\Nt\X}{\eqsys}$ and $\theta\in\gsol[\Goals]{\eqsys}$ and
$\dom\theta = \vars\eqsys\cup\vars\Goals$ and  
$\sigma\sbtord\theta$ and $\theta(\X) = \G$. 
\end{lem}
\begin{proof}
By induction on the derivation of $\tynI{\Nset}{\Nt:\G}$. 
\begin{description}
\item [\rn{\itr{End}}]
The thesis is immediate by Rule \rn{\infn{End}} taking $\theta = \sigma \sbtplus \set{\X\mapsto\End}$. 
\item [\rn{\itr{Cycle}}]
We have   $\Nset = \Nset',\ipair  \G\Nt $,   then $\Goals = \Goals', \pair\Nt\Y$ and $\sigma(\Y) = \G$. 
By Rule \rn{\infn{Cycle}}, we get $\tyalg{\Goals}{\pair\Nt\X}{\set{\agteq\X\Y}}$, hence 
$\theta = \sigma \sbtplus \set{\X\mapsto \G}$ is a solution of $\set{\agteq\X\Y}$, which agrees with $\Goals$   being $\plays\G=\plays\Nt$,   as needed. 
\item [\rn{\itr{Out}}]
  In this case   we have $\Nt \equiv \pP\pp{\oup\q{i}{I}\la\PP} \parN \Nt'$ and $\G = \agtO\pp\q{i}{I}\la\G$ and
$\tynI{\Nset,\ipair  \G\Nt }{\N_i:\G_i}$ with $\Nt_i \equiv \pP\pp{\PP_i}\parN\Nt'$   and $\plays{\G_i}\setminus\set\pp= \plays{\Nt'}$,   for all $i \in I$.  This last condition implies $\plays\G=\plays\Nt$. 
Set $\sigma' = \sigma \sbtplus \set{\X\mapsto \G}$  and $\Goals'=\Goals, \pair\Nt\X$,  then, by induction hypothesis, we get 
that  there are $\eqsys_i$ and  $\theta_i$ such that  
$\tyalg{\Goals'}{\pair{\Nt_i}{\Y_i}}{\eqsys_i}$ and $\theta_i \in \gsol[\Goals']{\eqsys_i}$ and 
 $\dom{\theta_i}= \vars{\eqsys_i}\cup \vars{\Goals'} $ and   
$\sigma'\sbtord \theta_i$ and $\theta_i(\Y_i) = \G_i$, for all $i \in I$. 
We can assume that   $i\ne j$ implies $Y_i\ne\Y_j$ and $\dom{\eqsys_i}\cap\dom{\eqsys_j} = \emptyset$ for all $i,j\in I$,   because the algorithm always introduces fresh variables. 
  This   implies $\dom{\theta_i}\cap\dom{\theta_j} = \vars{\Goals'}$ for all $i \ne j$, and so $\theta = \sum_{i \in I} \theta_i$ is well defined. 
Moreover, we have $\theta \in \gsol[\Goals']{\eqsys_i}$ and
$\sigma \sbtord \theta$ and $\theta(\X) = \G$, as $\sigma\sbtord \sigma'$ and $\sigma'\sbtord \theta_i \sbtord \theta$ for all $i \in I$. 
  From $\plays{\G_i}\setminus\set\pp= \plays{\Nt'}$ we get $\plays{\three{\Goals'}{\eqsys_i}{\Y_i}}\setminus\set\pp = \plays{\Nt'}$  
 for all $i \in I$ by \refToLem{eq-play}.  
By Rule \rn{\infn{Out}} we get $\tyalg{\Goals}{\pair\Nt\X}{\eqsys}$ with 
$\eqsys = \set{\agteq\X{\agtO\pp\q{i}{I}\la\Y}}\cup\bigcup_{i \in I} \eqsys_i$ and $\theta \in \gsol[\Goals]{\eqsys}$, since 
$\theta(\X) = \agtO\pp\q{i}{I}\la\G = \agtOg\pp\q{i \in I}{\la_i}{\theta_i(\Y_i)} = (\agtO\pp\q{i}{I}\la\Y)\theta$ and 
$\sigma\sbtord\theta$. 
\item [\rn{\itr{In}}]   In this case   we have $\Nt \equiv \pP\pp{\inp\q{j}{J}\la\PP} \parN \Nt'$ and $\G = \agtII\q\pp{i}{I}\la\G$ with $I\subseteq J$ and
$\tynI{\Nset,\ipair\G\Nt}{\N_i : \G_i}$ with $\Nt_i \equiv \pP\pp{\PP_i}\parN\Nt'$   and $\plays{\G_i}\setminus\set\pp= \plays{\Nt'}$,   for all $i \in I$.   This last condition implies $\plays\G=\plays\Nt$. 
Set $\sigma' = \sigma \sbtplus \set{\X\mapsto \G}$  and $\Goals'=\Goals, \pair\Nt\X$,  then, by induction hypothesis,  we get 
 that there are $\eqsys_i$ and $\theta_i$ such that  
$\tyalg{\Goals'}{\pair{\Nt_i}{\Y_i}}{\eqsys_i}$ and $\theta_i \in \gsol[\Goals']{\eqsys_i}$ and 
 $\dom{\theta_i}= \vars{\eqsys_i}\cup \vars{\Goals'} $ and   
$\sigma'\sbtord \theta_i$ and $\theta_i(\Y_i) = \G_i$, for all $i\in I$.  
We can assume that   $i\ne j$ implies $Y_i\ne\Y_j$ and $\dom{\eqsys_i}\cap\dom{\eqsys_j} = \emptyset$ for all $i,j\in I$,   because the algorithm always introduces fresh variables. 
  This    implies $\dom{\theta_i}\cap\dom{\theta_j} = \vars{\Goals'}$ for all $i \ne j$, and so $\theta = \sum_{i \in I} \theta_i$ is well defined. 
Moreover, we have $\theta \in \gsol[\Goals']{\eqsys_i}$, 
$\sigma \sbtord \theta$ and $\theta(\X) = \G$, as $\sigma\sbtord \sigma'$ and $\sigma'\sbtord \theta_i\sbtord \theta$ for all $i \in I$. 
  From $\plays{\G_i}\setminus\set\pp= \plays{\Nt'}$ we get $\plays{\three{\Goals'}{\eqsys_i}{\Y_i}}\setminus\set\pp = \plays{\Nt'}$  
 for all $i \in I$ by \refToLem{eq-play}.  
By Rule \rn{\infn{In}}, we get $\tyalg{\Goals}{\pair\Nt\X}{\eqsys}$ with 
$\eqsys = \set{\agteq\X{\agtII\q\pp{i}{I}\la\Y}}\cup\bigcup_{i \in I} \eqsys_i$ and $\theta \in \gsol[\Goals]{\eqsys}$, since 
$\theta(\X) = \agtII\q\pp{i}{I}\la\G = \agtIg\q\pp{i \in I}{\la_i}{\theta_i(\Y_i)} = (\agtII\q\pp{i}{I}\la\Y)\theta$ and 
$\sigma\sbtord\theta$.  \qedhere
\end{description}
\end{proof}

\begin{thm}[Soundness and completeness of inference]\label{thm:sac}~\begin{enumerate}
\item\label{thm:sac1} If $\tyalg{}{\pair\Nt\X}{\eqsys}$, then $\tyn{\theta(\X)}\Nt$ for all $\theta \in\gsol{\eqsys}$.
 \item\label{thm:sac2} If $\tyn\G\Nt$, then there are $\eqsys$ and $\theta$ such that $\tyalg{}{\pair\Nt\X}{\eqsys}$ and $\theta \in\gsol{\eqsys}$ and $\theta(\X)=\G$. 
\end{enumerate}
\end{thm}
\begin{proof}
(\ref{thm:sac1}). By \refToLemma{lem:inf-sound} $\tyalg{}{\pair\Nt\X}{\eqsys}$ implies $\tynI{}{\Nt:\theta(\X)}$ for all $\theta \in\gsol{\eqsys}$.This is enough, since $\tynI{}{\Nt:\theta(\X)}$ gives $\tyn{\theta(\X)}\Nt$.

(\ref{thm:sac2}). From $\tyn\G\Nt$ we get $\tynI{}{\Nt:\G}$. 
By \refToLemma{lem:inf-complete} this implies that there are $\eqsys$ and $\theta$ such that 
$\tyalg{}{\pair\Nt\X}{\eqsys}$ and $\theta \in\gsol{\eqsys}$ and $\theta(\X)=\G$. 
\end{proof}

\begin{rem}[Termination]
As happens for (co)SLD-resolution in logic programming, the termination of our inference algorithm depends on the choice of a resolution strategy. 
Indeed, we have many sources of non-determinism: 
we have to select a goal to be resolved and, for each of such goals, 
we can either pick a participant of the network and expand it using rules \rn{\infn{In}} or \rn{\infn{Out}}, 
or try to close a cycle using the rule \rn{\infn{Cycle}}. 
A standard way to obtain a sound and complete resolution strategy is to build a tree where all such choices are performed in  parallel and then visit the tree using a breadth-first strategy. 
The tree is potentially infinite in depth, but it is finitely branching, since at each point we have only finitely many different choices, hence this strategy necessarily enumerates all solutions. 
Moreover, by \refToTheorem{th:gtNet}, we know that every network has a type, therefore this strategy necessarily finds a solution, that is, it terminates.  Notice that the same network can be typed by a  weakly  balanced global type and a non  weakly  balanced global type (weak balancing is defined in~\refToFigure{fig:wwfConfig}). For instance the hospital network can be typed by the  weakly  balanced global type of~\refToFigure{fig:exvd} and the non  weakly  balanced global types of~\refToExample{ex:gyN}.  \end{rem}



\section{Taming Types}\label{sec:tt}

In this section we tame global types so that the resulting typing enforces the required properties, i.e. Subject Reduction and Progress. We start by 
 characterising  a 
 class of global types for which,  when looking at global types as trees, the first occurrences of players  are at a bounded depth in all paths starting from arbitrary nodes. 
This property, that we call 
{\em boundedness},  is required by the transition Rules \rn{Inside-Out} and \rn{Inside-In}. It is  also needed for Subject Reduction and Progress. We then pass to formalise a 
condition on 
type configurations
 inspired  again  by the ``inside'' rules of the LTS in \refToFig{ltsgtAs}. 
 We require 
 that for each choice of inputs 
 there should be either an output
or a  message  in the queue matching 
one of the choice labels.  We call this property 
{\em weak balancing}, since it only requires matching for inputs. 
Instead, we call  {\em balanced}  a type configuration in which, in addition to the 
weak balancing,  
each output emitted by
the global type and each message in the queue must have an input  reading it in the type. 
For a session $\Nt\parG\Msg$ whose network is typed by a global type $\G$, boundedness of $\G$ 
and weak balancing of $\G\parG\Msg$ suffice to prove Subject Reduction and the properties 
of deadlock freedom and input lock-freedom. 
To prove orphan-message freedom, as expected, we also need balancing of $\G\parG\Msg$.

\subsection{Boundedness}
To formalise boundedness we use $\ipth$ to denote 
a \emph{path} in global type trees, i.e.,   
a possibly infinite sequence of communications  $\commO{\pp}{\la}{\q}$ or $\commI{\pp}{\la}{\q}$. 
With $\ipth_n$ we represent  
the $n$-th communication in the path $\ipth$, where $0\le n < x$ and $x \in \mathbf{N} \cup \set{\omega}$ is the length of $\ipth$.   By 
$\epsilon$ we denote  the empty sequence and by $\cdot$  the concatenation of a finite sequence with a possibly infinite sequence. 
The function ${\sf Paths}$ gives the set of \emph{paths} of global types, which are the greatest sets such that:

\[\begin{split} 
\IPaths{\End} &= \set{\epsilon}  \\ 
\IPaths{\agtb{\pp}{\q}i I{\la}{\G}} &= \bigcup_{i\in I} \set{ \concat{\commIO\pp{\la_i}\q}{\ipth} \mid \ipth\in\IPaths{\G_i} } 
\end{split}\]  
It is handy to define the \emph{depth} of a player $\pp$ in a global type 
$\G$, $\weight(\G,\pp)$.

\begin{defi}[Depth of a  player]\label{def:depth}
Let $\G$ be a global type. 
For ${\ipth\in\IPaths{\G}}$ set $\weight(\ipth,\pp) = \inf \{ n \mid \play{\ipth_n} = \pp \}$, 
and define $\weight(\G,\pp)$, the \emph{depth} of $\pp$ in $\G$, as follows: 
\[\weight(\G,\pp) = \begin{cases} 
1 + \sup \{ \weight(\ipth,\pp) \mid \ipth \in \IPaths{\G} \}&\text{if } \pp \in \plays{\G} \\ 
0 & \text{otherwise} 
\end{cases}
\]
\end{defi}
Note that $\weight(\G,\pp)=0$ iff  $\pp \not\in \plays{\G}$. 
 Moreover, if  $\pp\ne\play{\ipth_n}$ for some path $\ipth$ and all $n\in\mathbf{N}$,  then $\weight(\ipth,\pp) = \inf\, \emptyset = \infty$. 
Hence, if $\pp$ is a player of a global type $\G$, but it does not occur as a player in some path of $\G$, then  $\weight(\G,\pp) = \infty$. 

We can now define the condition we were looking for.
\begin{defi}[Boundedness]\label{def:bound}
A global type $\G$ is \emph{bounded} if  $\weight(\G',\pp)$ is finite
for all participants $\pp\in\plays{\G'}$ and  all types  
$\G'$ which occur in   $\G$. 
\end{defi}
\begin{exa}\label{b} The following example shows the  necessity  of considering all types  occurring in a global type for defining boundedness. 
Consider 
 $\G= \agtoneO \pr\q {\la}{\agtIS\pr\q {\la};\G'} $,  where
\[\G'=\agtSOS \pp\q {\{\Seq{\la_1}{\agtIS \pp\q \la_1;\agtSOS \q\pr {\la_3};\agtIS\q\pr {\la_3}}\, ,\Seq{\la_2}{\Seq{\agtIS \pp\q \la_2}{\G'}}\}}
\] 
 Then we have: 
\[
 \weight(\G,\pr)=1\quad\quad \weight(\G,\pp)=3\quad\quad\weight(\G,\q)=2
\]
whereas 
\[\weight(\G',\pr)=\infty\quad\quad \weight(\G',\pp)=1\quad\quad\weight(\G',\q)=2
\]
\end{exa}
Since global types are regular the boundedness condition is decidable. 

It is easy to check that the LTS of type configurations preserves the boundedness of global types.
\begin{prop}\label{prop:b}
If $ \G$ is bounded and $\G\parG\Msg\stackred{\beta}\G'\parG\Msg'$, then $\G'$ is bounded.
\end{prop}
\begin{proof} 
By induction on the transition rules of Figure~\ref{fig:ltsgtAs}. 
\end{proof}



\subsection{Balancing}
The {\em weak balancing predicate} of~\refToFig{wwfConfig}  ensures that at least one of the labels of every input choice 
of the global  type is matched either by a message in the queue or by a preceding
output in the global  type.  
Rule  \rulename{WB-Out} removes an output choice from the type 
and puts the corresponding   message  in the queue. This allows a subsequent input choice to be matched
by one of its output messages.   
Rule  \rulename{WB-In} says that every initial input choice should find one of 
its corresponding messages in the queue.  Notice  the similarity between these two rule and Rules \rn{Inside-Out}, \rn{Inside-In} of \refToFig{ltsgtAs}. 

\begin{exa}\label{ex:wb} 
 \refToFig{wbh} displays the initial part of the weakly balancing derivation for the type configuration $\G\parG\emptyset$, where  $\G$, $\G_1$ and $\G_2$ are defined in the caption of \refToFigure{fig:exvd}. 
 
 \begin{figure}[h]
 \begin{math}
 \begin{array}{c}
\NamedCoRule{\rulename{WB-End}}{}{\tupleOKW{\End}{\emptyset}} {}
\qquad
\NamedCoRule{\rulename{WB-Out}}{\tupleOKW{\G_i}{\addMsg\Msg{\mq\pp{\la_i}\q}}\quad \forall i\in I} 
{\tupleOKW{\agtO{\pp}{\q}i I{\la}{\G}}{\Msg}}
{}
\\[3ex]
\NamedCoRule{\rulename{WB-In}}{\tupleOKW{\G_h}{\Msg}} 
{\tupleOKW{\agtI \pp\q \la \G}{\mq\pp{\la_h}\q\cdot\Msg}}  
{h\in I}
 \end{array}
\end{math}
\caption{Weak balancing  of type configurations.} \label{fig:wwfConfig}
\end{figure}

\begin{figure}[b]
\begin{center}
\prooftree
\prooftree
\prooftree
\prooftree
\begin{array}{c}\vdots\\
\tupleOKW{\G}\emptyset
\end{array}
\Justifies
\tupleOKW{\s\cl?\ok;\G}{\mq\s\ok\cl}
\endprooftree
\quad
\prooftree
\prooftree
\prooftree
\prooftree
\DD
\Justifies
\tupleOKW{\G_1}{\mq\cl\lql\s\cdot\mq\cl\hq\s}
\endprooftree
\Justifies
\tupleOKW{\G}{\mq\cl\lql\s}
\endprooftree
\Justifies
\tupleOKW{\cl\s!\lql;\G}{\emptyset}
\endprooftree
\Justifies
\tupleOKW{\s\cl?\ko;\cl\s!\lql;\G}{\mq\s\ko\cl}
\endprooftree
\Justifies
\tupleOKW{\G_2}{\emptyset}
\endprooftree
\Justifies
\tupleOKW{\G_1}{\mq\cl\hq\s}
\endprooftree
\Justifies
\tupleOKW{\G}{\emptyset}
\endprooftree
\end{center}

\bigskip

\raggedright
\qquad 
where 

\begin{center}
$\DD=$\prooftree
\prooftree
\begin{array}{c}\vdots\\
\tupleOKW{\G}{\mq\cl\hq\s}
\end{array}
\Justifies
\tupleOKW{\s\cl?\ok;\G}{\mq\cl\hq\s\cdot\mq\s\ok\cl}
\endprooftree
\quad
\prooftree
\prooftree
\begin{array}{c}\vdots\\
\tupleOKW{\G}{\mq\cl\hq\s\cdot\mq\cl\lql\s}
\end{array}
\Justifies
\tupleOKW{\cl\s!\lql;\G}{\mq\cl\hq\s}
\endprooftree
\Justifies
\tupleOKW{\s\cl?\ko;\cl\s!\lql;\G}{\mq\cl\hq\s\cdot\mq\s\ko\cl}
\endprooftree
\Justifies
\tupleOKW{\G_2}{\mq\cl\hq\s}
\endprooftree
\end{center}
\caption{Weak balancing of the hospital global type with the empty queue.}\label{fig:wbh}
\end{figure}
\end{exa}

The following proposition shows that reduction of global types preserves weak balancing.

\begin{prop}\label{prop:IO}
If $\tupleOKW{ \G}{\Msg}$ and $\G\parG\Msg\stackred{\beta}\G'\parG\Msg'$, then $\tupleOKW{ \G'}{\Msg'}$.
\end{prop}
\begin{proof} 
By induction on the transition rules of Figure~\ref{fig:ltsgtAs}.
\end{proof}

With boundedness of global types and weak balancing of type configurations we can prove 
the expected Subject Reduction Theorem.

\begin{thm}[Subject Reduction]\label{thm:srA}
If $\tyn{\G}\Nt$ and $\Nt\parN\Msg\stackred\beta\Nt'\parN\Msg'$ and $\G$ is bounded and $\tupleOKW\G\Msg$, then $\G\parG\Msg\stackred\beta\G'\parG\Msg'$ and $ \tyn{\G'}{\Nt'}$.
\end{thm}
\begin{proof} The proof is by induction on $d=\weight(\G,\pp)$ where $\pp=\play\beta$. Notice that $\Nt\parN\Msg\stackred\beta\Nt'\parN\Msg'$ implies $\pp\in\plays\Nt$, which with  $\tyn{\G}\Nt$  gives  $\pp\in\plays\G$ by \refToLemma{ep}. 
Then $d> 0$. Moreover $d$ is finite since  $\G$ is bounded.\\
{\em Case $d=1$.} If $\beta=\q\pp?\la$, then $\Nt\equiv\pP\pp{\inp{\q}{i}{I}{\la}{\PP}}\parN\Nt_0$ and $\Nt'=\pP\pp{\PP_h}\parN\Nt_0$ and $\Msg\equiv\mq\q\la\pp\cdot\Msg'$ and $\la=\la_h$ for some  $h\in I$. From $d=1$ we get $\G=\agtJ\q\pp\la\G$ and  $\tyn{\G_j}{\pP\pp{\PP_j}\parN\Nt_0}$ for all $j \in J$ with $J\subseteq I$. From $\Msg\equiv\mq\q\la\pp\cdot\Msg'$ and $\tupleOKW\G\Msg$ we get $h\in J$. We conclude $\G\parG\Msg\stackred\beta\G_h\parG\Msg'$. If $\beta=\pp\q!\la$ the proof is similar and simpler.\\
{\em Case $d>1$.}Then either  $\G=\agtO{\pr}{\q}i I{\la}{\G}$ or $\G=\agtI{\q}{\pr}{\la}{\G}$ with $\pr\ne\pp=\play\beta$. Let $\G=\agtO{\pr}{\q}i I{\la}{\G}$, the proof for $\G=\agtI{\q}{\pr}{\la}{\G}$ is similar. Since $\tyn{\G}\Nt$ must be derived using Rule \rn{Out} we get $\Nt\equiv\pP\pr{\oup{\q}{i}{I}{\la}{\R}}\parN\Nt_0$ and $ \tyn{\G_i}{\pP\pr{\R_i}\parN\Nt_0}$ for all $i\in I$. The condition $\pr\ne\play{\asCom}$ ensures that the reduction $\Nt\parN\Msg\stackred\beta\Nt'\parN\Msg'$ does not modify the process of participant $\pr$, then we get $\Nt'\equiv\pP\pr{\oup{\q}{i}{I}{\la}{\R}}\parN\Nt'_0$.
  Moreover,  the reduction can be done also if messages are added 
at the end of the queue. 
Therefore 
\[
\pP\pr{\R_i}\parN\Nt_0\parG\Msg\cdot\mq\pr{\la_i}\q\stackred\beta  \pP\pr{\R_i}\parN\Nt_0'\parG\Msg'\cdot\mq\pr{\la_i}\q
\]
for all $i \in I$. 
From $\G$ bounded we have 
$\G_i$ bounded and from $\tupleOKW\G\Msg$ we have 
$\tupleOKW{\G_i}{\Msg\cdot\mq\pr{\la_i}\q}$ by Rule \rulename{WB-Out}   for all $i \in I$. Since $\weight(\G_i,\pp)< d$  for all $i \in I$, by induction we get $\G_i\parG\Msg\cdot\mq\pr{\la_i}\q\stackred\beta\G_i'\parG\Msg'\cdot\mq\pr{\la_i}\q$ and $ \tyn{\G'_i}{\pP\pr{\R_i}\parN\Nt'_0}$  for some $\G'_i$ and  for all $i\in I$. We  get $\G\parG\Msg\stackred\beta\agtO{\pr}{\q}i I{\la}{\G'}\parG\Msg'$ using Rule \rulename{Inside-Out}. 
By \refToLemma{ep} $\tyn{\G'_i}{\pP\pr{\R_i}\parN\Nt_0'}$ implies $\plays{\G'_i}\setminus\set\pr=\plays{\Nt'_0}$ for all $i\in I$. Then we can derive $\tyn{\agtO{\pr}{\q}i I{\la}{\G'}}{\Nt'}$ using Rule \rulename{Out}. 
\end{proof}

Deadlock-freedom easily follows from Subject Reduction and Session Fidelity.

\begin{thm}[Deadlock-freedom]\label{thm:df}
If $ \tyn{\G}\Nt$  and $\G$ is bounded and $\tupleOKW\G\Msg$, then $\Nt\parN\Msg$ is deadlock free.  
\end{thm}
\begin{proof} By Subject Reduction (\refToTheorem{thm:srA}) and Session Fidelity (\refToTheorem{thm:sfA}) it is enough to show that $\G\parN\Msg$ is deadlock-free. One of the two Rules \rn{Top-Out} and \rn{Top-In} is applicable whenever $\tupleOKW\G\Msg$ holds and $\G$ is not $\End$. 
\end{proof} 

Our next goal is to show that no player offering a choice of inputs can wait forever. 
To this aim we prove Subject Reduction for the lockstep transition relation. We dub  this result  Strong Subject Reduction.

\begin{thm}[Strong Subject Reduction] \label{thm:ssr}
 If $ \tyn{\G}\Nt$  and $\G$ is bounded and $\tupleOKW\G\Msg$ and $\Nt\parG\Msg \lockred{\Delta} \Nt'\parG\Msg'$, then $\G\parN\Msg \lockred{\Delta} \G'\parN\Msg'$
  and $\tyn{\G'}{\Nt'}$.  
\end{thm}
\begin{proof} 
Let $\Delta = \{\co_1,\ldots,\co_n\}$. 
By definition of $\lockred\Delta$,  the set  $\Delta$ is coherent for $\Nt\parN\Msg$, that is, we know that, for all $i\in 1..n$, 
$\Nt\parN\Msg \stackred{\co_i}$, hence, by \refToThm{srA}, we get, for all $i\in 1..n$, 
$\G\parG\Msg \stackred{\co_i}$, namely $\Delta$ is coherent for $\G\parG\Msg$. 
Consider a coherent set $\Delta'$ for $\G\parG\Msg$ such that $\Delta\subseteq \Delta'$, then for all  $\co\in\Delta'$, 
by Session Fidelity (\refToThm{sfA}), we get $\Nt\parN\Msg \stackred{\co}$. 
Hence $\Delta'$ is coherent for $\Nt\parN\Msg$ and, 
since $\Delta$ is a maximal coherent set for $\Nt\parN\Msg$ again by definition of $\lockred\Delta$, we get $\Delta = \Delta'$, that is, 
$\Delta$ is a maximal coherent set for $\G\parG\Msg$.

Finally, since $\Nt\parN\Msg\lockred\Delta \Nt'\parN\Msg'$, we know that, for all $i\in 1..n$, 
$\Nt_{i-1}\parN\Msg_{i-1} \stackred{\co_i} \Nt_i\parN\Msg_i$, with $\Nt_0 = \Nt$, $\Msg_0 = \Msg$, $\Nt_n = \Nt'$ and $\Msg_n = \Msg'$.
By Subject Reduction (\refToThm{srA}) and Propositions~\ref{prop:b}, \ref{prop:IO}, we get, for all $i\in 1..n$, 
$\G_{i-1}\parG\Msg_{i-1}\stackred{\co_i} \G_i\parG\Msg_i$ with $\G_0 = \G$, $\tupleOKW{\G_i}{\Msg_i}$  and $\tyn{\G_i}{\Nt_i}$. 
 Hence,  setting $\G' = \G_n$, we get 
$\G\parG\Msg \lockred{\Delta} \G'\parG\Msg'$ and $\tyn{\G'}{\Nt'}$.  
\end{proof}

We need a lemma showing that a communication labelling a transition of a type configuration with global type $\G$ must occur in all  the paths of  the tree representing $\G$. 

\begin{lem}\label{lem:step-path}
If $\G\parG\Msg \stackred{\co}\G'\parG\Msg'$, then for all paths 
$\ipth \in \IPaths{\G'}$, $\ipth = \ipth_1\cdot\ipth_2$ and $\ipth_1\cdot\co\cdot\ipth_2\in\IPaths{\G}$, 
for some $\ipth_1,\ipth_2$. 
\end{lem}
\begin{proof}
By induction on  the  rules defining the transition relation,  see \refToFigure{fig:ltsgtAs}. 
\end{proof}

A last lemma before establishing the property of  ``input-enabling'' relates the depth of players in global types with the reductions of the corresponding type configurations.

\begin{lem}\label{lem:w-play-dec}
Let $\G$ be bounded and  $\tupleOKW\G\Msg$  and $\pp \in \plays\G$, then the following hold:
\begin{enumerate}
\item\label{lem:w-play-dec:1} If $\weight(\G,\pp) = 1$, then $\G\parG\Msg\stackred{\co}\G'\parG\Msg'$, with $\play{\co} = \pp$. 
\item\label{lem:w-play-dec:2} If $\G\parG\Msg\stackred{\co}\G'\parG\Msg'$, with  $\pp\ne \play{\co}$, then  $\weight(\G',\pp)\le\weight(\G,\pp)$. 
\item\label{lem:w-play-dec:3} If $\G\parG\Msg\lockred{\Delta}\G'\parG\Msg'$, with   $\pp\ne \play{\co}$ for all $\co\in\Delta$,  then $\weight(\G',\pp) < \weight(\G,\pp)$. 
\end{enumerate}
\end{lem}
\begin{proof}
(\ref{lem:w-play-dec:1}). If $\weight(\G,\pp) = 1$ there are two possibilities: 
either $\G = \agtO{\pp}{\q}{i}{I}{\la}{\G}$ or $\G = \agtI{\q}{\pp}{\la}{\G'}$. In the first case,  $\G\parG\Msg \stackred{\commO{\pp}{\la_h}{\q}} \G_h\parG\Msg'$, for some $h\in I$, by Rule \rn{Top-Out}. 
In the second case,    from $\G\parG\Msg$  weakly balanced  we get 
$\Msg \equiv \addMsg{\mq{\q}{\la_h}{\pp}}{\Msg'}$ and    $\G\parG\Msg \stackred{\commI{\q}{\la_h}{\pp}} \G_h\parG\Msg'$, for some $h\in I$, by Rule \rn{Top-In}.  

(\ref{lem:w-play-dec:2}).  
Let $\ipth\in\IPaths{\G'}$, then by \refToLem{step-path}, we have $\ipth = \ipth_1\cdot\ipth_2$ and $\ipth_1\cdot\co\cdot\ipth_2\in\IPaths{\G}$, for some $\ipth_1,\ipth_2$. 
Since $\G$ is bounded, $\weight(\G,\pp)$ is finite, hence there is a first occurrence of $\co'$ with $\play{\co'} = \pp$ either in $\ipth_1$ or in $\ipth_2$, because $\play\co \ne \pp$. 
In the former case we have  $\weight(\ipth,\pp) =\weight(\ipth_1\cdot\co\cdot\ipth_2,\pp) \le \weight(\G,\pp)$, and, in the latter case, 
we have $\weight(\ipth,\pp) < \weight(\ipth_1\cdot\co\cdot\ipth_2,\pp) \le \weight(\G,\pp)$. 
This implies
\[\weight(\G',\pp) = \sup\{\weight(\ipth',\pp) \mid \ipth'\in\IPaths{\G'}\} \le \weight(\G,\pp)\] as needed. 

(\ref{lem:w-play-dec:3}).  We have two cases on $\G$. 
\begin{itemize}
\item If $\G = \agtO{\q}{\pr}{i}{I}{\la}{\G}$, then,     by Rule \rn{Top-Out},
$\G\parG\Msg \stackred{\co} \widehat\G\parG\widehat\Msg$ where   $\co = \commO{\q}{\la_h}{\pr} \in \Delta$  and $\widehat\G = \G_h$ for some $h \in I$.  From $\Delta$ maximal we get $\co\in \Delta$.   
Then, for all paths $\ipth\in\IPaths{\widehat\G}$, we have    $\co\cdot\ipth\in\IPaths{\G}$. Since     $\weight(\G,\pp)$ is finite, because $\G$ is bounded, and $\pp\ne\play{\co}$,  being  
 $\pp\ne\play{\co'}$ for all $\co'\in\Delta$,  there is in $\ipth$ a first occurrence   of  $\co_0$ with $\play{\co_0} = \pp$.
Thus we get 
$\weight(\ipth,\pp) <   \weight(\co\cdot\ipth,\pp)$. 
This implies $\weight(\widehat\G,\pp) < \weight(\G,\pp)$. 
\item If $\G = \agtI{\q}{\pr}{\la}{\G}$, then  $\tupleOKW\G\Msg$ implies
$\Msg \equiv \addMsg{\mq\q{\la_h}\pr}{\widehat\Msg}$ for some $h\in I$.
By  Rule \rn{Top-In},
$\G\parG\Msg \stackred{\co} \widehat\G\parG\widehat\Msg$ where  $\co = \commI{\q}{\la_h}{\pr}$. 
From $\Delta$ maximal we get $\co\in \Delta$. 
Then, for all paths $\ipth\in\IPaths{\widehat\G}$, we have $\co\cdot\ipth\in\IPaths{\G}$.  Since $\weight(\G,\pp)$ is finite, because $\G$ is bounded, and $\pp\ne\play{\co}$,  being   
 $\pp\ne\play{\co'}$ for all $\co'\in\Delta$,    
there is in $\ipth$ a first occurrence   of  $\co_0$ with $\play{\co_0} = \pp$,  thus we get 
$\weight(\ipth,\pp) <   \weight(\co\cdot\ipth,\pp)$. 
This implies $\weight(\widehat\G,\pp) < \weight(\G,\pp)$. 
\end{itemize}
In both cases we have $\G\parG\Msg\stackred{\co}\widehat\G\parG\widehat\Msg$ with $\pp\ne\play\co$ and $\weight(\widehat\G,\pp) < \weight(\G,\pp)$. 
If   $\Delta=\set\co$ we are done. Otherwise  we set $\Delta\setminus\{\co\} = \{\co_1,\ldots,\co_n\}$, where $n\ge 1$, and we get 
\[\widehat\G\parG\widehat\Msg\stackred{\co_1}\G_1'\parG\Msg_1\stackred{\co_2}\ldots\stackred{\co_n}\G_n'\parG\Msg_n\] with $\G_n' = \G'$ and $\Msg_n = \Msg'$. 
By Item~(\ref{lem:w-play-dec:2}), we get $\weight(\G'_k,\pp) \le \weight(\G'_{k-1},\pp)$ for all $k \in 2..n$ and $\weight(\G'_1,\pp) \le \weight(\widehat\G,\pp)$ and this implies $\weight(\G',\pp) \le \weight(\widehat\G,\pp) < \weight(\G,\pp)$ as needed. 
\end{proof}

\begin{thm}[No locked inputs]\label{thm:nli}
If $ \tyn{\G}\Nt$  and $\G$ is bounded and $\tupleOKW\G\Msg$, then $\Nt\parN\Msg$ is input-enabling.  
\end{thm}
\begin{proof} 
If $\Nt\equiv\pP\pp{\inp\q{i}{I}{\la}{\PP}}\parN\Nt_0$, then $\tyn{\G}\Nt$ implies $\pp\in\plays\G$.
By Strong Subject Reduction  (\refToTheorem{thm:ssr}) all complete lockstep computations of $\Nt\parN\Msg$ are complete lockstep computations of $\G\parN\Msg$. 
Then it is enough to show that $\G\parN\Msg$  is input-enabling. 

We prove the  statement  by induction on $\weight(\G,\pp)$. 
Consider a complete lockstep computation $\set{\G_k\parG\Msg_k \lockredsub{\Delta_k}{k}\G_{k+1}\parG\Msg_{k+1}}_{k<x}$ with $\G_0\parG\Msg_0 = \G\parG\Msg$. 
If $\weight(\G,\pp)=1$ by \refToLemma{lem:w-play-dec}(\ref{lem:w-play-dec:1}) we get $\G\parG\Msg\stackred{\co}\G'\parG\Msg'$, with $\play{\co} = \pp$, hence, 
since $\Delta_0$ is a maximal coherent set, there is $\co' \in \Delta_0$ such that $\play{\co'} = \play{\co} = \pp$, as needed. 
If $\weight(\G,\pp)>1$, let us assume  $\pp \ne \play{\co}$ for all  $\co \in \Delta_0$  
(otherwise the   statement  is trivial), 
then, since  $\G\parG\Msg = \G_0\parG\Msg_0 \lockred{\Delta_0}\G_1\parG\Msg_1$,  by \refToLemma{lem:w-play-dec}(\ref{lem:w-play-dec:3}), 
we get $\weight(\G_1,\pp) < \weight(\G,\pp)$. 
Hence, the   statement  follows by induction hypothesis. 
\end{proof}

To avoid messages waiting forever in the queue we need to  strengthen  weak balancing with the requirement that 
all messages in the queue will
be eventually read.  
This last restriction is enforced by the auxiliary judgment 
$\algread{\G}{\Msg}$, which means that (each path of) {\em $\G$ reads all the messages in
$\Msg$}.
The inductive definition of this judgment is given at the bottom of \refToFigure{fig:wfConfig}. 
If the queue is empty, Rule \rn{Empty-R}, then the judgement holds.  If $\G$ is a choice of outputs, 
Rule \rn{Out-R}, then in each branch of the choice all the messages  in the queue  
must be read. For an
input type $\agtI{\pp}{\q}{\la}{\G}$, if the message at the top of the queue (considered modulo $\equiv$)  is    
$\mq\pp{\la_h}\q$ for some $h\in I$ we read it, Rule \rn{In-R1}, otherwise we  do 
not to read 
messages,  Rule \rn{In-R2}.  In Rule \rn{In-R1} we could check only $\algread{\G_h}{\Msg}$ since the other branches are discharged.  Instead we required $\algread{\G_i}{\Msg}$ for all $i\in I$. This is needed  to ensure  that a message on top of a queue in a balanced type configuration has a finite weight for the global type (\refToDef{def:w}), see \refToExample{aw}.

\begin{figure}
 \begin{math}
 \begin{array}{c}
\NamedCoRule{\rulename{B-End}}{}{\tupleOK{\End}{\emptyset}} {}
\\[3ex]
\NamedCoRule{\rulename{B-Out}}{\tupleOK{\G_i}{\addMsg\Msg{\mq\pp{\la_i}\q}}\quad \forall i\in I} 
{\tupleOK{\agtO{\pp}{\q}i I{\la}{\G}}{\Msg}}
{ \algread{\G_i}{\Msg}\  \forall i\in I}
\\[3ex]
\NamedCoRule{\rulename{B-In}}{\tupleOK{\G_h}{\Msg}} 
{\tupleOK{\agtI \pp\q \la \G}{\addMsg{\mq\pp{\la_h}\q}\Msg}}  
{ \algread{\G_i}{\Msg}\  \forall i\in I}
\\[5ex]
\NamedRule{\rulename{Empty-R}}{}{\algread{\G}{\emptyset}}{}
\qquad 
\NamedRule{\rulename{Out-R}}
{ \algread{\G_i}{\Msg}\ \forall i \in I }
{ \algread{\agtO{\pp}{\q}{i}{I}{\la}{\G}}{\Msg} }{}
\\[3ex]
\NamedRule{\rulename{In-R1}}
{ \algread{\G_i}{\Msg}\ \forall i \in I} 
{ \algread{\agtI{\pp}{\q}{\la}{\G}}{\addMsg{\mq\pp{\la_h}\q}{\Msg} } }
{ h\in I }
\\[3ex] 
\NamedRule{\rulename{In-R2}}
{ \algread{\G_i}{\Msg}\ \forall i\in I }
{ \algread{\agtI{\pp}{\q}{\la}{\G}}{\Msg}}
{} 
{\Msg\not\equiv\addMsg{\mq\pp{\la_i}\q}{\Msg'}\ \forall i \in I } 
 \end{array}
\end{math}\caption{Balancing  of type configurations.} \label{fig:wfConfig}
\end{figure}

\begin{exa}\label{ex:b} 
 \refToFig{wbh} displays the initial part of the weak balancing derivation for the type configuration $\G\parG\emptyset$, where   $\G$, $\G_1$ and $\G_2$ are defined in the caption of \refToFigure{fig:exvd}.  We can easily get a balancing derivation by checking the  readability  side-conditions. \refToFig{rh} only shows the more interesting case.
\begin{figure}
\begin{center}
\prooftree
\prooftree
\algread{\G_2}\emptyset
\justifies
\prooftree
\algread{\G_1}{\mq\cl\lql\s}
\justifies
\prooftree
\algread{\G}{\mq\cl\lql\s}
\justifies
\algread{\s\cl?\ok;\G}{\mq\cl\lql\s}
\endprooftree
\endprooftree
\endprooftree
\qquad 
\prooftree
 
\prooftree
{
\prooftree
\algread{\G_2}{\emptyset}
\justifies
\algread{\G_1}{\mq\cl\lql\s}
\endprooftree
}
\justifies
\algread{\G}{\mq\cl\lql\s}
\endprooftree

\justifies
\prooftree
\algread{\cl\s!\lql;\G}{\mq\cl\lql\s}
\justifies
\algread{\s\cl?\ko;\cl\s!\lql;\G}{\mq\cl\lql\s}
\endprooftree
\endprooftree
\justifies
\prooftree
\algread{\G_2}{\mq\cl\lql\s}
\justifies
\prooftree
\algread{\G_1}{\mq\cl\hq\s\cdot\mq\cl\lql\s}
\justifies
\algread{\G}{\mq\cl\hq\s\cdot\mq\cl\lql\s}
\endprooftree
\endprooftree
\endprooftree
\end{center}
\caption{Read of the hospital global type with the queue $\mq\cl\hq\s\cdot\mq\cl\lql\s$.}\label{fig:rh}
\end{figure}

\end{exa}

Comparing Figures~\ref{fig:wwfConfig} and~\ref{fig:wfConfig}  we can see  that the balancing is obtained from the weak balancing just by adding the 
readability  judgment as side condition of rules \rn{B-Out} and \rn{B-In}, checking that \emph{all} messages in the queue can be read, even in case of infinite derivations, see \refToExample{ex:r}. The necessity of the readability  judgment also in Rule \rn{B-In} is shown in \refToExample{aw}. Notice that  this rule by itself only ensures that the first message of the queue is read.
\begin{exa}\label{ex:r}
Let $\G=\Seq{\commO\pp\la\q}{\Seq{\commI\pp\la\q}{\G}}$ and $\Msg={\mq\pp{\la'}\pr}$. Then $\G\parN\Msg$  is weakly  balanced but not balanced, since $\algread{\G}\Msg$ does not hold. 
In fact, $\mq\pp{\la'}\pr$ would never be read.
\end{exa}

The regularity of global types  and the fact that the initial queue may only decrease  ensure the decidability of the  readability  judgment, 
while we will show in the next section that both weak balancing and balancing are undecidable.

We use $\ms$ as short for a message of the form $\mq\pp\la\q$ and 
let $\inc{\mq\pp{\la}\q}=\commI\pp\la\q$,  i.e. $\inc{\mq\pp{\la}\q}$ is the input communication which reads the message $\mq\pp{\la}\q$.  
We define  the weight of a message in a global type, notation $\wgs{\ms}{\G}$, by looking at the corresponding input. 
\begin{defi}[Weight of messages]\label{def:w}
For a global type $\G$ we define $\wgs{\ms}{\G}$ corecursively as follows: 
$\wgs\ms\G = $ 
\[
\begin{cases}
\infty & \text{if $\G = \End$} \\ 
1 + \sup_{i\in I} \wgs\ms{\G_i} & \text{if either $\G = \agtO{\pp}{\q}{i}{I}{\la}{\G}$}\\
&\quad \text{or $\G = \agtII{\pp}{\q}{i}{I}{\la}{\G}$ and $\inc\ms \ne \commI\pp{\la_i}\q$ for all $i \in I$}\\ 
0      & \text{if $\G = \agtII{\pp}{\q}{i}{I}{\la}{\G}$ and $\inc\ms = \commI\pp{\la_h}\q$ for some $h \in I$} \\ 
\end{cases}
\]
\end{defi}
Note that $\wgs\ms\G$ is well defined as every recursive call is guarded.   We get $\wgs\ms\G=\infty$ when the global type $\G$ never reads the message $\ms$. For example, if $\G$ only contains outputs, then  $\wgs\ms\G=\infty$ for all messages $\ms$. 

The key point is that in a balanced type configuration the message on top of the queue has a finite weight, as proved by the following lemma. 
\begin{lem}\label{lem:aux-wg}
If $\tupleOK{\G}{\addMsg\ms\Msg}$, then $\wgs\ms\G$ is finite. 
\end{lem}
\begin{proof}
 By  definition 
$\tupleOK{\G}{\addMsg\ms\Msg}$ implies $\algread{\G}{\addMsg\ms\Msg}$ (see rules in \refToFig{wfConfig}). 
The proof is by induction on the derivation of $\algread{\G}{\addMsg\ms\Msg}$, splitting cases on the last applied rule. 
\begin{description}
\item [\rn{Empty-R}] This is an empty case, because the queue is not empty by hypothesis. 
\item [\rn{Out-R}] We have $\G = \agtO{\pp}{\q}{i}{I}{\la}{\G}$ and $\algread{\G_i}{\addMsg\ms\Msg}$ for all $i \in I$. 
By induction hypothesis we get that $\wgs\ms{\G_i}$ is finite for all $i\in I$, hence 
$\wgs\ms\G = 1 + \sup_{i \in I} \wgs\ms{\G_i}$ is finite as well. 
\item [\rn{In-R1}] We have $\G = \agtII{\pp}{\q}{i}{I}{\la}{\G}$ and $\addMsg\ms\Msg \equiv \addMsg{\mq\pp{\la_h}\q}{\Msg'}$ with $h \in I$ and \mbox{$\algread{\G_i}{\Msg'}$} for all $i \in I$. 
If $\ms = \mq\pp{\la_h}\q$, then $\wgs\ms\G = 0$ by definition. Otherwise, 
$\Msg' \equiv \addMsg\ms{\Msg''}$, hence, by induction hypothesis,  $\wgs\ms{\G_i}$ is finite for all $i \in I$ and so 
$\wgs\ms\G = 1 + \sup_{i \in I} \wgs\ms{\G_i}$  is finite as well. 
\item [\rn{In-R2}] We have $\G = \agtII{\pp}{\q}{i}{I}{\la}{\G}$ and $\addMsg\ms\Msg \not\equiv \addMsg{\mq\pp{\la_i}\q}{\Msg'}$ and $\algread{\G_i}{\addMsg\ms\Msg}$ for all $i \in I$. 
Therefore, $\ms \ne \mq\pp{\la_i}\q$ for all $i \in I$ and, by induction hypothesis, $\wgs\ms{\G_i}$ is finite for all $i \in I$, hence, 
$\wgs\ms\G = 1 + \sup_{i \in I} \wgs\ms{\G_i}$ is finite as well. 
\qedhere 
\end{description}
\end{proof}

\begin{exa}\label{aw} 
\refToLemma{lem:aux-wg} fails if either Rule \rn{B-In} does not require 
the readability judgment as side condition or in the premise of Rule \rn{In-R1} we only consider the global type $\G_h$.
A simple example  is  
$\G=\agtIS{\pp}{\q}{\set{\la_1;\commI{\pp}{\la_2}{\pr}\, ,\, \la_3}}$  
and $\Msg=\mq\pp{\la_1}\q \cdot\mq\pp{\la_2}\pr$. 
In fact, in both cases we would get $\tupleOK{\G}{\Msg}$ (derived from $\algread{\G}\Msg$ in the second case), but    $\wgs{\mq\pp{\la_2}\pr}\G=\infty$.
\end{exa}

A last lemma before establishing the property of  ``No orphan message'' relates the weight  of messages in global types with the reductions of the corresponding type configurations.

\begin{lem} \label{lem:w-msg-dec}
Let  $\tupleOK{\G}{\Msg}$  and $\Msg \equiv \addMsg\ms{\Msg'}$, then the following hold: 
\begin{enumerate}
\item\label{lem:w-msg-dec:1} If $\wgs\ms \G = 0$, then $\G\parG\Msg\stackred{\inc\ms}\G'\parG\Msg'$. 
\item\label{lem:w-msg-dec:2} If $\G\parG\Msg\stackred{\co}\widehat{\G}\parG\widehat{\Msg}$ with  $\inc\ms \ne \co $, then  $\wgs\ms {\widehat{\G}} \le \wgs{\ms}{\G}$ and $\widehat{\Msg} \equiv \addMsg\ms{\widehat{\Msg}'}$.
\item\label{lem:w-msg-dec:3}  If $\G\parG\Msg\lockred{\Delta}\widehat{\G}\parG\widehat{\Msg}$, with  $\inc\ms \notin \Delta$, then  $\wgs\ms {\widehat{\G}} < \wgs{\ms}{\G}$ and $\widehat{\Msg} \equiv \addMsg\ms{\widehat{\Msg}'}$.
\end{enumerate}
\end{lem}
\begin{proof}
(\ref{lem:w-msg-dec:1}). Assume $\ms = \mq\pp\la\q$, hence $\inc\ms = \commI{\pp}{\la}{\q}$. 
If $\wgs\ms\G = 0$, then $\G = \agtI{\pp}{\q}{\la}{\G}$  and  $\la=\la_h$ and by Rule \rn{Top-In} we have  $\G\parG\Msg \stackred{\commI{\pp}{\la}{\q}}\G_h\parG\Msg'$  for some $h\in I$.  

(\ref{lem:w-msg-dec:2}). The fact that $\widehat{\Msg} \equiv \addMsg\ms{\widehat{\Msg}'}$ is immediate, since $\inc\ms \ne \co$. 
 We prove $\wgs\ms {\widehat{\G}} \le \wgs{\ms}{\G}$ by induction on reduction rules for type configurations.
\begin{description}
\item[\rulename{Top-Out}] Then $\G=\agtO{\pp}{\q}i I{\la}{\G}$  and $\widehat\Msg=\addMsg{\Msg}{\mq\pp{\la_h}\q}$ and $\widehat\G=\G_h$ for some $h\in I$.  By \refToDef{def:w} $\wgs{\ms}{\G}=1 + \sup_{i\in I} \wgs\ms{\G_i} > \wgs{\ms}{\G_h}$. 
\item[\rulename{Top-In}] Then $\G=\agtI \pp\q \la \G$  and $\Msg=\addMsg{\mq\pp{\la_h}\q}{\widehat\Msg}$ and $\widehat\G=\G_h$ for some $h\in I$.  By \refToDef{def:w} $\wgs{\ms}{\G}=1 + \sup_{i\in I} \wgs\ms{\G_i} > \wgs{\ms}{\G_h}$. 
\item[\rulename{Inside-Out}] Then $\G=\agtO{\pp}{\q}i I{\la}{\G}$ and $\widehat\G=\agtO{\pp}{\q}i I{\la}{\G'}$ and 
$\G_i\parG \Msg\cdot\mq\pp{\la_i}\q \stackred\asCom\G'_i \parG \widehat\Msg\cdot\mq\pp{\la_i}\q $
 for all $i \in I$.  By \refToDef{def:w} $\wgs{\ms}{\G}=1 + \sup_{i\in I} \wgs\ms{\G_i}$ and  $\wgs{\ms}{\widehat\G}=1 + \sup_{i\in I} \wgs\ms{\G'_i}$. By induction 
$\wgs\ms {\G'_i} \le \wgs{\ms}{\G_i}$  for all $i \in I$ and this concludes the proof. 
\item[\rulename{Inside-In}] Then $\G=\agtI \pp\q \la \G$  and  $\widehat\G=\agtII \pp\q i I\la {\G'}$  and 
 $\Msg\equiv\mq\pp{\la_k}\q \cdot\Msg''$ and $\widehat\Msg\equiv\mq\pp{\la_k}\q \cdot \widehat\Msg'$ for some $k\in I$ and
$\G_i\parG \Msg''\stackred\asCom\G'_i \parG \widehat\Msg' $     
for all $i \in I$. If $\ms=\mq\pp{\la_k}\q$, then $\wgs\ms \G = \wgs\ms{\widehat \G} =0$. Otherwise $\wgs{\ms}{\G}=1 + \sup_{i\in I} \wgs\ms{\G_i}$ and  $\wgs{\ms}{\widehat\G}=1 + \sup_{i\in I} \wgs\ms{\G'_i}$. By induction 
$\wgs\ms {\G'_i} \le \wgs{\ms}{\G_i}$  for all $i \in I$ and this concludes the proof. \end{description}

(\ref{lem:w-msg-dec:3}).  
Let $\Delta = \set{\co_1,\ldots,\co_n}$, hence by definition of lockstep transition we have 
\[\G_{k-1}\parG\Msg_{k-1} \stackred{\co_k} \G_k\parG\Msg_k\] for all $k \in 1..n$, with 
$\G\parG\Msg = \G_0\parG\Msg_0$ and $\widehat\G\parG\widehat\Msg = \G_n\parG\Msg_n$. 
By Item ~(\ref{lem:w-msg-dec:2}), for all $k \in 1..n$, we have 
$\wgs\ms{\G_k} \le \wgs\ms{\G_{k-1}}$ and $\Msg_k\equiv \addMsg\ms{\Msg'_k}$, hence
$\widehat\Msg \equiv \addMsg\ms{\widehat\Msg'}$. 
We have two cases on $\G$.
 \begin{itemize}
\item If $\G = \agtO{\pp}{\q}{i}{I}{\la}{\G'}$, as $\Delta$ is maximal, we have $\co_j = \commO{\pp}{\la_h}{\q}$ for some $j \in 1..n$ and $h \in I$,   and so 
$\G_{j-1} = \agtO{\pp}{\q}{i}{I}{\la}{\G''}$ and 
the step $\G_{j-1}\parG\Msg_{j-1}\stackred{\co_j}\G_j\parG\Msg_j$ is derived by Rule \rn{Top-Out}, hence 
$\G_j = \G''_h$. 
Then, since $\wgs\ms{\G_{j-1}} = 1 + \sup_{i\in I} \wgs\ms{\G''_i}$ and it is finite, we get 
$\wgs\ms{\G_j} < \wgs\ms{\G_{j-1}}$ and this proves the  statement. 
\item If $\G = \agtI{\pp}{\q}{\la}{\G'}$, 
as $\tupleOK{\G}{\addMsg\ms\Msg}$ holds, we have 
$\addMsg\ms\Msg \equiv \addMsg{\mq\pp{\la_h}\q}{\Msg'}$ for some $h \in I$. 
Since $\Delta$ is maximal, we have   $\co_j = \commI{\pp}{\la_h}{\q}$ for some $j \in 1..n$ and so 
$\G_{j-1} = \agtII{\pp}{\q}{i}{I}{\la}{\G''}$ and 
the step $\G_{j-1}\parG\Msg_{j-1}\stackred{\co_j}\G_j\parG\Msg_j$ is derived by Rule \rn{Top-In}, hence 
$\G_j = \G''_h$. 
Then, since $\wgs\ms{\G_{j-1}} = 1 + \sup_{i\in I} \wgs\ms{\G''_i}$ and it is finite, we get 
$\wgs\ms{\G_j} < \wgs\ms{\G_{j-1}}$ and this proves the  statement.  \qedhere
\end{itemize}
\end{proof}

\begin{thm}[No orphan messages]\label{thm:nom}
If $ \tyn{\G}\Nt$  and $\G$ is bounded and $\tupleOK\G\Msg$, then $\Nt\parN\Msg$ is queue-consuming.  
\end{thm}
\begin{proof} 
By Strong Subject Reduction  (\refToTheorem{thm:ssr}) all complete lockstep computations of $\Nt\parN\Msg$ are complete lockstep computations of $\G\parN\Msg$. Then it is enough to show that $\G\parN\Msg$ is queue-consuming.

We prove the  statement  by induction on $\wgs\ms\G$. 
Consider a complete lockstep computation $\set{\G_k\parG\Msg_k \lockredsub{\Delta_k}{k}\G_{k+1}\parG\Msg_{k+1}}_{k<x}$ where $\G_0\parG\Msg_0 = \G\parG\Msg$ with $\Msg \equiv \addMsg\ms{\widehat\Msg}$. 
If $\wgs\ms \G=0$, by \refToLemma{lem:w-msg-dec}(\ref{lem:w-msg-dec:1}) we get $\G\parG\Msg\stackred{\inc\ms}\G'\parG\Msg'$ with $\Msg' \equiv \widehat\Msg$, 
then, since  $\Delta_0$ 
is a maximal coherent set, we have $\inc\ms\in\Delta_0$ as needed. 
If $\wgs\ms \G>0$, let us assume $\inc\ms\notin\Delta_0$ (otherwise the  statement  is trivial), 
then  by \refToLemma{lem:w-msg-dec}(\ref{lem:w-msg-dec:3}) $\G\parG\Msg\lockred{\Delta_0}\G_1\parG\Msg_1$  implies 
$\Msg_1 \equiv \addMsg\ms{\Msg'_1}$ and $\wgs\ms {\G_1} < \wgs\ms \G$. 
Hence the   statement  follows by induction hypothesis. 
\end{proof}

Now progress is an immediate consequence of Theorems~\ref{thm:df},~\ref{thm:nli} and ~\ref{thm:nom}. 

\begin{thm}[Progress]\label{thm:pr}
If $ \tyn{\G}\Nt$ and $\G$ is bounded and $\tupleOK\G\Msg$, then $\Nt\parN\Msg$ has the  progress   property.  
\end{thm}


\section{The undecidability of (weak) balancing} 
\label{sect:dec-algo}

The cornerstone we build on to prove well-typed sessions have progress is the balancing property. 
This has a non-trivial definition mixing coinduction and induction, hence a natural question which may arise is whether it is decidable. 

In this section we answer this question, proving that (weak) balancing is actually undecidable. 
Then, in the next section, 
we will define an algorithm to test balancing of a type configuration proving it is sound with respect to the coinductive definition. 
 We carry out these results for the balancing property and only discuss how they can be easily adapted to the weak version.

The undecidability proof follows the same strategy used to show the undecidability of asynchronous session subtyping \cite{BCZ17}. 
Indeed, we provide a reduction from the complement of the acceptance problem on queue machines, known to be undecidable as queue machines are Turing complete \cite{Kozen97}, to the problem of checking if a given type configuration has the balancing property. 
We first recall basic definitions about queue machines. 

\begin{defi}[Queue machine] \label{def:qm}
A \emph{queue machine} is a tuple $M = \ple{Q,\Sigma,\uG,\$,s,\delta}$ where 
\begin{itemize}
\item $Q$ is a finite set of \emph{states};
\item $\uG$ is a finite set named \emph{queue alphabet};
\item $\Sigma \subseteq\uG$ is a finite set named \emph{input alphabet};
\item $\$ \in \uG\setminus\Sigma$ is the \emph{initial queue symbol};
\item $s \in Q$ is the \emph{initial state};
\item $\fun{\delta}{Q\times\uG}{Q\times\uG^\star}$ is the \emph{transition function}. 
\end{itemize}
A \emph{configuration} of $M$ is a pair $\ple{q,\ug} \in Q\times\uG^\star$, where $q$ is the current state and $\ug$ is a sequence of symbols in $\uG$ modelling the queue. 
An \emph{initial configuration} has shape $\ple{s,\ug\$}$ where $\ug \in \Sigma^\star$. 

The \emph{transition relation} $\to_M$ is a binary relation on configurations of $M$  defined  by 
\[ \ple{q,a\alpha} \to_M \ple{q',\alpha\ub} \quad\text{ if }\quad \delta(q,a) = \ple{q',\ub} \]
\end{defi}
\noindent
Note that, since $\delta$ is a total function, computation is deterministic and the only   final configurations of the machine $M$    have shape $\ple{q, \epsilon }$. 
That is, $M$ terminates only when the queue is empty. 

We can now define the acceptance condition for a queue machine. 
As usual, we denote by $\to_M^\star$ the reflexive and transitive closure of the transition relation. 

\begin{defi}[Acceptance] \label{def:queue.accept}
Let $M = \ple{Q,\Sigma,\uG,\$,s,\delta}$ be  a queue machine. 
A sequence $\ug \in \Sigma^\star$ is \emph{accepted} iff $\ple{s,\ug\$} \to_M^\star \ple{q, \epsilon }$, for some $q \in Q$. 
\end{defi}
Intuitively, this means that a sequence $\ug$ on the input alphabet is accepted iff starting from $\ple{s,\ug\$}$ the machine terminates. 

Let us now define the encoding which will give us the reduction. 
Assume a queue machine $M=(Q,\Sigma,\uG,\$,s,\delta)$.  
We define a global type whose players are $\pp$ and $\q$ and labels are symbols in $\uG$. 
Since $\pp$ is always the sender and $\q$ the receiver in defining the global type we avoid prefixing $!$ and $?$ by $\pp\,\q$. 
Moreover in the queue we only write the exchanged labels, that is, we identify message queues (in type configurations) with elements of $\uG^\star$. 
For every state $q \in Q$, let $\G_q$ be defined by the following equation: 
 \begin{align*} 
\G_q = ?\set{a;!b^a_1;\dots ;!b^a_{n_a}; \G_{q'}}_{a\in\uG} && 
\text{if }\delta(q,a)=(q',b^a_1\dots b^a_{n_a}) 
\end{align*} 
I.e. the global type $\G_q$ is an input reading all characters $a$ of the queue alphabet followed by the outputs of the characters $b^a_1,\dots, b^a_{n_a}$ which are put on the queue and then by the global type $\G_{q'}$, where $q'$ is the following state.   

The encoding of the machine configuration $\ple{q,\ug}$, denoted by $\textlbrackdbl\ple{q,\ug}\textrbrackdbl$,  is the type configuration 
$\G_q\parG \ug$. 

Our goal is to relate the acceptance problem  of  the queue machine $M$ with the balancing property of the type configurations obtained through the above encoding. 
The next theorem formally states this relationship. 

\begin{thm} \label{thm:reduction}
The machine $M$ does not accept $\ug \in \Sigma^\star$  iff the encoding of its initial configuration $\ple{s,\ug\$}$ is balanced, that is, 
$\tupleOK{\G_s}{\ug\$}$. 
\end{thm} 

To prove the theorem we rely on  some  
lemmas. 
The first one shows that diverging computations in a queue machine correspond to derivations of the balancing property. 

\begin{lem}\label{lem:div-iom}
If $\ple{q,\ug}$ diverges, then $\tupleOK{\G_q}{\ug}$ holds. 
\end{lem}
\begin{proof}
Since $\ple{q,\ug}$ diverges, we have a sequence $\ple{q_i,\ug_i}_{i \in \N}$ such that 
$\ple{q,\ug} = \ple{q_0,\ug_0}$ and, for all $i \in \N$, $\ple{q_i,\ug_i} \to_M \ple{q_{i+1},\ug_{i+1}}$. 
By definition of $\to_M$ we have $\ug_i = a_i\ub_i$ and $\delta(q_i,a_i) = \ple{q_{i+1},b_1^{a_i}\ldots b_{n_{a_i}}^{a_i}}$. 
For all $i \in \N$, define the set   $A_i$ of balancing judgments as follows:   
\[ 
 A_i = \{ \tupleOK{\G_{q_i}}{\ug_i}\ ,\
\tupleOK{!b_1^{a_i}; \ldots !b_{n_{a_i}}^{a_i}; \G_{q_{i+1}}}{\ub_i}\ ,\  
\ldots\ ,\ 
\tupleOK{!b_{n_{a_i}}^{a_i}; \G_{q_{i+1}}}{\ub_i b_1^{a_i}\ldots b_{n_{a_i}-1}^{a_i}} \} 
\]
and let $A = \bigcup_{i \in \N} A_i$. 
We prove that $A$ is consistent with respect to the rules defining balancing, then the    statement    follows by coinduction. 

We can  easily  prove that for every judgement   $\tupleOK{\G}{\ub}$ in $A$, the judgement $\algread{\G}{\ub}$ is derivable. In fact each $\ub$ only contains symbols in $\uG$ and $\G$ always offers some outputs  and then a choice of inputs for all symbols in $\uG$. 
  
Therefore, to conclude, we have just to show that every $\tupleOK{\G}{\ub}$ in $A$ is the consequence of a rule whose premises are still in $A$. 
By definition of $A$, we have that $\tupleOK{\G}{\ub}$ belongs to $A_i$ for some $i \in \N$. 
We distinguish three cases. 
\begin{itemize}
\item If $\G = \G_{q_i}$, then $\ub = \ug_i = a_i\ub_i$ and, by definition of $\G_{q_i}$, 
  $\G = ?\set{a;\G_a}_{a\in\uG}$    with $\G_{a_i} =  !b_1^{a_i}; \ldots !b_{n_{a_i}}^{a_i}; \G_{q_{i+1}}$. 
Thus, $\tupleOK{\G}{\ub}$ is the consequence of Rule \rn{B-In} with   unique    premise $\tupleOK{\G_{a_i}}{\ub_i}$, that belongs to $A_i \subseteq A$ by definition. 
\item If $\G = !b_k^{a_i}; \ldots !b_{n_{a_i}}^{a_i}; \G_{q_{i+1}}$ for some $k < n_{a_i}$, then $\ub = \ub_i b_1^{a_i}\ldots b_{k-1}^{a_i}$. 
Thus, $\tupleOK{\G}{\ub}$ is the consequence of Rule \rn{B-Out} with unique    premise 
$\tupleOK{!b_{k+1}^{a_i}; \ldots !b_{n_{a_i}}^{a_i}; \G_{q_{i+1}}}{\ub b_k^{a_i}}$, that belongs to $A_i \subseteq A$ by definition. 
\item If $\G = !b_{n_{a_i}}^{a_i}; \G_{q_{i+1}}$, then $\ub = \ub_ib_1^{a_i}\ldots b_{n_{a_i}-1}^{a_i}$. 
Thus, $\tupleOK{\G}{\ub}$ is the consequence of Rule \rn{B-Out} with unique    premise 
$\tupleOK{\G_{q_{i+1}}}{\ub b_{n_{a_i}}^{a_i}}$ that belongs to $A_{i+1}\subseteq A$ as $\ub b_{n_{a_i}}^{a_i} = \ub_i b_1^{a_i}\ldots b_{n_{a_i}}^{a_i} = \ug_{i+1}$ by definition. \qedhere
\end{itemize}
\end{proof}

The following lemma shows that transitions in the queue machine preserve the balancing property. 

\begin{lem}\label{lem:tr-sr} 
If $\ple{q,\ug}\to_M\ple{q',\ug'}$ and $\tupleOK{\G_q}{\ug}$, then   $\tupleOK{\G_{q'}}{\ug'}$.
\end{lem}
\begin{proof}
If $\ple{q,\ug}\to_M \ple{q',\ug'}$, then $\ug = a\alpha$, $\delta(q,a) = \ple{q',\ub}$ and $\ug' = \alpha\ub$ with $\ub = b_1^a\ldots b_{n_a}^a$. 
If $\tupleOK{\G_q}{\ug}$ holds, 
since   $\G_q = ?\set{a'; b_1^{a'};\ldots b_{n_{a'}}^{a'}; \G_{q_{a'}}}_{{a'}\in\uG}$ with $\delta(q,{a'}) = \ple{q_{a'},b_1^{a'}\ldots b_{n_{a'}}^{a'}}$, we get $q_a=q'$ and $\tupleOK{\G_q}{\ug}$   
 is derived by Rule \rn{B-In}  with 
unique premise 
$\tupleOK{!b_1^a ; \ldots b_{n_a}^a; \G_{q'}}{\alpha}$.
This premise  
is obtained by applying $n_a$ times Rule \rn{B-Out} to the judgement \linebreak
$\tupleOK{\G_{q'}}{\alpha b_1^a\ldots b_{n_a}^a}$, hence it holds as needed. 
\end{proof}

Finally, the next lemma shows that a configuration $\ple{q,\ug}$ whose encoding has the balancing property is not stuck. 

\begin{lem}\label{lem:tr-prog}
If $\tupleOK{\G_q}{\ug}$ holds, then $\ple{q,\ug}\to_M$. 
\end{lem}
\begin{proof}
By definition we have $\G_q = ?\set{a; !b_1^a; \ldots !b_{n_a}^a; \G_{q_a}}_{a\in\uG}$ with $\delta(q,a) = \ple{q_a,b_1^a\ldots b_{n_a}^a}$. 
Then, $\tupleOK{\G_q}{\ug}$ is derived by Rule \rn{B-In}, therefore 
$\ug = a'\ug'$ for some $a' \in \uG$ and $\ug' \in \uG^\star$ and so 
$\ple{q,\ug}\to_M\ple{q_{a'},\ug'b_1^{a'}\ldots b_{n_{a'}}^{a'}}$. 
\end{proof}

\begin{proof}[Proof of \refToThm{reduction}]
To prove the left-to-right implication, note that 
if $M$ does not accept $\ug \in \Sigma^\star$, then $\ple{s,\ug\$}$ diverges, hence, by \refToLem{div-iom}, we get $\tupleOK{\G_s}{\ug\$}$. 
Towards a proof of the other implication, we inductively construct an infinite computations $\ple{q_i,\ug_i}$ where $\tupleOK{\G_{q_i}}{\ug_i}$ holds for all $i \in \N$ as follows: 
\begin{itemize}
\item set $\ple{q_0,\ug_0} = \ple{s,\ug\$}$, then $\tupleOK{\G_{q_0}}{\ug_0}$ holds by hypothesis; 
\item since $\tupleOK{\G_{q_i}}{\ug_i}$ holds by induction hypothesis, by \refToLem{tr-prog}, we have $\ple{q_i,\ug_i}\to_M\ple{q',\ug'}$ and, by \refToLem{tr-sr}, we get $\tupleOK{\G_{q'}}{\ug'}$ holds. 
Then, we set $\ple{q_{i+1},\ug_{i+1}} = \ple{q',\ug'}$. 
\end{itemize}
Therefore, $\ple{s,\ug\$}$ diverges and so $\ug$ is not accepted. 
\end{proof}

\begin{cor}\label{cor:undecidability}
The balancing property is undecidable. 
\end{cor}

 We can get the same result for weak balancing just removing the parts of the proofs dealing with the readability judgement in Lemmas~\ref{lem:div-iom}, \ref{lem:tr-sr}, as they do not play any relevant role in the reduction.



\section{Recovering effectiveness} 
\label{sect:algo}

In order to recover from the undecidability result in \refToCor{cor:undecidability}, 
we define an inductive version of the balancing property, proving it is sound with respect to the coinductive definition (\refToThm{alg-sound}), thus getting a sound algorithm to check balancing. 
Note that this soundness result is enough to ensure that well-typed sessions - where the corresponding type configuration is successfully checked by this algorithm - has progress, thus obtaining an effective type system.  

This inductive definition, described by  the rules in \refToFig{wfConfigA},  follows a standard pattern used to deal with regular structures: 
we consider an enriched judgement $\Hset\OKA\G\Msg$, where 
$\Hset$ is an auxiliary set  of pairs  of  types, queues used to detect cycles in global types and test a condition on them (see the side condition of Rule \rn{\ib{Cycle}}), thus avoiding non-termination. 
Other rules are essentially the same as in the coinductive version. 

\begin{figure} 
\begin{math} 
\begin{array}{c}
\NamedRule{\rulename{\ib{End}}}{} 
{\Hset\OKA{\End}{\emptyset}} {}
\quad\quad
\NamedRule{\rulename{\ib{Cycle}}}{} 
{\Hset,(\G,\Msg)\OKA{\G}{\Msg'}} {\OK\G\Msg{\Msg'}}
\\[3ex]
\NamedRule{\rulename{\ib{Out}}}
{\Hset,(\agtO{\pp}{\q}i I{\la}{\G},\Msg)\OKA{\G_i}{\addMsg\Msg{\mq\pp{\la_i}\q}}\quad \forall i\in I} 
{\Hset\OKA{\agtO{\pp}{\q}i I{\la}{\G}}{\Msg}}
{ \algread{\agtO{\pp}{\q}{i}{I}{\la}{\G}}{\Msg}}
\\[3ex]
\NamedRule{\rulename{\ib{In}}}
{\Hset,(\agtI \pp\q \la \G,\mq\pp{\la_h}\q\cdot\Msg)\OKA{\G_h}{\Msg}} 
{\Hset\OKA{\agtI \pp\q \la \G}{\mq\pp{\la_h}\q\cdot\Msg}}  
{\algread{\agtII{\pp}{\q}{i}{I}{\la}{\G}}{\addMsg{\mq\pp{\la_h}\q}{\Msg}} \  h\in I}
\end{array}
\end{math} 
\caption{Balancing of type configurations (inductive version).} 
\label{fig:wfConfigA}
\end{figure}

To make this pattern work, we have to appropriately define the judgement $\OK\G\Msg{\Msg'}$. 
One obvious possibility is to require $\Msg \equiv \Msg'$  and $\algread{\G}{\Msg}$, which is always sound,  because it forces derivations to be regular. 
However, this is too restrictive, as it accepts only type configurations where the queue cannot grow indefinitely. 
This for instance is not the case for  our running example. 
Therefore, we consider a more sophisticated definition of $\OK\G\Msg{\Msg'}$ which goes beyond regular derivations. 
The judgement is defined in \refToFig{algwf}. 

\begin{figure}[b] 
\begin{math}
\begin{array}{c}
\NamedRule{}{
  \algcomp{}{\G}{\Msg''} \quad 
  \algreadinf{}{\G}{\Msg''} \quad 
  \algread{\G}{\Msg} 
}{ \OK\G\Msg{\Msg'} }{ \Msg' \equiv \addMsg{\Msg}{\Msg''} } 
\\[3ex]
\hline
\\[2ex]
\NamedRule{\rn{\agr{End}}}{}{\algcomp{\Bset}{\End}{\Msg}}{} 
\qquad
\NamedRule{\rn{\agr{Cycle}}}{}{ \algcomp{\Bset,(\G,\Msg)}{\G}{\Msg} }{} 
\\[3ex]
\hspace*{-2mm} 
\NamedRule{\rn{\agr{Out}}}
{ \algcomp{\Bset, (\agtO{\pp}{\q}{i}{I}{\la}{\G},\Msg)}{\G_i}{\Msg_i} \ \ \forall i\in I }
{ \algcomp{\Bset}{\agtO{\pp}{\q}{i}{I}{\la}{\G}}{\Msg} }
{ \begin{array}{c}  \addMsg\Msg{\mq\pp{\la_i}\q} \equiv_{\G_i} \addMsg{\mq\pp{\la_i}\q}{\Msg_i}~\forall i\in I \end{array}}
\hspace*{-5mm} 
\\[3ex]
\NamedRule{\rn{\agr{In}}}
{ \algcomp{\Bset,(\agtI{\pp}{\q}{\la}{\G},\Msg)}{\G_i}{\Msg}  \ \ \forall i\in I}
{ \algcomp{\Bset}{\agtI{\pp}{\q}{\la}{\G}}{\Msg}} {} 
\\[3ex]
\hline
\\[2ex]
\NamedRule{\rn{End-DR}}{}
{ \algreadinf{\Gset}{\End}{\emptyset} } {} 
\qquad
\NamedRule {\rn{Cycle-DR}}
{ \algread{\G}{\Msg} }
{ \algreadinf{\Gset,\G}{\G}{\Msg} } {}
\\[3ex] 
\NamedRule {\rn{Out-DR}}
{ \algreadinf{\Gset,\agtO{\pp}{\q}{i}{I}{\la}{\G}}{\G_i}{\Msg} \ \ \forall i \in I } 
{ \algreadinf{\Gset}{\agtO{\pp}{\q}{i}{I}{\la}{\G}}{\Msg} } {}
\\[3ex] 
\NamedRule {\rn{In-DR}}
{ \algreadinf{\Gset,\agtI{\pp}{\q}{\la}{\G}}{\G_i}{\Msg}  \ \ \forall i\in I}
{ \algreadinf{\Gset}{\agtI{\pp}{\q}{\la}{\G}}{\Msg} } {}
\end{array}
\end{math}
\caption{Ok, agreement and deep read judgments.} \label{fig:algwf} 
\end{figure}

The first condition we require is  the equivalence $\Msg' \equiv \addMsg\Msg{\Msg''}$. 
This is needed because, if a message in $\Msg$ is not in $\Msg'$,   then   a coinductive derivation
of the judgement $\tupleOK{\G}{\Msg'}$ would get stuck on Rule \rn{B-In}, i.e., the judgement would not be derivable, making the algorithm unsound. 
However, we allow the presence of additional messages, so that 
the queue between two occurrences 
of the same global type can grow. 

In order to ensure that
{\em the messages in the queue $\Msg''$} do not interfere with the balancing 
of $\G$, we demand that 
they {\em can be exchanged with the outputs of $\G$.} 
This condition is checked by the judgement $\algcomp{}{\G}{\Msg''}$, whose definition relies on 
an extension of the equivalence $\equiv$ on queues, parametrised over a global type $\G$, 
denoted by $\equiv_\G$, see \refToDef{def:Gequiv}. 
The key rule is $\rn{\agr{Out}}$. 
This rule requires that for each branch $\G_i$ of an
output the message put at the end of the queue can go on top when queues are considered modulo  $\equiv_{\G_i}$.

Finally, we have to check that all messages in $\Msg'$ will be eventually read. 
To this end, we use judgements $\algread{\G}{\Msg}$ and $\algreadinf{}{\G}{\Msg''}$, where the latter ensures that 
messages in $\Msg''$  can be read in any path of $\G$ an unbounded number of times. 
This is needed because 
many copies of messages in $\Msg''$ will be accumulated at the end of the queue and we do not know in which path of $\G$ they will be read, thus 
they must be readable in all of them. 
In $\vdash_{\mathsf{dread}}$ all branches of a global type are explored without removing any message until getting $\End$ or an already considered global type. The first case is successful if the queue is empty, in the second case the judgement $\vdash_{\mathsf{read}}$ is required. 
For instance, if $\G=\agtIS{\pp}{\q}{\set{\Seq{\la}{\End},\Seq{\la'}{\End}}}$, then
$\algreadinfNot{}{\G}{\mq\pp{\la}\q}$, since the message $\mq\pp{\la}\q$ is not read in all branches.
However, $\algread{\G}{\mq\pp{\la}\q}$, since the message $\mq\pp{\la}\q$ is read in one of the branches.

To formally give the equivalence $\equiv_\G$ on queues  
we first define when two messages are $\G$-indistinguishable, meaning that  
the behaviour of $\G$ is independent on reading one or the other. Then, the equivalence $\equiv_\G$ 
is defined by extending the equivalence on queues to allow the replacement of messages with 
$\G$-indistinguishable ones.
\begin{defi}\label{def:Gequiv}
\begin{enumerate}
\item  The messages $\mq\pp\la\q$ and $\mq\pp{\la'}\q$ are \emph{$\G$-indistinguishable} if $\la$ and $\la'$ 
occur in $\G$ and 
for every input choice $\agtII{\pp}{\q}{i}{I}{\la}{\G}$ occurring in $\G$, 
\begin{itemize}
\item
either $\{\la,\la'\}\cap \{\la_i \mid i \in I\} = \emptyset$ 
\item or $\la = \la_h$, $\la' = \la_k$ and $\G_h = \G_k$ for some $h,k \in I$. 
\end{itemize}
\item The equivalence $\equiv_\G$ between queues is obtained extending $\equiv$ by the following clause: 
\[
\Msg_1\cdot \mq\pp\la\q \cdot \Msg_2 \equiv_\G \Msg_1\cdot \mq\pp{\la'}\q \cdot \Msg_2 
\quad \text{if}\quad 
\text{$\mq\pp\la\q,\mq\pp{\la'}\q$ are $\G$-indistinguishable}
\]
\end{enumerate}
\end{defi}
Notice that, if $\Msg \equiv_\G \Msg_1\cdot\mq\pp\la\q\cdot\Msg_2$, then 
$\Msg \equiv \Msg'_1\cdot\mq\pp{\la'}\q\cdot\Msg'_2$ with $\mq\pp\la\q,\mq\pp{\la'}\q$ $\G$-indistinguishable and $\Msg_1 \equiv_\G \Msg'_1$ and $\Msg_2 \equiv_\G \Msg'_2$. 

Taking $\G=\agtIS{\pp}{\q}{\set{\Seq{\la}{\End},\Seq{\la'}{\End}}}$, we get that 
$\addMsg{\mq\pp{\la}\q}{\mq\pp{\la'}\q}\equiv_\G\addMsg{\mq\pp{\la'}\q}{\mq\pp{\la}\q}$,
whereas $\addMsg{\mq\pp{\la}\q}{\mq\pp{\la'}\q}\not\equiv\addMsg{\mq\pp{\la'}\q}{\mq\pp{\la}\q}$. 
We will see that  this parametrised equivalence plays a crucial role in  deriving  the agreement
judgement  for the example of the hospital.

To prove the soundness theorem (\refToThm{alg-sound}), 
we need properties of the auxiliary judgements.

First properties concern the two readability judgements. 
\refToLemma{lem:alg-read} shows that queues with the readability property can be split preserving readability. 
\refToLemma{lem:alg-read-inf}(\ref{lem:alg-read-inf:1}) and \refToLemma{lem:alg-read-inf}(\ref{lem:alg-read-inf:2}) prove weakening and cut elimination for deep readability, respectively.  
Relying on these two properties, 
\refToLemma{lem:alg-read-inf}(\ref{lem:alg-read-inf:3}) proves a strong form of inversion for deep readability, deriving it for any 
subtree  of a deeply  readable  global type with the same queue (when the set of 
hypotheses is empty). 
\refToLemma{lem:alg-read-inf}(\ref{lem:alg-read-inf:4}) and  \refToLemma{lem:read-merge}(\ref{lem:read-merge:1}) 
connect readability with  deep  readability. 
\refToLemma{lem:read-merge}(\ref{lem:read-merge:3}) extends \refToLemma{lem:alg-read} to deep readability. 
Finally, \refToProp{prop:equivG} shows that readability and deep readability are preserved by the parametrised equivalence $\equiv_\G$.

\begin{lem}\label{lem:alg-read}
If $\algread{\G}{\addMsg\Msg{\Msg'}}$, then $\algread{\G}{\Msg}$ and $\algread{\G}{\Msg'}$.
\end{lem}
\begin{proof}  
We show that  $\algread{\G}{\Msg}$ implies $\algread{\G}{\Msg'}$ when  $\Msg'$ is a subsequence of $\Msg$ (modulo $\equiv$).
The proof is by induction on the derivation of $\algread{\G}{\Msg}$.  
The only interesting cases are when the last applied rule is either \rn{In-R1} or \rn{In-R2}.  
\begin{description}
\item [\rn{In-R1}] We have $\G = \agtI{\pp}{\q}{\la}{\G}$ and 
$\Msg \equiv \addMsg{\mq\pp{\la_h}\q}{\widehat\Msg}$ with $h\in I$  and $\algread{\G_i}{\widehat\Msg}$ for all $i\in I$. 
If $\Msg' \equiv \addMsg{\mq\pp{\la_k}\q}{\Msg''}$ for some $k \in I$ (note that $k$ and $h$ may differ), then $\Msg''$ is a subsequence of $\widehat\Msg$. 
Then, by induction hypothesis, we get  $\algread{\G_i}{\Msg''}$ for all $i\in I$  and so we conclude  by applying Rule  \rn{In-R1}.  
Otherwise, $\Msg'$ is a subsequence of $\widehat\Msg$, then by induction hypothesis  we get $\algread{\G_i}{\Msg'}$ for all $i \in I$ and so we conclude  by applying Rule \rn{In-R2}. 
\item [\rn{In-R2}] We have $\G = \agtI{\pp}{\q}{\la}{\G}$ and 
$\Msg \not\equiv \addMsg{\mq\pp{\la_i}\q}{\widehat\Msg}$  and $\algread{\G_i}{\Msg}$  for all $i\in I$.  
If $\Msg' \equiv \addMsg{\mq\pp{\la_h}\q}{\Msg''}$   with $h\in I$, 
since $\Msg''$ is still a subsequence of $\Msg$, by induction hypothesis,  we get $\algread{\G_i}{\Msg''}$ for all $i \in I$  and so we conclude  by applying Rule  \rn{In-R1}.  
Otherwise, again by induction hypothesis  we get $\algread{\G_i}{\Msg'}$, for all $i \in I$,  and so we conclude  by applying Rule \rn{In-R2}. \qedhere
\end{description}
\end{proof}

\begin{lem}\label{lem:alg-read-inf}\label{lem:read-merge}
The following properties hold. 
\begin{enumerate}
\item\label{lem:alg-read-inf:1} If $\algreadinf{\Gset}{\G}{\Msg}$, then $\algreadinf{\Gset,\Gset'}{\G}{\Msg}$. 
\item\label{lem:alg-read-inf:2} If $\algreadinf{\Gset}{\G}{\Msg}$ and $\algreadinf{\Gset,\G}{\G'}{\Msg}$, then $\algreadinf{\Gset}{\G'}{\Msg}$. 
\item\label{lem:alg-read-inf:3} If $\algreadinf{}{\G}{\Msg}$, then $\algreadinf{}{\G'}{\Msg}$ for any  type  $\G'$  occurring in  $\G$. 
\item\label{lem:alg-read-inf:4} If $\algreadinf{\Gset}{\G}{\Msg}$, then $\algread{\G}{\Msg}$. 
\item\label{lem:read-merge:1} If $\algread{\G}{\Msg_1}$ and  $\algreadinf{}{\G}{\Msg_2}$, then $\algread{\G}{\addMsg{\Msg_1}{\Msg_2}}$. 
\item\label{lem:read-merge:3} If $\algreadinf{\Gset}{\G}{\addMsg{\Msg_1}{\Msg_2}}$, then $\algreadinf{\Gset}{\G}{\Msg_1}$ and $\algreadinf{\Gset}{\G}{\Msg_2}$. 
\end{enumerate}
\end{lem}
\begin{proof}
(\ref{lem:alg-read-inf:1}). The proof is by a straightforward induction on the derivation of $\algreadinf{\Gset}{\G}{\Msg}$. 

(\ref{lem:alg-read-inf:2}). The proof is by a straightforward induction on the derivation of $\algreadinf{\Gset,\G}{\G'}{\Msg}$ using Item (\ref{lem:alg-read-inf:1}) on $\algreadinf{\Gset}{\G}{\Msg}$ to properly extend $\Gset$ in order to apply the induction hypothesis. 

(\ref{lem:alg-read-inf:3}). 
The proof is by induction on the distance of  the subtree  $\G'$ from the root  in the tree of $\G$. 
If the distance is $0$, then $\G' = \G$ and so the statement is trivial. 
Otherwise, $\G'$ is the direct subtree of another subtree $\G''$ of $\G$, which is closer to the root; 
then by induction hypothesis, we have $\algreadinf{}{\G''}{\Msg}$. 
Let  $\G'' = \agtbg{\pp}{\q}{i\in I}{\la_i}{\G_i}$ and $\G' = \G_h$ for some $h \in I$. 
We get $\algreadinf{\G''}{\G'}{\Msg}$, as it is a premise of the rule used to derive $\algreadinf{}{\G''}{\Msg}$, 
then, by Item (\ref{lem:alg-read-inf:2}), we get $\algreadinf{}{\G'}{\Msg}$, as needed. 

(\ref{lem:alg-read-inf:4}). The proof is by induction on the derivation of $\algreadinf{\Gset}{\G}{\Msg}$, splitting cases on the last applied rule in the derivation. 
The only non-trivial case is for Rule \rn{In-DR}. 
We have $\G = \agtI{\pp}{\q}{\la}{\G}$ and, by induction hypothesis,  we get $\algread{\G_i}{\Msg}$ for all $i\in I$. 
We have two cases: 
if $\Msg\not\equiv\addMsg{\mq\pp{\la_i}\q}{\Msg'}$  for all $i\in I$, the statement follows by Rule  \rn{In-R2}. 
 Otherwise, $\Msg = \addMsg{\mq\pp{\la_h}\q}{\Msg'}$ for some $h \in I$ and, by \refToLem{alg-read} applied to $\algread{\G_i}{\Msg}$,  we get $\algread{\G_i}{\Msg'}$ for all $i \in I$,  thus the statement follows by Rule \rn{In-R1}. 

(\ref{lem:read-merge:1}). The proof is by induction on the derivation of $\algread{\G}{\Msg_1}$. 
We split cases on the last applied rule.
\begin{description}
\item [\rn{Empty-R}] We have $\Msg_1 = \emptyset$ and so the statement follows by Item~(\ref{lem:alg-read-inf:4}). 
\item [\rn{Out-R}] We have $\G = \agtO{\pp}{\q}{i}{I}{\la}{\G}$ and, since by hypothesis we have $\algreadinf{}{\G}{\Msg_2}$, by Item~(\ref{lem:alg-read-inf:3}) we get $\algreadinf{}{\G_i}{\Msg_2}$, for all $i\in I$. 
Then, by induction hypothesis, 
we get $\algread{\G_i}{\addMsg{\Msg_1}{\Msg_2}}$, for all $i\in I$ and so the statement follows by Rule \rn{Out-R}. 
\item [\rn{In-R1}] We have $\G = \agtI{\pp}{\q}{\la}{\G}$ and $\Msg_1 \equiv \addMsg{\mq\pp{\la_h}\q}{\Msg'}$ for some $h\in I$, and, since by hypothesis we have $\algreadinf{}{\G}{\Msg_2}$, by Item~(\ref{lem:alg-read-inf:3}) we get  $\algreadinf{}{\G_i}{\Msg_2}$ for all $i\in I$. 
Then, by induction hypothesis,  
we get $\algread{\G_i}{\addMsg{\Msg'}{\Msg_2}}$ for all $i\in I$  and so the statement follows by Rule \rn{In-R1}. 
\item [\rn{In-R2}] We have $\G = \agtI{\pp}{\q}{\la}{\G}$ and $\Msg_1 \not\equiv \addMsg{\mq\pp{\la_i}\q}{\Msg'}$ for all $i\in I$,  and, since by hypothesis we have $\algreadinf{}{\G}{\Msg_2}$, by Item~(\ref{lem:alg-read-inf:3}), we get $\algreadinf{}{\G_i}{\Msg_2}$  for all $i\in I$. 
Then, by induction hypothesis, 
we get $\algread{\G_i}{\addMsg{\Msg_1}{\Msg_2}}$  for all $i\in I$.  
Now, if $\addMsg{\Msg_1}{\Msg_2}\equiv\addMsg{\mq\pp{\la_h}\q}{\widehat\Msg}$ for some $h\in I$, \refToLemma{lem:alg-read} implies  $\algread{\G_i}{\widehat\Msg}$  for all $i\in I$ 
and so the statement follows by Rule \rn{In-R1}. 
Otherwise  the statement follows by Rule  \rn{In-R2}. 
\end{description}

(\ref{lem:read-merge:3}). The proof is by a straightforward induction on the derivation of $\algreadinf{\Gset}{\G}{\addMsg{\Msg_1}{\Msg_2}}$ using \refToLem{alg-read}. 
\end{proof}

\begin{prop}\label{prop:equivG}
Let $\Msg\equiv_\G\Msg'$. 
\begin{enumerate}
\item\label{prop:equivG:1} If $\algread{\G}{\Msg}$, then  $\algread{\G}{\Msg'}$. 
\item\label{prop:equivG:2} If $\algreadinf{\Gset}{\G}{\Msg}$, then  $\algreadinf{\Gset}{\G}{\Msg'}$. 
\end{enumerate}
\end{prop}
\begin{proof}
(\ref{prop:equivG:1}). 
The proof is by induction on the derivation of $\algread{\G}{\Msg}$, splitting cases on the last applied rule. 
The only relevant case is for Rule \rn{In-R1}. 
We have $\G = \agtII{\pp}{\q}{i}{I}{\la}{\G}$ and $\Msg \equiv \addMsg{\mq\pp{\la_h}\q}{\widehat\Msg}$ with $h \in I$. 
Then, we have $\Msg' \equiv \addMsg{\mq\pp{\la'}\q}{\widehat\Msg'}$ with $\mq\pp{\la_h}\q,\mq\pp{\la'}\q$ $\G$-indistinguishable and $\widehat\Msg \equiv_\G \widehat\Msg'$, which implies 
$\widehat\Msg\equiv_{\G_i}\widehat\Msg'$ for all $i\in I$.
Hence, by induction hypothesis,  we get $\algread{\G_i}{\widehat\Msg'}$ for all $i\in I$. By definition of $\G$-indistinguishability we have $\la' = \la_k$ 
for some $k \in I$.  
Finally, we get the statement by Rule \rn{In-R1}. 

(\ref{prop:equivG:2}). It follows by a straightforward induction on the derivation of $\algreadinf{\Gset}{\G}{\Msg}$, using Item~(\ref{prop:equivG:1}). 
\end{proof}

To reason about the agreement judgement with an empty set of hypotheses, it is often convenient to consider an equivalent coinductive formulation reported in \refToFig{agree-coind}. 
These two formulations are equivalent because all infinite derivations using rules in \refToFig{agree-coind} are actually regular. 
In fact, 
 in a derivation of $\pairA\G\Msg$, there are only finitely many different global types, as $\G$ is regular. 
Moreover, there are finitely many different queues,  since they must all have the same length as the initial one and labels must occur in $\G$, by definition of $\G$-equivalence. 
Hence, such a derivation has only finitely many different nodes, namely, it is regular.  
This  implies  that  the judgement in \refToFig{agree-coind} is equivalent to the agreement judgement in \refToFig{algwf}  by \cite[Theorem~5.2]{Dagnino21}. 

\begin{figure}
\begin{math} 
\begin{array}{c} 
\NamedCoRule{\rn{\agr{End$'$}}}{}{\pairA\End\Msg}{} 
\\[3ex] 
\NamedCoRule{\rn{\agr{Out$'$}}}{ \pairA{\G_i}{\Msg_i}\ \ \forall i\in I }{ \pairA{\agtO{\pp}{\q}{i}{I}{\la}{\G}}\Msg }
{  \addMsg\Msg{\mq\pp{\la_i}\q} \equiv_{\G_i} \addMsg{\mq\pp{\la_i}\q}{\Msg_i}\ \forall i\in I } 
\\[3ex]
\NamedCoRule{\rn{\agr{In$'$}}}{ \pairA{\G_i}\Msg \ \ \forall i\in I}{ \pairA{\agtI{\pp}{\q}{\la}{\G}}\Msg }{} 
\end{array}
\end{math} 
\caption{Agreement judgement (coinductive version).}
\label{fig:agree-coind}
\end{figure}

The next lemma proves inversion for Rules \rn{\agr{Out}} and \rn{\agr{In}}. 

\begin{lem}\label{lem:comp-inv}
Let $\algcomp{}{\G}{\Msg}$, then the following hold. 
\begin{enumerate}
\item\label{lem:comp-inv:2} If $\G = \agtO{\pp}{\q}{i}{I}{\la}{\G}$, then,  for each $i\in I$, there is $\Msg_i$ such that  
$\algcomp{}{\G_i}{\Msg_i}$ and $\addMsg{\Msg}{\mq\pp{\la_i}\q} \equiv_{\G_i} \addMsg{\mq\pp{\la_i}\q}{\Msg_i}$. 
\item\label{lem:comp-inv:1}If $\G = \agtI{\pp}{\q}{\la}{\G}$, then $\algcomp{}{\G_i}{\Msg}$ for all $i\in I$. 
\end{enumerate} 
\end{lem} 
\begin{proof}
 This is straightforward by the equivalent coinductive characterisation of $\algcomp{}{\G}{\Msg}$ given in \refToFig{agree-coind}.  
\end{proof}

 We now prove  
that agreement and deep readability judgements (when the set of hypothesis is empty) are preserved by concatenation of queues. 

\begin{lem}\label{lem:sound-aux}
The following properties hold. 
\begin{enumerate}
\item\label{lem:sound-aux:1} If $\algcomp{}{\G}{\Msg_1}$ and $\algcomp{}{\G}{\Msg_2}$, then $\algcomp{}{\G}{\addMsg{\Msg_1}{\Msg_2}}$. 
\item\label{lem:sound-aux:2} If $\algreadinf{}{\G}{\Msg_1}$ and $\algreadinf{}{\G}{\Msg_2}$, then $\algreadinf{}{\G}{\addMsg{\Msg_1}{\Msg_2}}$. 
\end{enumerate}
\end{lem}
\begin{proof}
(\ref{lem:sound-aux:1}). The proof is by coinduction relying on the  
coinductive characterisation of $\algcomp{} {\G}{\Msg}$ in \refToFig{agree-coind}. 
Consider the set $\AS$ defined as follows: 
\begin{align*} 
\pairA{\G}{\widehat\Msg} \text { belongs to }\AS\text { if }\widehat\Msg \equiv \Msg_1\cdot\Msg_2\text { and }\algcomp{}{\G}{\Msg_1}\text { and }\algcomp{}{\G}{\Msg_2}\text { hold}
\end{align*}  
We have to prove that $\AS$ is consistent, that is, every element $\pairA{\G}{\widehat\Msg}$ in $\AS$ is the conclusion of a rule in \refToFig{agree-coind}, whose premises are still in $\AS$. 
We split cases on $\G$. 
\begin{itemize} 
\item If $\G = \End$, the statement follows immediately by Rule \rn{\agr{End$'$}}, as it has no premises. 
\item  If $\G = \agtO{\pp}{\q}{i}{I}{\la}{\G}$, then by \refToLemItem{comp-inv}{2} we get $\algcomp{}{\G_i}{\Msg_i^k}$ and $\addMsg{\Msg_k}{\mq\pp{\la_i}\q} \equiv_{\G_i} \addMsg{\mq\pp{\la_i}\q}{\Msg_i^k}$ for $k=1,2$ and for all $i \in I$.  
Therefore, $\pairA{\G_i}{\addMsg{\Msg_i^1}{\Msg_i^2}}$ belongs to $\AS$ for all $i \in I$. 
Notice that
\[
\addMsg{\addMsg{\Msg_1}{\Msg_2}}{\mq\pp{\la_i}\q} \equiv_{\G_i} \addMsg{\Msg_1}{\addMsg{\mq\pp{\la_i}\q}{\Msg_i^2}} \equiv_{\G_i} \addMsg{\mq\pp{\la_i}\q}{\addMsg{\Msg_i^1}{\Msg_i^2}} 
\]
for all $i\in I$, hence,  we conclude by Rule \rn{\agr{Out$'$}}.
\item If $\G = \agtII{\pp}{\q}{i}{I}{\la}{\G}$, then by \refToLemItem{comp-inv}{1} we get $\algcomp{}{\G_i}{\Msg_k}$ for $k = 1,2$ and for all $i \in I$. 
Therefore, $\pairA{\G_i}{\addMsg{\Msg_1}{\Msg_2}}$ belongs to $\AS$ for all $i \in I$, hence we conclude by Rule \rn{\agr{In$'$}}. 
\end{itemize} 

(\ref{lem:sound-aux:2}). 
We generalise the statement proving that, 
if $\algreadinf{\Gset}{\G}{\Msg_1}$ and $\algreadinf{}{\G}{\Msg_2}$, then $\algreadinf{\Gset}{\G}{\addMsg{\Msg_1}{\Msg_2}}$.  
We get the thesis when $\Gset$ is empty. 
The proof is by induction on the derivation of $\algreadinf{\Gset}{\G}{\Msg_1}$. 
For Rule \rn{End-DR}, we have $\G = \End$, hence both $\Msg_1$ and $\Msg_2$ are empty. 
Then the thesis follows immediately by Rule \rn{End-DR}.  
For Rules \rn{Out-DR} and \rn{In-DR}, we have $\G = \agtb\pp\q{i}{I}\la\G$ and $\algreadinf{\Gset,\G}{\G_i}{\Msg_1}$ for all $i \in I$. 
By \refToLemItem{alg-read-inf}{3}, we get $\algreadinf{}{\G_i}{\Msg_2}$, for all $i \in I$, hence the thesis follows from the induction hypothesis  applying  again Rules \rn{Out-DR} and \rn{In-DR}, respectively. 
For Rule \rn{Cycle-DR}, we have $\Gset = \Gset',\G$ and $\algread{\G}{\Msg_1}$, then \refToLemItem{read-merge}{1} implies $\algread{\G}{\addMsg{\Msg_1}{\Msg_2}}$ using 
$\algreadinf{}{\G}{\Msg_2}$. 
So the thesis follows applying Rule \rn{Cycle-DR}. 
\end{proof}

 Lastly, we prove  
that deep readability is preserved when the queue is changed as in Rule \rn{\agr{Out}} of \refToFig{algwf}.

\begin{lem}\label{lem:read-rev}
If $\algreadinf{}{\G}{\Msg}$ and $\addMsg{\Msg}{\mq\pp\la\q} \equiv_\G \addMsg{\mq\pp\la\q}{\Msg'}$, then $\algreadinf{}{\G}{\Msg'}$. 
\end{lem}
\begin{proof}
From the hypothesis we have $\addMsg{\Msg}{\mq\pp\la\q} \equiv \addMsg{\mq\pp{\la'}\q}{\widehat\Msg}$ with 
$\mq\pp\la\q$ and $\mq\pp{\la'}\q$ $\G$-indistinguishable and $\widehat\Msg \equiv_\G \Msg'$. 
We have two cases. 
\begin{itemize}
\item If there is no message from $\pp$ to $\q$ in $\Msg$, then 
$\addMsg{\Msg}{\mq\pp\la\q} \equiv \addMsg{\mq\pp\la\q}{\Msg}$, hence we have $\la  = \la'$ and $\Msg \equiv \widehat\Msg$. 
Therefore, $\Msg \equiv_\G \Msg'$ and the statement follows by \refToPropItem{equivG}{2}. 
\item Otherwise, $\Msg \equiv \addMsg{\mq\pp{\la'}\q}{\Msg_1}$ and so 
$\addMsg{\Msg_1}{\mq\pp\la\q} \equiv \widehat\Msg \equiv_\G \Msg'$, which implies 
$\Msg' \equiv \addMsg{\Msg_2}{\mq\pp{\la''}\q}$ with $\mq\pp\la\q$ and $\mq\pp{\la''}\q$ $\G$-indistinguishable and $\Msg_1\equiv_\G \Msg_2$. 
By \refToLemItem{read-merge}{3}, we have $\algreadinf{}{\G}{\mq\pp{\la'}\q}$ and $\algreadinf{}{\G}{\Msg_1}$ and, 
by \refToPropItem{equivG}{2}, we get $\algreadinf{}{\G}{\Msg_2}$. 
Note that $\mq\pp{\la'}\q\equiv_\G \mq\pp\la\q \equiv_\G \mq\pp{\la''}\q$, hence by \refToPropItem{equivG}{2} we get $\algreadinf{}{\G}{\mq\pp{\la''}\q}$. 
Finally, by \refToLemItem{sound-aux}{2}, we get $\algreadinf{}{\G}{\addMsg{\Msg_2}{\mq\pp{\la''}\q}}$, which is the statement. \qedhere
\end{itemize}
\end{proof}

Now we are able to state and prove the soundness result for the inductive  balancing judgement. 

\begin{thm}[Soundness]\label{thm:alg-sound}
If $\algwf{}{\G}{\Msg}$, then $\tupleOK{\G}{\Msg}$. 
\end{thm}
\begin{proof}
First of all we observe that the definition of the judgement $\algwf{\Hset}{\G}{\Msg}$ can be equivalently expressed assuming $\Hset$ to be a sequence rather than a set, 
the only difference is Rule \rn{\ib{Cycle}} which will have the following shape\\
\[
\NamedRule{\rulename{\ib{Cycle$'$}}}{
 \algread{\G}{\Msg} \quad 
  \algcomp{}{\G}{\Msg''} \quad 
  \algreadinf{}{\G}{\Msg''} 
}{ \Hset_1,(\G,\Msg),\Hset_2\OKA{\G}{\Msg'} }{\Msg' \equiv \addMsg{\Msg}{\Msg''}} 
\] 
where we have  expanded the $\vdash_{\sf ok}$ judgment.  
We say that a sequence $\Hset$ is coherent if $\algwf{\Hset_1}{\G}{\Msg}$ holds for any decomposition $\Hset = \Hset_1,(\G,\Msg),\Hset_2$. 

The proof is by coinduction on the definition of $\tupleOK{\G}{\Msg}$ (see \refToFig{wfConfig}). 
To this end, we define the set $\AS$ as follows:
\begin{align*} 
{\widehat\G}\parG{\widehat\Msg} \in \AS \ \text{ if }\  &\widehat\Msg \equiv_{\widehat\G} \addMsg{\Msg_1}{\Msg_2} \text{ and }
\algwf{\Hset}{\widehat\G}{\Msg_1} 
\text{ for some coherent $\Hset$} \\
&\text{ and }\algcomp{}{\widehat\G}{\Msg_2} \text{ and } \algreadinf{}{\widehat\G}{\Msg_2}  
\end{align*} 
From the hypothesis $\algwf{}{\G}{\Msg}$, we have immediately that $\G\parG\Msg\in \AS$, since $\Msg \equiv_\G \addMsg{\Msg}{\emptyset}$, $\algcomp{}{\G}{\emptyset}$ and $\algreadinf{}{\G}{\emptyset}$ always hold, and  the empty sequence is coherent. 
Thus, to conclude the proof, we just have to show that $\AS$ is consistent with respect to the rules in \refToFig{wfConfig}. 
Then, we prove that for all coherent $\Hset$, $\widehat\G$, $\Msg_1$ and $\Msg_2$, 
if $\widehat\Msg \equiv_{\widehat\G} \addMsg{\Msg_1}{\Msg_2}$  and  $\algwf{\Hset}{\widehat\G}{\Msg_1}$ and $\algcomp{}{\widehat\G}{\Msg_2}$ and $\algreadinf{}{\widehat\G}{\Msg_2}$,  then $\tupleOK{\widehat\G}{\widehat\Msg}$  is the conclusion of a rule in \refToFig{wfConfig} whose premises are in $\AS$. 
The proof is by induction on the length of $\Hset$, splitting cases on the last rule used to derive $\algwf{\Hset}{\widehat\G}{\Msg_1}$. 
\begin{description}
\item [\rn{\ib{End}}] We have $\widehat\G = \End$ and $\Msg_1 = \emptyset$.
 By \refToLemItem{alg-read-inf}{4} $\algreadinf{}{\widehat\G}{\Msg_2}$ implies $\algread{\widehat\G}{\Msg_2}$.
 From $\widehat\G = \End$ we get  
 $\Msg_2 = \emptyset$, and so the statement is immediate by Rule \rn{B-End}.  
\item [\rn{\ib{Cycle$'$}}] 
We have $\Hset = \Hset_1,(\widehat\G,\Msg'),\Hset_2$ and $\Msg_1 \equiv \addMsg{\Msg'}{\Msg''}$ and $\algcomp{}{\widehat\G}{\Msg''}$ and $\algreadinf{}{\widehat\G}{\Msg''}$, and, 
since $\Hset$ is coherent, we have $\algwf{\Hset_1}{\widehat\G}{\Msg'}$ and $\Hset_1$ is coherent as well. 
Note that, since $\equiv_{\widehat\G}$ extends $\equiv$, we have 
$\widehat\Msg \equiv_{\widehat\G} \addMsg{\Msg_1}{\Msg_2} \equiv_{\widehat\G}  \addMsg{\Msg'}{\addMsg{\Msg''}{\Msg_2}}$ and, 
by \refToLemItem{sound-aux}{1} and \refToLemItem{sound-aux}{2}, we have $\algcomp{}{\widehat\G}{\addMsg{\Msg''}{\Msg_2}}$ and $\algreadinf{}{\widehat\G}{\addMsg{\Msg''}{\Msg_2}}$. Then the statement follows by induction hypothesis, as $\Hset_1$ is strictly shorter than $\Hset$. 
\item [\rn{\ib{Out}}] 
We have $\widehat\G = \agtO{\pp}{\q}{i}{I}{\la}{\G}$,  $\algread{\widehat\G}{\Msg_1}$    and $\algwf{\Hset,(\widehat\G,\Msg_1)}{\G_i}{\addMsg{\Msg_1}{\mq\pp{\la_i}\q}}$  for all $i\in I$. 
 We need to show that $\algread{\widehat\G}{\widehat\Msg}$ and $\G_i\parG\addMsg{\widehat\Msg}{\mq\pp{\la_i}\q}$ belongs to $\AS$ for all $i\in I$, then the statement follows by Rule \rn{B-Out}. 
Note that $\algread{\widehat\G}{\widehat\Msg}$ follows from $\widehat\Msg\equiv_{\widehat\G}\addMsg{\Msg_1}{\Msg_2}$ and $\algread{\widehat\G}{\Msg_1}$ and $\algreadinf{}{\widehat\G}{\Msg_2}$ by \refToLemItem{read-merge}{1} and \refToPropItem{equivG}{1}. 
We now prove that $\G_i\parG\addMsg{\widehat\Msg}{\mq\pp{\la_i}\q}$ belongs to $\AS$ for all $i\in I$.    
The coherence of $\Hset$ implies the coherence of $\Hset,(\widehat\G,\Msg_1)$. From $\algcomp{}{\widehat\G}{\Msg_2}$ we get, for all $i\in I$, $\algcomp{}{\G_i}{\Msg'_i}$  for some $\Msg_i'$ and  $\addMsg{\Msg_2}{\mq\pp{\la_i}\q} \equiv_{\G_i} \addMsg{\mq\pp{\la_i}\q}{\Msg'_i}$ by \refToLemItem{comp-inv}{2}. From $\algreadinf{}{\widehat\G}{\Msg_2}$ we get $\algreadinf{}{\G_i}{\Msg_2}$ by  \refToLemItem{alg-read-inf}{3} for all $i\in I$. \refToLem{read-rev} implies $\algreadinf{}{\G_i}{\Msg'_i}$
for all $i\in I$. 
From $\widehat\Msg \equiv_{\widehat\G} \addMsg{\Msg_1}{\Msg_2}$ we have 
$\addMsg{\widehat\Msg}{\mq\pp{\la_i}\q} \equiv_{\G_i} \addMsg{\addMsg{\Msg_1}{\mq\pp{\la_i}\q}}{\Msg'_i}$ for all $i\in I$.
 
%
\item [\rn{\ib{In}}] In this case we have $\widehat\G = \agtI{\pp}{\q}{\la}{\G}$, $\algread{\widehat\G}{\Msg_1}$ and
$\Msg_1 \equiv \addMsg{\mq\pp{\la_h}\q}{\Msg'}$ and \mbox{$\algwf{\Hset,(\widehat\G,\Msg_1)}{\G_h}{\Msg'}$} for some $h\in I$. We need to show that $\algread{\widehat\G}{\widehat\Msg}$ and $\G_h\parG\Msg'\cdot\Msg_2$ belongs to $\AS$, then the statement follows by Rule \rn{B-In}. 
Note that $\algread{\widehat\G}{\widehat\Msg}$ follows from $\widehat\Msg\equiv_{\widehat\G}\addMsg{\Msg_1}{\Msg_2}$ and $\algread{\widehat\G}{\Msg_1}$ and $\algreadinf{}{\widehat\G}{\Msg_2}$ by \refToLemItem{read-merge}{1} and \refToPropItem{equivG}{1}. 
We now prove that $\G_h\parG\Msg'\cdot\Msg_2$ belongs to $\AS$.
The coherence of $\Hset$ implies the coherence of $\Hset,(\widehat\G,\Msg_1)$.  From $\algcomp{}{\widehat\G}{\Msg_2}$ we get   $\algcomp{}{\G_h}{\Msg_2}$ by \refToLemItem{comp-inv}{1}. From $\algreadinf{}{\widehat\G}{\Msg_2}$ we get $\algreadinf{}{\G_h}{\Msg_2}$ by \refToLemItem{alg-read-inf}{3}. 
\qedhere
\end{description}
\end{proof}

\begin{figure}
\prooftree
\prooftree
\prooftree
\prooftree
\Hset,(\s\cl?\ok;\G,\mq\s\ok\cl)\OKA{\G}\emptyset
\justifies
\Hset\OKA{\s\cl?\ok;\G}{\mq\s\ok\cl}
\endprooftree
\hspace*{-2pt} 
\prooftree
\prooftree
\prooftree 
\algcomp{}{\G}{\mq\cl\lql\s}\quad \algreadinf{}{\G}{\mq\cl\lql\s}\quad\algread{\G}{\emptyset} 
\justifies
\Hset,(\s\cl?\ko;\cl\s!\lql;\G,\mq\s\ko\cl),(\cl\s!\lql;\G,\emptyset)\OKA{\G}{\mq\cl\lql\s}
\endprooftree
\justifies
\Hset,(\s\cl?\ko;\cl\s!\lql;\G,\mq\s\ko\cl)\OKA{\cl\s!\lql;\G}{\emptyset}
\endprooftree
\justifies
\Hset\OKA{\s\cl?\ko;\cl\s!\lql;\G}{\mq\s\ko\cl}
\endprooftree
\justifies
(\G,\emptyset), (\G_1,\mq\cl\hq\s)\OKA{\G_2}{\emptyset}
\endprooftree
\justifies
(\G,\emptyset)\OKA{\G_1}{\mq\cl\hq\s}
\endprooftree
\justifies
\OKA{\G}{\emptyset}
\endprooftree

\bigskip

\raggedright
\qquad where $\Hset=\set{(\G,\emptyset), (\G_1,\mq\cl\hq\s),(\G_2,\emptyset)}$
\caption{Inductive balancing of the hospital global type with the empty queue.}\label{fig:hb}
\end{figure}
\begin{figure}
\prooftree
\prooftree 
\prooftree 
\prooftree
\algcomp{\Hset',(\s\cl?\ok;\G,\Msg)}{\G}{\Msg}
\justifies
\algcomp{\Hset'}{\s\cl?\ok;\G}{\Msg}
\endprooftree
\quad
\prooftree
\prooftree
\algcomp{\Hset',(\s\cl?\ko;\cl\s!\lql;\G,\Msg),(\cl\s!\lql;\G,\Msg)}{\G}{\Msg}
\justifies
\algcomp{\Hset',(\s\cl?\ko;\cl\s!\lql;\G,\Msg)}{\cl\s!\lql;\G}{\Msg}
\endprooftree
\justifies
\algcomp{\Hset'}{\s\cl?\ko;\cl\s!\lql;\G}{\Msg}
\endprooftree
\justifies
\algcomp{(\G,\Msg),(\G_1,\Msg)}{\G_2}{\Msg}
\endprooftree
\justifies
\algcomp{(\G,\Msg)}{\G_1}{\Msg}
\endprooftree
\justifies
\algcomp{}{\G}{\Msg}
\using \quad  \Msg\cdot\mq\cl\hq\s\equiv_{\G_1}\mq\cl\hq\s\cdot\Msg 
\endprooftree

\bigskip

where $\Msg=\mq\cl\lql\s$ and   $\Hset'=\set{(\G,\Msg),(\G_1,\Msg),(\G_2,\Msg)}$ 

\caption{Agreement of the hospital global type with the queue $\mq\cl\lql\s$.}\label{fig:ha}
\end{figure}

\begin{figure}
\prooftree
\prooftree 
\prooftree 
\prooftree
\prooftree
\algread{\G}{\Msg}
\justifies
 \algreadinf{ \G,\G_1,\G_2,\s\cl?\ok;\G}{\G}{\Msg}
 \endprooftree
\justifies
 \algreadinf{ \G,\G_1,\G_2}{\s\cl?\ok;\G}{\Msg}
\endprooftree
\quad
\prooftree
\prooftree
\prooftree
\algread{\G}{\Msg}
\justifies
 \algreadinf{ \G,\G_1,\G_2,\s\cl?\ko;\cl\s!\lql;\G,\cl\s!\lql;\G}{\G}{\Msg}
 \endprooftree
\justifies
 \algreadinf{ \G,\G_1,\G_2,\s\cl?\ko;\cl\s!\lql;\G}{\cl\s!\lql;\G}{\Msg}
\endprooftree
\justifies
 \algreadinf{ \G,\G_1,\G_2}{\s\cl?\ko;\cl\s!\lql;\G}{\Msg}
\endprooftree
\justifies
 \algreadinf{\G,\G_1}{\G_2}{\Msg}
\endprooftree
\justifies
 \algreadinf{\G}{\G_1}{\Msg}
\endprooftree
\justifies
 \algreadinf{}{\G}{\Msg}
\endprooftree

\bigskip

\raggedright
\qquad \qquad where $\Msg=\mq\cl\lql\s$
\caption{Deep read of the hospital global type with the queue $\mq\cl\lql\s$.}\label{fig:hd}
\end{figure}

\begin{figure}
\prooftree
\prooftree
\algread{\G_2}{\emptyset}
\justifies
\algread{\G_1}{\mq\cl\lql\s}
\endprooftree
\justifies
\algread{\G}{\mq\cl\lql\s}
\endprooftree
\caption{Read of the hospital global type with the queue $\mq\cl\lql\s$.}\label{fig:hr}
\end{figure}

 Finally, we observe that we can obtain a sound inductive version of weak balancing as well. 
We just have to remove premises involving readability and deep readability from the rules in \refToFig{wfConfigA}.  

In Figures~\ref{fig:hb}, ~\ref{fig:ha},~\ref{fig:hd} and~\ref{fig:hr} we prove that the hospital global type with the empty queue is inductively balanced using the agreement, deep read and read judgments. 
Notice that in \refToFigure{fig:ha} we need $\mq\cl\lql\s\cdot\mq\cl\hq\s\equiv_{\G_1}\mq\cl\hq\s\cdot\mq\cl\lql\s$. 
This holds since the messages $\mq\cl\lql\s$ and $\mq\cl\hq\s$ are $\G_1$-indistinguishable, being $\G_1=\agtIS{\pp}{\ps}{\set{\Seq{\nD}{\G_2}, \Seq{\pR}{\G_2}}}$, where $\G_2$ does not contain different inputs from $\pp$ to $\s$.


\section{Related and future works}\label{sec:rfw}

Since their introduction in~\cite{CHY08} multiparty sessions have been equipped with global types. The literature on multiparty sessions is very large, see~\cite{H2016} for a survey. In almost all papers the building blocks of global types are synchronous communications between senders and receivers. This does not fit well with asynchronous communications. To fill this gap a suitable subtyping is proposed in~\cite{HMY09}. This subtyping enjoys completeness~\cite{GPPSY21}, i.e.  
any extension of it would be unsound. Unfortunately the obtained type system is not effective, since this subtyping is undecidable~\cite{BCZ17,LY17}. Decidable restrictions of this subtyping relation have been proposed~\cite{BCZ17,LY17,BCZ18}.  In particular, subtyping is decidable when both internal and external
choices are forbidden in one of the two compared processes~\cite{BCZ17}.  This result is improved in~\cite{BCZ18},
where both the subtype and the supertype can contain either internal or external choices.  More interestingly, the work~\cite{BCLYZ21}
presents a sound (though not complete) algorithm for checking asynchronous subtyping  for binary sessions.  It is worthwhile to notice that  Example 3.21 of~\cite{BCLYZ21} is  typable in our system  with a provably balanced global type. Instead, in standard type systems it requires a subtyping which is not derivable by the algorithm in~\cite{BCLYZ21}. 
In our notation  this example  
is $\pP\pp\PP\parN\pP\q\Q$ where $\PP=\q?\{\la_1;\PP_1,\la_2;\PP_2\}$, $\PP_1=\q!\la;\q!\la;\q!\la;\PP$, $\PP_2=\q!\la;\PP_2$ and $\Q=\pp!\{\la_1;\Q_1,\la_2;\Q_2\}$, $\Q_1=\pp?\la;\Q$, $\Q_2=\pp?\la;\Q_2$. The reduction of this network has a loop in which for each cycle the number of messages $\mq\pp{\la}\q$ in the queue increases by two. Our type inference algorithm gives between others the global type  
$\G=\agtSOS{\q}{\pp}{  \{\Seq{\la_1}{\Seq{\CommAsI{\q}{\la_1}{\pp}}{\G_1} }\, ,\Seq{\la_2}{\Seq{\CommAsI{\q}{\la_2}{\pp}}{{\G_2} }}\} }$, where 
$\G_1={\Seq{\CommAs{\pp}{\la}{\q}}{ \Seq{\CommAs{\pp}{\la}{\q}}{  \Seq{\CommAs{\pp}{\la}{\q}}{ \Seq{\CommAsI{\pp}{\la}{\q}}{  \G }   } }  }   }$ and 
$\G_2=\Seq{\CommAs{\pp}{\la}{\q}}{\Seq{\CommAsI{\pp}{\la}{\q}}{  \G_2 }  }$. 

To the best of our knowledge global types with split outputs and inputs are considered only in~\cite{CDG20,DGD21,CDG22}. 
The above network can be typed in~\cite{DGD21} and in~\cite{CDG22}, but not in~\cite{CDG20}. 
In~\cite{CDG20,DGD21} there is no choice of inputs and the running example  of the current paper  cannot be typed. 
Instead~\cite{CDG22} further extends the syntax of global types by allowing multiple senders in internal choices and multiple receivers in external choices. 
The obtained typing is undecidable and it does not ensure the absence of orphan messages. 
The present work springs from~\cite{DGD21}, but with many improvements:
\begin{itemize}
\item there are choices of inputs in global types;
\item networks are typed by global types, while in all other papers, but~\cite{CDG22}, either global types are projected onto local types and local types are assigned to processes or 
global types are projected onto processes;
\item the proof that the type system allows us to type an arbitrary network is new; 
\item the type inference algorithm is new and very simple compared for example to the algorithm in~\cite{GHH21}, essentially thanks to the flexible syntax of global types; 
\item the conditions for taming types are more permissive: this also is due to the new syntax of global types; 
\item  the proof of undecidability for the (weak) balancing predicate is original also if inspired by the proof in~\cite{BCZ17};
\item the decision algorithm for balancing  is empowered by  the use of the equivalence on queues parametrised on global types.
\end{itemize}
The  type system of~\cite{DGD21} is implemented  in co-logic programming, see~\cite{BD21}. The tool  is available at \cite{Impl21}. We plan to extend this implementation to the present type system.
   
As future work we want to compare typability in our type system with  typability  
in the standard type system enriched by asynchronous subtyping.

 \paragraph{Acknowledgment} 
We are grateful to Ilaria Castellani and Elena Zucca for enlightening discussions on the subject of this paper. 
We are strongly indebted to the anonymous referees since thanks to their constructive remarks the current version of the paper  greatly improves the  submitted one. 
 \bibliographystyle{alphaurl}
 \bibliography{agtbib}

\end{document}